\documentclass[draftclsnofoot, onecolumn, 12pt]{IEEEtran}

\usepackage{floatflt}

\usepackage{color}
\usepackage{cite}

\usepackage{graphicx}
\usepackage{amsmath}
\usepackage{amsthm}
\usepackage{amssymb}

\usepackage{verbatim}

\usepackage{psfrag,array}

\usepackage{subfigure}
\usepackage{multirow}
\usepackage{hhline}

\usepackage{stfloats}

\usepackage{amsmath}
\usepackage{epstopdf}
\usepackage{algorithmic}

\newcommand{\p}[1]{\mathop{\mbox{\it p} } }

\renewcommand{\vec}[1]{\ensuremath{\boldsymbol{#1}}}
\newcommand{\be}{\begin{equation}}
\newcommand{\ee}{\end{equation}}
\newcommand{\ba}{\begin{array}}
\newcommand{\ea}{\end{array}}
\newcommand{\bea}{\begin{eqnarray}}
\newcommand{\eea}{\end{eqnarray}}
\newcommand{\bean}{\begin{eqnarray*}}
\newcommand{\eean}{\end{eqnarray*}}

\newcommand{\rmh}{^{\dag}}
\newcommand{\rmih}{^{-\dag}}
\newcommand{\rmt}{^{\rm T}}
\newcommand{\rmit}{^{-\rm T}}

\renewcommand{\Re}{\mathcal{ R}}

\newcommand{\nR}{\nu_{\mathrm{R}}}

\definecolor{white}{rgb}{1,1,1}

\newtheorem{theorem}{Theorem}
\newtheorem{lemma}{Lemma}
\newtheorem{example}{Example}
\newtheorem{proposition}{Proposition}

\newtheorem{corollary}{Corollary}
\newtheorem{remark}{Remark}

\begin{document}

\title{On the Design of Channel Shortening Demodulators for Iterative Receivers in Linear Vector Channels}

\author
{ \begin{tabular}{c}
 Sha Hu and Fredrik Rusek\\
% Department of Electrical and Information Technology $\quad\quad$\\
% Lund University, Lund, Sweden $\quad\quad$\\
% \{firstname.lastname\}@eit.lth.se $\quad\quad$
\end{tabular}
\thanks{The authors are with the Department of Electrical and Information Technology, Lund University, Lund, Sweden (email: firstname.lastname@eit.lth.se).} 

\thanks{This paper has been presented in part \cite{HR15} at the IEEE 26th annual international symposium on personal, indoor and mobile radio communications (PIMRC), Hongkong, Sep. 2015.} 
}

\maketitle

\vspace{-4mm}

\begin{abstract}
We consider the problem of designing demodulators for linear vector channels with memory that use reduced-size trellis descriptions for the received signal. We assume an overall iterative receiver, and use interference cancellation (IC) based on the soft information provided by the outer decoder, to mitigate the parts of the signal that are not covered by the reduced-size trellis description. In order to reach a trellis description, a linear filter is applied as front-end to compress the signal structure into a small trellis. This process requires three parameters to be designed: (i) the front-end filter, (ii) the feedback filter through which the IC is done, and (iii) a target response which specifies the trellis. Demodulators of this form have been studied before under then name \textit{channel shortening} (CS), but the interplay between CS, IC and the trellis-search process has not been adequately addressed in the literature. In this paper, we analyze two types of CS demodulators that are based on the Forney and Ungerboeck detection models, respectively. The parameters are jointly optimized based on a generalized mutual information (GMI) function. We also introduce a third type of CS demodulator that is in general suboptimal, but has closed-form solutions. Moreover, signal to noise ratio (SNR) asymptotic properties are analyzed and we show that the third CS demodulator asymptotically converges to the optimal CS demodulator in the sense of GMI-maximization.
\end{abstract}


\section{Introduction}

Channel shortening (CS) demodulators have a long and rich history, see \cite{clve1,clve2,clve3,clve4,clve5,clve6,clve7,AL00,clve10,clve11,clve8,clve9, DF07}. For intersymbol interference (ISI) channels, Forney \cite{F72} showed that the Viterbi Algorithm (VA) \cite{VA} implements maximum likelihood (ML) detection. However, the complexity of the VA is exponential in the memory of the channel which prohibits its use in many cases of interest. As a remedy, Falconer and Magee proposed in 1973 the concept of CS \cite{clve1}. The concept is to filter the received signal with a prefilter so that the effective channel has much shorter duration than the original channel,  and then apply the VA to the shorter effective channel.

Traditionally, CS demodulators have been optimized from a minimum mean square error (MMSE) perspective \cite{clve2,clve3,clve4,clve5,clve6,clve7,clve9,clve10,clve11}. Two exceptions from this are the papers \cite{clve8} and \cite{AL00}. In \cite{clve8}, the authors attempt to minimize the error probability of an uncoded system which leads to a new notion of posterior equivalence between the target response and the filtered channel. However, since \cite{clve8} works with uncoded error probabilities, the analysis in \cite{clve8} does not adequately address the case of coded systems and Shannon capacity properties. The first paper that works with capacity-related cost measures is \cite{AL00}. In \cite{AL00} the authors consider the achievable rate, in the form of generalized mutual information (GMI) \cite{mism,mism2,mct10, RP12,WSS04}, that the transceiver system can achieve if a CS demodulator is adopted. However, \cite{AL00} is limited to ISI channels only, and the design method in \cite{AL00} of the CS demodulator is in fact not always possible to execute. The limitations of \cite{AL00} were first dealt with in \cite{RP12}, which extended the CS concept to any linear vector channel and resulted in a closed-form optimization procedure.

Iterative receivers such as turbo equalization \cite{TS11, LSS05, HR15, SZM08, GLL97, LB06} followed as a natural extension to turbo codes as an iterative technique for detection and decoding of forward error correction (FEC) protected data that is transmitted over dispersive channel. However, when it comes to turbo equalization, common settings of the equalizer are \cite{TS11} the maximum \textit{a posterior} (MAP) demodulator \cite{MAP} and its suboptimal variants such as dimension-reduction and subspace based detections \cite{CBSC10, CLS15}, and MMSE based approaches \cite{LSS05, GLL97, LB06, CSWC10} that replace the MAP demodulator with a linear equalizer or a decision feedback equalizer (DFE) to reduce the prohibitive complexity of the MAP demodulator. One important open problem in the area of turbo equalization is the development of other non-trellis-based detection methods that provide performance between that of MAP and MMSE performance \cite{TSK00, TS11}. Instead of fully removing the trellis-based detection, another possible approach is to reduce the memory size of the original linear vector channel through an interference cancellation (IC) based prefiltering. To the best of our knowledge, there is only limited literature \cite{SZM08, HKHR17} on such a design of demodulator that combines both IC based prefiltering and a memory-size shortened BCJR in iterative receiver design. A closely related concept is delayed-decision-feedback-sequence-estimation (DDFSE) \cite{DHH89, BVT00}, which also reduces the number of states in the BCJR. However, in DDFSE the IC is done within a single iteration, and not between the iterations of an iterative receiver.

In this paper, we generalize the idea in \cite{RP12} of GMI-maximization based CS demodulators to iterative receivers. With iterative receivers it is reasonable to expect that better performance can be reached by allowing the parameters of the CS demodulator to change in each iteration. The CS demodulator in \cite{RP12} does not take the prior information into account, rendering its design static in all iterations. We aim at constructing a CS demodulator that takes soft information provided by the outer decoder into account so that the parameters of the CS demodulator are designed for a particular level of prior knowledge. This procedure includes an IC mechanism to deal with the signal part that can not be handled by the trellis-search. Preliminary results for CS demodulators in iterative receivers are available in \cite{GC21}, but this paper non-trivially advances the state-of-the-art. 

Although the trellis-search based detection is still utilized in the CS demodulator, the memory size $\nu$ of the linear vector channel has been reduced which results in significant complexity reduction compared to the MAP demodulator. Meanwhile, with different values of $\nu$, the CS demodulator provides trade-off between the performance of MMSE and MAP. As will become clear later, the CS demodulator is closely related to the concept of linear MMSE receiver with parallel interference cancellation (LMMSE-PIC)\cite{SFS11, MS02, Z05}, which cooperates the soft information into the filter coefficients and interference cancellation process. With setting $\nu\!=\!0$, the CS demodulator is identical to the LMMSE-PIC demodulator whose trellis-search process is trivial since different symbols are assumed to be independent after the front-end filtering. The CS demodulator can also be viewed as an extension of the LMMSE-PIC to include a trellis-search, where the parameters of the front-end filter, IC, and trellis-search are jointly optimized. On the other hand, by setting $\nu$ to be equal to the original memory size of the linear vector channel, the CS demodulator is identical to MAP. Therefore, the CS demodulator is a generalized framework that includes both the MAP and LMMSE-PIC in iterative receiver design.

The rest of the paper is organized as follows: The linear vector channel model and the iterative receiver structure are introduced in Sec.\,II, while the general form of the CS demodulators and the GMI are described in Sec.\,III. In Sec.\,IV we analyze three types of CS demodulators for finite length linear vector channels. In Sec.\,V we deal with ISI channels as asymptotic versions of the results established in Sec.\,IV. The signal to noise ratio (SNR) asymptotic of the CS demodulators are discussed in Sec.\,VI. Empirical results are provided in Sec.\,VII, and Sec.\,VIII summarizes the paper. For improved readability, we have deferred some long proofs and derivations to Appendices A-K.

\subsubsection*{Notation}

Throughout the paper, a capital bold letter such as $\vec{A}$ represents a matrix, a lower case bold  letter $\vec{a}$ represents a vector, and a capital letter $A$ represents a number. The expression  $\vec{A}\!\prec\!0$ means matrix $\vec{A}$ is negative definite, while $\vec{A}\!\succ\!0$ means $\vec{A}$ is positive definite. Matrix $\vec{I}$ represents the identity matrix and in general the dimension will be omitted; when it cannot be understood from the context, we let $\vec{I}_{K}$ represent a $K\!\times\!K$ identity matrix. Our superscripts have the following meanings: $(\,)^\ast$ is complex conjugate, $(\,)^\mathrm{T}$ is matrix transpose, $(\,)^\dag$ denotes the conjugate transpose of a matrix, $(\,)^{-1}$ is matrix inverse. In addition, $\propto$ means proportional to, $\mathbb{E[\,]}$ is the expectation operator, $\mathrm{Tr(\,)}$ takes the trace of a matrix, $\Re\{\,\}$ returns the real part of a variable, $\otimes$ is the Kronecker multiplication operator, $\mathrm{vec}(\vec{A})$ is a column vector containing the columns of matrix $\vec{A}$ stacked on top of each other, and $[A, B]$ is the set of integers $\{k\!:\!A\!\leq\!k\!\leq\!B\}$. Furthermore, we say that a matrix $\vec{A}$ is banded within diagonals $[-\nu_1,\nu_2]$ ($\nu_1, \nu_2\!\geq\!0$), if the $(k,\ell)$th element $A(k,\ell)$ satisfies\footnote{Note that $\nu_1$ refers to the number of upper diagonals of $\vec{A}$ that are nonzero. We have this convention in order to subsequently follow standard notation for Toeplitz matrices \cite{MK}.}
$$A(k,\ell)=0, \;  \ell-k>\nu_1\;\,\mathrm{or} \,\;k-\ell>\nu_2.$$
Moreover, we define two matrix operators $[\;]_{\nu}$ and $[\;]_{\backslash\nu}$ such that $\vec{A}\!=\![\vec{A}]_{\nu}\!+\![\vec{A}]_{\backslash\nu}$, with $[\vec{A}]_{\nu}$ banded within diagonals $[-\nu,\nu]$ where $[\vec{A}]_{\backslash\nu}$ is constrained to zero.

\section{System Model}

We consider linear vector channels according to
\bea \label{p1sm1} \vec{y} = \vec{H}\vec{x}\!+\!\vec{n}\eea
where $\vec{y}$ is an $N\!\times\! 1$ vector of received signal, $\vec{x}$ is a $K\!\times\!1$ vector comprising unit energy coded symbols that belong to a constellation $\mathcal{X}$, $\vec{H}$ is an $N\!\times\!K$ matrix representing the communication channel which is perfectly known to the receiver and $\vec{n}$ is zero-mean complex Gaussian noise vector with covariance matrix $N_0\vec{I}$. Model (\ref{p1sm1}) may represent many different communication systems, such as for example multi-input multi-output (MIMO) or ISI channels. In the MIMO case, the variables $N$ and $K$ are finite while they grow without bounds in the ISI case. For the former case, a block fading model is assumed, where the coherence time is infinite. The block fading model allows us to perform an analysis for a single symbol period.

Denote $x_k$ as the $k$th element of $\vec{x}$ and $\vec{h}_k$ as the $k$th column vector of $\vec{H}$, (\ref{p1sm1}) can be rewritten as
\bea  \vec{y}=\sum_{k=0}^{K\!-\!1}\vec{h}_{k}x_k\!+\!\vec{n}.\eea

In an iterative receiver,  the feedback from the outer decoder can be utilized in the demodulator to improve the performance. As the outer decoder provides the demodulator with \textit{a posteriori} probability (APP) and extrinsic information (in terms of bit log-likelihood ratios (LLRs)) \cite{H95,S01}, side information is present about the symbols $\vec{x}$ and we represent this by the probability mass function $p_{k}(s) \!=\!\mathrm{P}(x_k\!=\!s), (0\!\leq \! k\!\leq\! K\!-\!1)$. Note that the side-information does not consider the dependency among the symbols, but are symbol-wise marginal probabilities. This reflects the situation encountered in iterative receivers with perfect interleaving. In those cases, the prior probabilities provided from previous iterations are assumed independent, i.e., $\mathrm{P}(\vec{x}\!=\!\vec{s})\!=\!\prod p_{k}(s)$. Due to the perfect interleaving assumption, the demodulator can compute $\hat{\vec{x}}\!=\!\mathbb{E}_{p(\vec{x})}[\vec{x}]\!=\![\!\begin{array}{ccc}\hat{x}_0,\hat{x}_1,\cdots,\hat{x}_{K-1}\end{array}\!]\rmt$ in a per-entry fashion as
\bea \hat{x}_k=\sum_{s\in\mathcal{X}}sp_{k}(s),\notag \eea
where the expectations are computed with respect to the prior distribution $p_{k}(s)$. 

With soft information $\hat{\vec{x}}$, we define a $K\!\times\!K$ diagonal matrix $\vec{P}$ as follows. For finite length linear vector channels, $\vec{P}$  equals
\bea \label{p1Pmat1} \vec{P}=\mathbb{E}_{\mathrm{T}}\big[\mathbb{E}_{\mathrm{p}(\vec{x})}[\vec{x}\hat{\vec{x}}\rmh]\big]=\mathbb{E}_{\mathrm{T}}[\hat{\vec{x}}\hat{\vec{x}}\rmh], \eea
where the exception \lq\lq{}$\mathbb{E}_{\mathrm{T}}$\rq\rq{} is taken over the transmitted blocks of $\vec{x}$ under the block fading assumption. For ISI case, as the whole data block experiences the same channel, we let 
\bea \vec{P}\!=\!\alpha\vec{I}, \eea
where the scalar
\bea  \alpha=\frac{1}{K}\sum_{k=0}^{K-1}|\hat{x}_k|^2.\eea 
The variable $\vec{P}$ in (\ref{p1Pmat1}) can alternatively be written as
\bea \vec{P}=\mathbb{E}_{\mathrm{T}}[\hat{\vec{x}}\hat{\vec{x}}\rmh]= \mathbb{E}_{\mathrm{T}}\big[\mathbb{E}_{\mathrm{p(\vec{x})}}[\vec{x}\vec{x}\rmh]\big]- \mathbb{E}_{\mathrm{T}}\big[\mathrm{cov}(\vec{x})\big].\eea
Under the natural assumption of soft information that satisfies
\bea \mathbb{E}_{\mathrm{T}}\big[\mathbb{E}_{\mathrm{p(\vec{x})}}[\vec{x}\vec{x}\rmh]\big]=\vec{I}, \eea
it follows that $\vec{0}\!\preceq\!\vec{P}\!\preceq\!\vec{I}$, and the same also holds for ISI case. The variable $\vec{P}$ reflects the accuracy of the side information. That is, when there is no soft information available, we have $\vec{P}\!=\!\vec{0}$, while with perfect feedback we get $\vec{P}\!=\!\vec{I}$.

The task of the demodulator is to generate soft information about the symbols in $\vec{x}$ given the observable $\vec{y}$ and the side information $\{p_{k}(s)\}$. The optimal demodulator is the MAP demodulator \cite{MAP, HB03} which evaluates the posterior probabilities
$\mathrm{P}(x_k\!=\!s|\vec{y})$. However, the number of leaves of the search tree corresponding to the MAP demodulator is in general $|\mathcal{X}|^{K}$ which is prohibitive for most practical applications. The purpose of the CS demodulator is to force the signal model to be an lower triangular matrix with only $\nu\!+\!1$ $(0\!\leq\!\nu\!<\!K\!-\!1)$ nonzero diagonals by means of a linear filter\footnote{For finite length linear vector channels such as MIMO channel, ``filtering'' means matrix multiplication.}, where $\nu$ is referred to as the memory size of the CS demodulator. Then, a BCJR \cite{BCJR} demodulator can be applied over a trellis with $|\mathcal{X}|^{\nu}$ states. Moreover, since there is side information present about $\vec{x}$, the parts of $\vec{H}$ that are outside the memory of the BCJR can be partly eliminated by means of IC through the prior mean $\hat{\vec{x}}$.

\begin{figure}
\begin{center}
\vspace*{0mm}
\scalebox{0.72}{\includegraphics{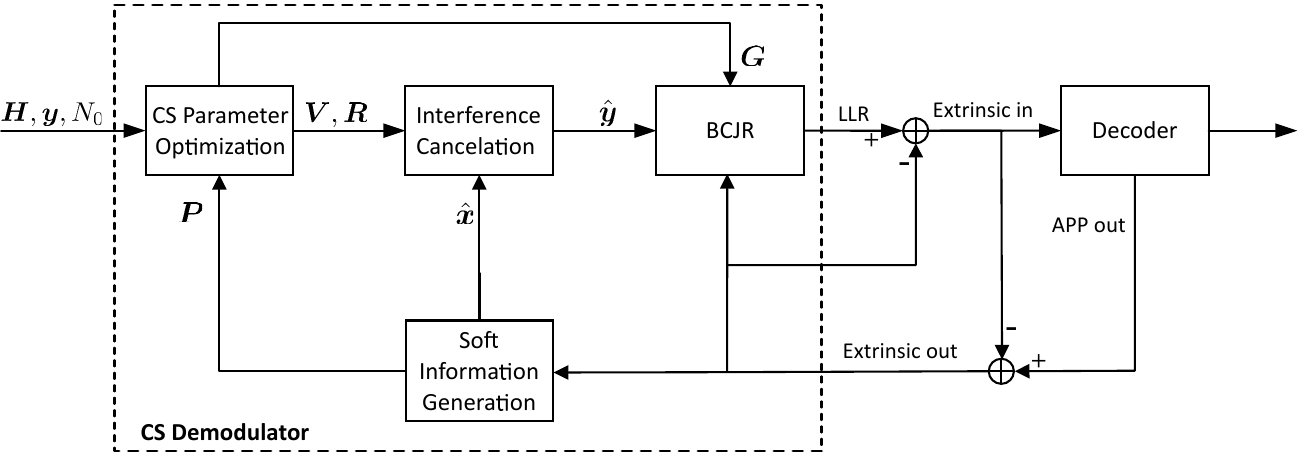}}
\vspace*{-4mm}
\caption{\label{p1fig1} Iterative receiver structure with CS demodulator and outer decoder. The target of the CS demodulator is to maximize the GMI through jointly optimizing the parameters $\vec{V}$, $\vec{R}$ and $\vec{G}$, which are referred to as the front-end filter, IC matrix and trellis representation matrix, respectively.} 
\vspace*{-10mm}
\end{center}
\end{figure}

The structure of an iterative receiver utilizing a CS demodulator is depicted in Fig.\,\ref{p1fig1}. The extrinsic information from the outer decoder is used to compute an estimate $\vec{\hat{x}}$ and a matrix $\vec{P}$ that indicates the feedback quality. Based on the updated $\vec{P}$ in each iteration, the optimal CS parameters are found by maximizing the GMI. A prefiltering and IC process are then implemented on $\vec{y}$ with optimal $\vec{V}$ and $\vec{R}$ to obtain the signal $\vec{\hat{y}}$, which is sent to a memory $\nu$ BCJR module specified by an optimal $\vec{G}$. Moreover, the extrinsic information iteratively exchanged between the BCJR and the outer decoder is also used as \textit{a priori} information for the transmitted symbols. Note that if we set $\nu\!=\!K\!-\!1$, the search space of the CS demodulator is no longer a trellis but corresponds to the original tree and is therefore equivalent to MAP, and LMMSE-PIC is a special case of the CS demodulation with $\nu\!=\!0$.

\section{The General Form of the CS Demodulator}

We state two lemmas that will be useful later, and Lemma \ref{lem2} can be verified straightforwardly.
\begin{lemma}{\label{lem1}}
Let $\vec{A}_1$ and $\vec{A}_2$ be two $K\!\times\!K$ matrices, where $\vec{A}_1$ is invertible and banded within diagonals $[-\nu, \nu]$. If $[\vec{A}_1^{-1}]_{\nu}\!=\![\vec{A}_2]_{\nu}$, then
\bea \mathrm{Tr}\big(\vec{A}_1\vec{A}_2\big)\!=\!\mathrm{Tr}\big(\vec{I}\big). \notag \eea
\end{lemma}
\begin{proof}
Let $\vec{A}_3\!=\!\vec{A}_2\!-\!\vec{A}_1^{-1}$, then $[\vec{A}_3]_{\nu}\!=\!\vec{0}$ and $\vec{A}_3\!=\![\vec{A}_3]_{\backslash\nu}$. As $\vec{A}_1\!=\![\vec{A}_1]_{\nu}$, the elements along the main diagonal of $\vec{A}_1\vec{A}_3$ are zero. Therefore $\mathrm{Tr}\big(\vec{A}_1\vec{A}_2\big)\!=\!\mathrm{Tr}\big(\vec{A}_1(\vec{A}_1^{-1}\!+\!\vec{A}_3)\big)\!=\!\mathrm{Tr}(\vec{I}).$
\end{proof}

\begin{lemma} \label{lem2}
Let $\vec{A}_1$ and $\vec{A}_2$ be two $K\!\times\!K$ matrices that are banded within diagonals $[-\nu_1, \nu_2]$ and $[-\nu_3, \nu_4]$, respectively. Then the product $\vec{A}_1\vec{A}_2$ is banded within diagonals $[\max(-(\nu_1\!+\!\nu_3),1-K), \min(\nu_2\!+\!\nu_4,K-1)]$.
\end{lemma}

\subsection{System Model of the CS Demodulator}
The CS demodulators that we investigate operate on the basis of the mismatched\footnote{By \lq\lq{}mismatched\rq\rq{} we mean that $\tilde{p}(\vec{y}|\vec{x})$ may not be a valid probability distribution function and in general differs from the true conditional probability distribution function $p(\vec{y}|\vec{x})$ even with $\hat{\vec{x}}\!=\!\vec{0}$, but such a \lq\lq{}mismatched\rq\rq{} property is for the purpose of reducing the size of trellis description in the BCJR.} function
\bea \label{p1md2} \tilde{p}(\vec{y}|\vec{x})= \exp\!\big(2\Re\{\vec{x}\rmh(\vec{V}\vec{y}\!-\!\vec{R}\hat{\vec{x}})\}\!-\!\vec{x}\rmh\vec{G}\vec{x}\big)\eea
instead of the true conditional probability
\bea \label{p1md0} p(\vec{y}|\vec{x}) = \frac{1}{(\pi N_0)^N}\exp\!\left(\!-\frac{\|\vec{y}\!-\!\vec{H}\vec{x}\|^2}{N_0}\right).\eea
The matrices $\vec{V}$, $\vec{R}$ and $\vec{G}$ are the front-end filter, IC matrix, and trellis representation matrix, respectively. Without loss of generality, we have absorbed $N_0$ into $\vec{V}$, $\vec{R}$, and $\vec{G}$. Models (\ref{p1md2}) and (\ref{p1md0}) are equivalent for demodulation if we set $\vec{V}\!=\!\vec{H}\rmh/{N_0}$, $\vec{R}\!=\!\vec{0}$, and $\vec{G}\!=\!\vec{H}\rmh\vec{H}/{N_0}$, in which case the CS demodulator represents the MAP demodulator.

The detection model (\ref{p1md2}) has its roots in Falconer and Magee's paper \cite{clve1} with adding an IC step, where the system model of the demodulator is described as
\bea  \label{p1md1} \tilde{T} (\vec{y}|\vec{x})=\exp\!\big(\!-\!\|\vec{W}\vec{y}-\vec{T}\hat{\vec{x}}-\vec{F}\vec{x}\|^2\big) \eea
By setting $\vec{T}\!=\!\vec{0}$, we obtain the same system model as in \cite{clve1}. If identifying $\vec{V}\!=\!\vec{F}\rmh\vec{W}$, $\vec{R}\!=\!\vec{F}\rmh\vec{T}$, and $\vec{G}\!=\!\vec{F}\rmh\vec{F}$, model (\ref{p1md1}) is equivalent to (\ref{p1md2}) since
{\setlength\arraycolsep{2pt}\bea \tilde{T} (\vec{y}|\vec{x})&\propto&\exp\!\big(2\Re\{\vec{x}\rmh(\vec{F}\rmh\vec{W}\vec{y}-\vec{F}\rmh\vec{T}\hat{\vec{x}})\}\!-\!\vec{x}\rmh\vec{F}\rmh\vec{F}\vec{x}\big)
\notag \\
&=&\exp\!\big(2\Re\{\vec{x}\rmh(\vec{V}\vec{y}-\vec{R}\hat{\vec{x}})\}-\vec{x}\rmh\vec{G}\vec{x}\big). \notag \eea}
\hspace*{-1.4mm}The detection model (\ref{p1md1}) is usually denoted as \lq\lq{}Forney\rq\rq{} model \cite{clve1} due to its Euclidean-distance form, while the more general model (\ref{p1md2}) is called \lq\lq{}Ungerboeck\rq\rq{} model \cite{U74, FMP07, RC15}. An advantage of the Ungerboeck model over the Forney model is that the parameter optimization through GMI-maximization is simpler \cite{RP12}. However, as both models can be viewed as \lq\lq{}natural\rq\rq{} CS demodulators, we shall investigate both in CS demodulator design for iterative receivers.

In order to optimize $(\vec{V},\vec{R},\vec{G})$, we choose to work with the GMI which is an achievable rate for a receiver that operates on the basis of a mismatched version of the channel law. The GMI in nats/channel is defined as
\bea  \label{egmi} I_{\mathrm{GMI}}=-\mathbb{E}_{p(\vec{y})}\left[\log \tilde{p}(\vec{y})\right]\!+\!\mathbb{E}_{p(\vec{y},\vec{x})}\left[\log \tilde{p}(\vec{y}|\vec{x})\right]\eea
where $\tilde{p}(\vec{y})\!=\!(1/\pi^K)\!\int\tilde{p}(\vec{y}|\vec{x})\exp(-\|\vec{x}\|^2)\mathrm{d}\vec{x}$ and the expectation is taken over the true statistics $p(\vec{y})$ and $p(\vec{y},\vec{x})$. Although finite constellations $\mathcal{X}$ are almost always used in practice, they are hard to analyze. In order to obtain a mathematically tractable problem, here we use a zero-mean, unit variance, complex Gaussian constellation for each entry of $\vec{x}$. With Gaussian inputs, the trellis discussed earlier has no proper meaning as the number of states is infinite even for finite $\nu$. However, the Gaussian assumption is only made in order to design the receiver parameters. We first state Theorem 1 which shows the calculation of the GMI for model (\ref{p1md2}).

\begin{theorem}{\label{p1thm1}}
The GMI for the detection model (\ref{p1md2}) equals
{\setlength\arraycolsep{2pt}\bea \label{p1metricMIMO2} &&I_{\mathrm{GMI}}(\vec{V}, \vec{R}, \vec{G})=\log\!\big(\!\det(\vec{I}\!+\!\vec{G})\big)\!-\!\mathrm{Tr}(\vec{G}) \!+\!2\Re\big\{\mathrm{Tr}(\vec{V}\vec{H}\!-\!\vec{R}\vec{P})\big\}\nonumber \\
&&\quad\quad -\mathrm{Tr}\big((\vec{I}\!+\!\vec{G})^{-\!1}\big(\vec{V}(N_0\vec{I}\!+\!\vec{H}\vec{H}\rmh)\vec{V}\rmh\!-\!2\Re\big\{\vec{V}\vec{H}\vec{P}\vec{R}\rmh\big\}\!+\!\vec{R}\vec{P}\vec{R}\rmh\big)\big). \;\;\eea}
\end{theorem}
The proof of Theorem \ref{p1thm1} is given in Appendix A. Here we make the same assumption as in \cite{RP12} that $\vec{I}+\vec{G}$ is positive definite, otherwise the GMI is not well defined. With any parameters $(\vec{V},\vec{R},\vec{G})$, the GMI can be calculated in (\ref{p1metricMIMO2}), although they may not be optimal in the sense GMI-maximization. We illustrate Theorem \ref{p1thm1} with two examples.  

\begin{example}{\label{p1exam1}}
Extended Zero-Forcing filter (EZF). We extend the zero-Forcing filter \cite{CS03} to only partly invert the channel so that a trellis-search is necessary after the EZF front-end filter. In view of the CS demodulator, we can select the parameters in (\ref{p1md2}) as:
 $$\vec{V}\!=\!(\vec{I}\!+\!\vec{G})(\vec{H}\rmh\vec{H})^{-1}\vec{H}\rmh, \; \vec{R}\!=\!\vec{0}, $$
and then optimize (\ref{p1metricMIMO2}) over G. To satisfy the constraint of having a trellis with $|\mathcal{X}|^{\nu}$ states, we should have $\vec{G}\!=\![\vec{G}]_{\nu}$. The optimal $\vec{G}$, in the sense of maximizing (\ref{p1metricMIMO2}), will be shown (Theorem \ref{p1thm2}) to satisfy
\bea [(\vec{I}\!+\!\vec{G})^{-1}]_{\nu}\!=\!N_0[(\vec{H}\rmh\vec{H})^{-1}]_{\nu}.\notag \eea
Utilizing Lemma \ref{lem1}, the GMI in (\ref{p1metricMIMO2}) for the optimal $\vec{G}$ equals
{\setlength\arraycolsep{2pt}\bea I_{\mathrm{GMI}}=\log\!\big(\!\det(\vec{I}\!+\!\vec{G})\big)\!+\!\mathrm{Tr}\big(\vec{I}\!-\!N_0(\vec{H}\rmh\vec{H})^{-1}(\vec{I}\!+\!\vec{G})\big)
=\log\!\big(\!\det(\vec{I}\!+\!\vec{G})\big). \notag \eea}
\end{example}

\begin{example}{\label{p1exam2}}
Truncated Matched filter (TMF). As previously mentioned, the MAP demodulator (\ref{p1md0}) can be written in the form (\ref{p1md2}) by setting $\vec{V}\!=\!\vec{H}\rmh/{N_0}$, $\vec{R}\!=\!\vec{0}$ and $\vec{G}\!=\!\vec{H}\rmh\vec{H}/{N_0}$. The front-end is in this case a matched filter \cite{T72} and the BCJR needs to be implemented over the Ungerboeck model \cite{U74}. To reach a trellis with $|\mathcal{X}|^{\nu}$ states, we can truncate $\vec{G}$ to its center $2\nu\!+\!1$ diagonals, i.e., we can use the following parameters in (\ref{p1md2}):
\bea\vec{V}\!=\!\vec{H}\rmh/N_0, \; \vec{R}\!=\!\vec{0},  \;\text{and}  \; \vec{G}\!=\![\vec{H}\rmh\vec{H}/N_0]_{\nu}.\notag \eea
With these choices, the GMI in (\ref{p1metricMIMO2}) equals
{\setlength\arraycolsep{2pt}\bea I_{\mathrm{GMI}}=\log\!\big(\!\det(\vec{I}\!+\![\vec{H}\rmh\vec{H}/N_0]_{\nu})\big)\!-\!\mathrm{Tr}\big(\vec{H}\rmh\vec{H}\big(N_0\vec{I}\!+\![\vec{H}\rmh\vec{H}]_{\nu}\big)^{-1}[\vec{H}\rmh\vec{H}]_{\backslash\nu}\big). \notag \eea}
\end{example}

\subsection{Constraints on the Parameter $\vec{R}$ for the CS Demodulator}
As mentioned earlier, optimization of the demodulator will be made on the basis of GMI which is evaluated for the statistical model of the tuple $(\vec{x},\hat{\vec{y}})$. As illustrated in Fig.\,\ref{p1fig2}, our approach to design a CS demodulator consists of two steps:
\begin{itemize}
\item Construction of a signal $\hat{\vec{y}} \!=\! \vec{V}\vec{y}\!-\!\vec{R}\hat{\vec{x}}$ based on the received signal $\vec{y}$ and prior mean $\hat{\vec{x}}$;
\item BCJR demodulation of $\hat{\vec{y}}$ operating on a reduced number of states $|\mathcal{X}|^{\nu}$.
\end{itemize}

\begin{figure}[t]
\begin{center}
\vspace*{-2mm}
\scalebox{1.2}{\includegraphics{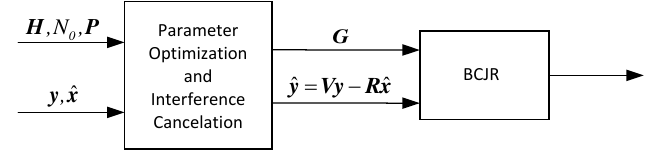}}
\vspace*{-4mm}
\caption{\label{p1fig2} CS demodulator that maximizes the GMI based on the tuple $(\hat{\vec{y}},\vec{x})$.}
\end{center}
\vspace*{-10mm}
\end{figure}

This procedure is fully analogous to LMMSE-PIC demodulator which first subtracts the interference, applies a Wiener filter, and concludes by a BCJR that operates with a diagonal matrix $\vec{G}$. The statistical behavior of $(\hat{\vec{y}},\vec{x})$ may be superior to that of the original $(\vec{y},\vec{x})$ as the former tuple corresponds to a statistically different channel than the true one. As what will be shortly shown in Example \ref{ex3}, the GMI obtained with tuple $(\hat{\vec{y}},\vec{x})$ based on perfect feedback $\hat{\vec{x}}$ can be infinitely large, which exceeds the channel capacity with the original tuple $(\vec{y},\vec{x})$. Therefore, the computed value of GMI may have little relevance for the performance of the transceiver system. In order for GMI to have bearing on performance, it is critical to put constraints on $\vec{R}$ as the next example will show.
\begin{example}\label{ex3}
Let the system model be
\bea \vec{y}=\vec{x}+\vec{n} \notag \eea
with noise density $N_0$, and $\vec{y}$, $\vec{x}$, $\vec{n}$ are $K\!\times\!1$ vectors. Assume perfect feedback information, i.e.,  $\hat{\vec{x}}=\vec{x}$. The demodulator parameters are taken as $\vec{V}=\vec{0}$, $\vec{R}=-(1\!+\!\beta)\vec{I}$, and $\vec{G}=\beta\vec{I}$, $\beta$ an arbitrary positive real value, then the statistical model for $\hat{\vec{y}}$ is
\bea \hat{\vec{y}} = \vec{V}\vec{y}-\vec{R}\hat{\vec{x}}=(1\!+\!\beta)\vec{x}. \notag
\eea
The GMI in (\ref{p1metricMIMO2}) for the tuple $(\vec{x}, \hat{\vec{y}})$ is 
\bea I_{\mathrm{GMI}}(\vec{V},\vec{R},\vec{G})=K\big(1\!+\!\log(1\!+\!\beta)\big).\notag \eea
\end{example}
In order to maximize the GMI, the demodulator will choose $\beta\!\rightarrow\!\infty$ to make $I_{\mathrm{GMI}}$ infinite. This is because, except for using the feedback information for IC, the demodulator uses the prior mean $\hat{\vec{x}}$ as a signal energy via $\vec{R}$. A demodulator equipped with these parameters will have significant error propagation and does not have much operational meaning for an iterative receiver. Thus, we conclude that unless constraints are put on $\vec{R}$, the GMI value is not relevant.

Three typical shapes of $\vec{R}$ are specified in Fig.\,\ref{p1fig3}. All three have in common that rather than adding signal energy, the rationale of $\vec{R}$ should be to remove interference. Therefore at the very minimum the diagonal elements of $\vec{R}$ should be constrained to zero, so that the demodulation of each symbol in $\vec{x}$ does not rely on its own prior mean $\hat{\vec{x}}.$ Such a constraint is perfectly aligned with the operations of LMMSE-PIC, where $\hat{x}_{\ell}$ is not used for demodulation of $x_{\ell}$. Furthermore, the rationale of the constraints we impose on $\vec{R}$ is to follow the principle of extrinsic information: The BCJR module should not rely on the prior information $\hat{x}_{\ell}$ when demodulating $x_{\ell}$ (this requires more than just the diagonal of $\vec{R}$ to be zero).

We point out that the fact that the GMI can exceed the channel capacity is a consequence of our choice not to include the side information as a prior distribution on $\vec{x}$ when evaluating the GMI. If we did, then the GMI is decaying with increasing quality of the side information (due to the mutual information $I(\vec{x},\vec{y}|\hat{\vec{x}})$ goes to 0 as $\hat{\vec{x}}$ becomes perfect). Finally, we acknowledge the fact that a permutation of the columns of $\vec{H}$ can boost the performance of the CS demodulator whenever $0\!<\!\nu\!<\!K\!-\!1$ for finite length linear vector channels. However, minimum-phase conversions of ISI channels are not beneficial as we will solve for the optimal front-end filter.

\begin{figure}
\begin{center}
\vspace*{-16mm}
\hspace*{-6mm}
\scalebox{.65}{\includegraphics{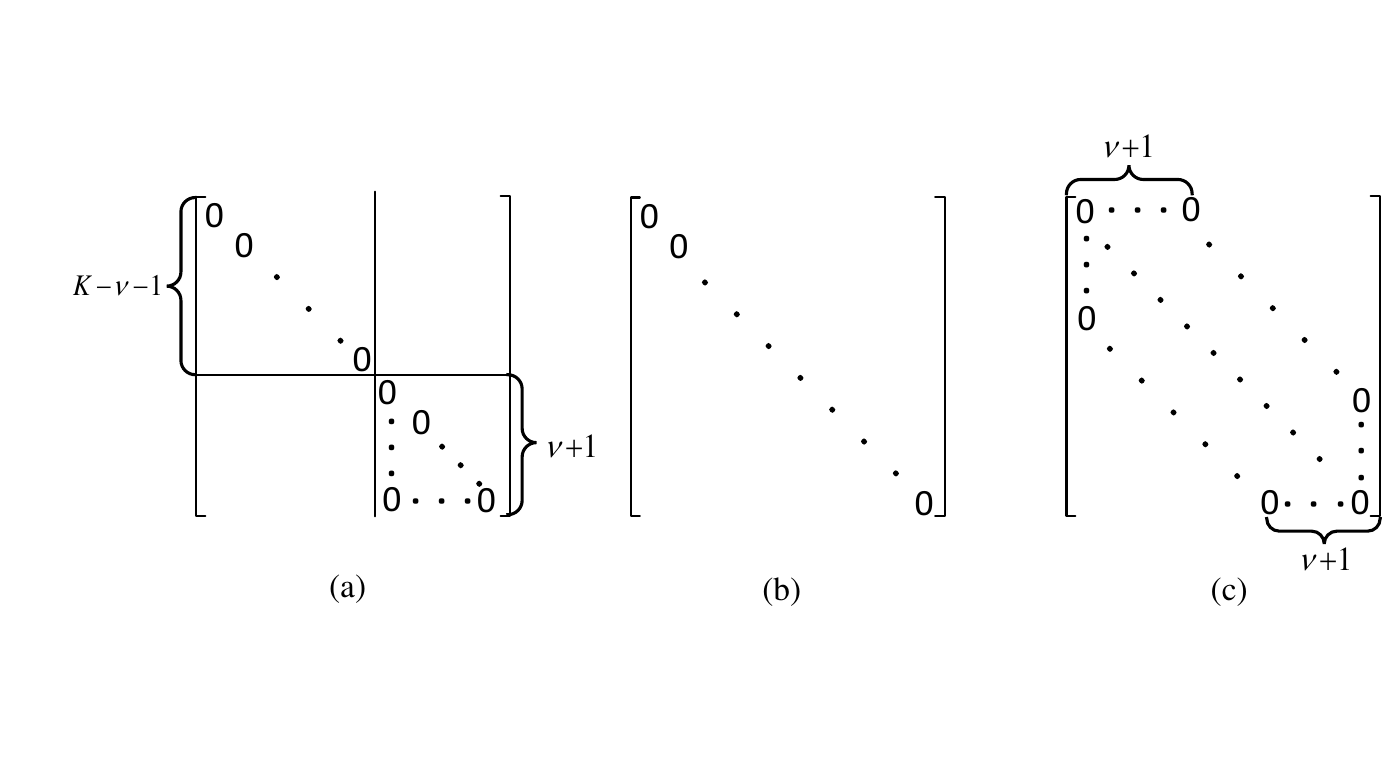}}
\vspace*{-25mm}
\caption{\label{p1fig3}  Three different types of shape of matrix $\vec{R}$, where $\nu$ is the memory size of $\vec{F}$ or $\vec{G}$, i.e., the memory size of the BCJR.}
\vspace*{-10mm}
\end{center}
\end{figure}

\section{Parameter Optimization for Finite Length Linear Vector Channel}

In this section, we elaborate the parameter optimization for finite length linear vector channels. We introduce three different methods, Method I, Method II and Method III. We start with the classical Forney model (\ref{p1md1}) based demodulator, i.e., Method I, and then extend the demodulation model into the Ungerboeck model (\ref{p1md2}), i.e., Method II. As both Method I and Method II need gradient-based approach for the optimization of target response, by carefully examining the properties of the CS demodulator with Ungerboeck model, we propose a suboptimal Method III which has an explicit construction based on an LMMSE-PIC and all parameters are in closed-forms.

\subsection{Method I} \label{p1MIMO_M1}
In Method I, the CS demodulator is based on detection model (\ref{p1md1}) and the following structures of the CS parameters $(\vec{W},\vec{T},\vec{F})$ are imposed:
\begin{itemize}
\item $\vec{W}$ is a $K\!\times\! N$ matrix with no constraints.
\item $\vec{F}$ is a $K\!\times\! K$ lower triangular matrix where only the main diagonal and the first $\nu$ lower diagonals are nonzero, i.e., $\vec{F}$ is banded within diagonals $[0,\nu]$ ($0\!\leq\!\nu\!<\!K\!-\!1$), where $\nu$ is denoted as the memory size of $\vec{F}$. Moreover, the main diagonal of $\vec{F}$ is constrained to only contain positive real values.
\item $\vec{T}$ is a $K\!\times\! K$ matrix that is constrained to be zero wherever $\vec{F}$ can take nonzero values.
\end{itemize}

The constraint of $\vec{F}$ is to shorten the memory for the trellis-search in BCJR, while the purpose of the constraint on $\vec{T}$ is to cancel the signal part that $\vec{F}$ can not handle. From Theorem 1, and by identifying $\vec{V}=\vec{F}\rmh\vec{W}$, $\vec{R}=\vec{F}\rmh\vec{T}$, and $\vec{G}=\vec{F}\rmh\vec{F}$, the GMI in (\ref{p1metricMIMO2}) of Method I equals
{\setlength\arraycolsep{2pt}\bea \label{p1metricMIMO1} I_{\mathrm{GMI}}(\vec{W}, \vec{T}, \vec{F})&=&\log\!\big(\!\det(\vec{I}\!+\!\vec{F}\rmh\vec{F})\big)\!-\!\mathrm{Tr}(\vec{F}\rmh\vec{F}) \!+\!2\Re\big\{\mathrm{Tr}\big(\vec{F}\rmh(\vec{W}\vec{H}\!-\!\vec{T}\vec{P})\big)\!\big\}\nonumber \\
&&-\mathrm{Tr}\big((\vec{I}\!+\!\vec{F}\rmh\vec{F})^{-\!1}\vec{L}_1\big)\quad\ \eea}
\hspace*{-1.4mm}where $$\vec{L}_1\!=\!\vec{F}\rmh\vec{W}(N_0\vec{I}\!+\!\vec{H}\vec{H}\rmh)\vec{W}\rmh\vec{F}\!-\!2\Re\big\{\vec{F}\rmh\vec{W}\vec{H}\vec{P}\vec{T}\rmh\vec{F}\big\}\!+\!\vec{F}\rmh\vec{T}\vec{P}\vec{T}\rmh\vec{F}.$$ With the aforementioned constraints on $\vec{F}$ and $\vec{T}$, the matrix $\vec{R}\!=\!\vec{F}\rmh\vec{T}$ has a form of shape (a) in Fig.\,\ref{p1fig3}. That is, all diagonal elements are zero as well as the lower triangular part of the $(\nu\!+\!1)\!\times\!(\nu\!+\!1)$ small matrix at the right bottom corner.

In order to optimize (\ref{p1metricMIMO1}) over $(\vec{W},\vec{T},\vec{F})$, we first introduce an $S\!\times\!K^2$ indication matrix $\vec{\Omega}$ only consisting of ones and zeros\footnote{For instance, assuming $\vec{T}\!=\!\left[\!\begin{array}{cc}    0&   1\\
   2&   0\end{array} \!\right]\!$, then the indication matrix $\vec{\Omega}\!=\!\left[\!\begin{array}{cccc}  0&1& 0& 0\\
 0&0& 1& 0\end{array} \!\right]\!$, and the vector $\vec{\Omega}\mathrm{vec}(\vec{T})\!=\!\left[\!\begin{array}{rr}  2\\
1\end{array} \!\right]\!$.}, having a single 1 in each row, and $S$ equals the number of elements in $\vec{T}$ that are allowed to be nonzero. Let $\mathbb{I}(\mathrm{vec}(\vec{T}))$ be a vector that contains the positions where the vector $\mathrm{vec}(\vec{T})$ is allowed to be nonzero. Then the value of the $k$th entry in $\mathbb{I}(\mathrm{vec}(\vec{T}))$ gives the column where row $k$ of $\vec{\Omega}$ is 1. That is, the $S\!\times\!1$ vector $\vec{\Omega}\mathrm{vec}(\vec{T})$ stacks the columns of $\vec{T}$ on top of each other but with all elements that are constrained to zero removed.

With such a definition of $\vec{\Omega}$, and define two $K\!\times\!K$ matrices as,
 {\setlength\arraycolsep{2pt}\bea  \label{p1M} \vec{M}&=&\vec{H}\rmh(N_0\vec{I}\!+\!\vec{HH}\rmh)^{-\!1}\vec{H}\!-\!\vec{I},\\
 \label{p1Mt} \tilde{\vec{M}}&=&\vec{P}(\vec{I}\!+\!\vec{M})\vec{P}\!-\!\vec{P},\eea}
\hspace*{-1.4mm}the GMI for the optimal $\vec{W}$ and $\vec{T}$ is given in Proposition \ref{p1prop1} and the proof is in Appendix B.
\begin{proposition} \label{p1prop1}
Define an $S\!\times\!K^2$ matrix $\vec{D}=\vec{\Omega}\big((\vec{PM}^{\ast})\!\otimes\!\vec{I}_{K}\big)$, the optimal $\vec{W}$ maximizing the GMI in (\ref{p1metricMIMO1}) is
 {\setlength\arraycolsep{2pt} \bea \label{p1mimo1optW} \vec{W}_{\mathrm{opt}} = \vec{F}\rmih(\vec{I}\!+\!\vec{F}\rmh\vec{F}\!+\!\vec{F}\rmh\vec{T}\vec{P})\vec{H}\rmh(N_0\vec{I}\!+\!\vec{HH}\rmh)^{-\!1},\eea}
\hspace*{-1.4mm}and when $\vec{P}\!\neq\!\vec{0}$, the optimal $\vec{T}$ maximizing the GMI is given by
 {\setlength\arraycolsep{2pt} \bea \label{p1mimo1optT} \mathrm{vec}(\vec{T}_{\mathrm{opt}})=
 -\vec{\Omega}\rmt\big(\vec{\Omega}\big(\tilde{\vec{M}}^{\ast}\!\otimes\!\big(\vec{F}(\vec{I}\!+\!\vec{F}\rmh\vec{F})^{-1}\!\vec{F}\rmh\big)\big)\vec{\Omega}\rmt\big)^{-1}\vec{D}\mathrm{vec}(\vec{F}).\eea}
\hspace*{-1.4mm}With the optimal $\vec{W}$ and $\vec{T}$, the GMI reads,
{\setlength\arraycolsep{2pt} \bea \label{p1Iwt} I_{\mathrm{GMI}}(\vec{W}_{\mathrm{opt}},\vec{T}_{\mathrm{opt}},\vec{F})=\left\{\begin{array}{ll} I_1(\vec{F}),& \vec{P}\!=\!\vec{0}\\
I_1(\vec{F})+\delta_1(\vec{F}),& \vec{P}\!\neq\!\vec{0} \end{array}\right.\eea}
\hspace*{-1.4mm}where the functions $I_1(\vec{F})$ and $\delta_1(\vec{F})$ are defined as
{\setlength\arraycolsep{2pt} \bea \label{p1Iwt2} I_1(\vec{F})&=& K\!+\! \log\!\big(\!\det(\vec{I}\!+\!\vec{F}\rmh\vec{F})\big)\!+\!\mathrm{Tr}\big(\vec{M}(\vec{I}\!+\!\vec{F}\rmh\vec{F})\big), \\
 \delta_1(\vec{F}) &=& -\mathrm{vec}(\vec{F})\rmh\!\vec{D}\rmh\Big(\vec{\Omega}\big(\tilde{\vec{M}}^{\ast}\!\otimes\!\big(\vec{F}(\vec{I}\!+\!\vec{F}\rmh\vec{F})^{-\!1}\vec{F}\rmh\big)\big)\vec{\Omega}\rmt\Big)^{-\!1}\!\vec{D}\mathrm{vec}(\vec{F}).\eea}
\end{proposition}

\begin{remark}
With the definitions in (\ref{p1M}) and (\ref{p1Mt}), $\vec{M}$ is the negative of the MSE matrix and $\tilde{\vec{M}}\!\preceq\!0$ holds. Hence $\delta_1(\vec{F})\!\geq\!0$ represents the GMI increments from the soft feedback.
\end{remark}

Before discussing the GMI-maximization of (\ref{p1Iwt}), we first state Theorem \ref{p1thm2} that deals with a general maximization problem.

\begin{theorem}{\label{p1thm2}}
Define a scalar function $I$ with respect to a $K\!\times\!K$ matrix $\vec{G}$ as
 \bea  \label{p1gerform} I(\vec{G})= K\!+\! \log\!\big(\!\det(\vec{I}\!+\!\vec{G})\big)\!+\!\mathrm{Tr}\big(\vec{M}(\vec{I}\!+\!\vec{G})\big) \eea
where $\vec{G}$ satisfies $\vec{G}\!=\![\vec{G}]_{\nu}$. Then the optimal $\vec{G}$ maximizing $I$ is the unique solution that satisfies
 \bea \label{p1optcond} [(\vec{I}\!+\!\vec{G}_{\mathrm{opt}})^{-1}]_{\nu}=-[\vec{M}]_{\nu}. \eea
With $\vec{G}_{\mathrm{opt}}$, the maximal $I$ equals
 \bea \label{p1optI} I(\vec{G}_{\mathrm{opt}})=\log\!\big(\!\det(\vec{I}\!+\!\vec{G}_{\mathrm{opt}})\big). \eea
\end{theorem}
\begin{proof}
Taking the first order differential of $I$ with respect to $\vec{G}$ and noticing that $\vec{G}$ is banded within diagonals $[-\nu,\nu]$, yields (\ref{p1optcond}) after some manipulations. The existence and uniqueness of such an optimal solution for (\ref{p1optcond}) is proved in \cite[Theorem 2]{KM00} and also illustrated in \cite[Proposition 2]{RP12}. By Lemma \ref{lem1}, $\mathrm{Tr}\big([\vec{I}+\vec{G}_{\mathrm{opt}}]^{-1}\vec{M}\big)\!=\!-K$ from (\ref{p1optcond}), and then (\ref{p1optI}) follows.
\end{proof}

Optimizing over $\vec{F}$ in (\ref{p1Iwt}) when $\vec{P}\!\neq\!\vec{0}$ is difficult and cannot be carried out in closed-form. In Appendix C we show by an example that (\ref{p1Iwt}) is in general non-concave. Therefore, a gradient based numerical optimization procedure is utilized to search for the optimal $\vec{F}$. In the $i$th iteration, we construct
$$\vec{F}^{(i)}=\vec{F}^{(i-1)}\!+\!\nabla_{\vec{F}^{\ast}}I_{\mathrm{GMI}}\big(\vec{W}_{\mathrm{opt}},\vec{T}_{\mathrm{opt}},\vec{F}^{(i-1)}\big)$$ where $\nabla_{\vec{F}^{\ast}}I_{\mathrm{GMI}}(\vec{W}_{\mathrm{opt}},\vec{T}_{\mathrm{opt}},\vec{F})$ is the conjugate of the gradient of the GMI with respect to (the nonzero part of) $\vec{F}$, and is given in Appendix D.

With $\vec{P}\!=\!\vec{0}$, if replacing $\vec{F}\rmh\vec{F}$ by $\vec{G}$, (\ref{p1Iwt2}) has the same form as (\ref{p1gerform}), and $\vec{G}_{\mathrm{opt}}$ is in closed-form as stated in Theorem \ref{p1thm2}. If $\vec{G}_{\mathrm{opt}}\!\succeq\!\vec{0}$, the optimal $\vec{F}$ then equals the Cholesky decomposition of $\vec{G}_{\mathrm{opt}}$. Whenever it is not, a gradient based numerical optimization procedure is utilized to optimize (\ref{p1Iwt2}), and $\vec{G}_{\mathrm{opt}}$ from Theorem 2 is used to initialize the starting point of $\vec{F}$ for any $\vec{P}$, which has been observed to be highly reliable.

Next we establish a connection between the front-end filter $\vec{W}$ and IC matrix $\vec{T}$ in Method I.

\begin{proposition} \label{p1prop2}
For $\vec{P}\!\neq\!\vec{0}$, and with the optimal $\vec{W}$ and $\vec{T}$, the matrix $\vec{F}\rmh(\vec{W}_{\mathrm{opt}}\vec{H}-\vec{T}_{\mathrm{opt}})$ is banded within diagonals $[-\nu,K\!-\!1]$.
\end{proposition}
\begin{proof} 
Noting that $\vec{\Omega}\rmt\vec{\Omega}\mathrm{vec}(\vec{T}_{\mathrm{opt}})\!=\!\mathrm{vec}(\vec{T}_{\mathrm{opt}})$ and $\vec{\Omega}\vec{\Omega}\rmt \!=\!\vec{I}$, from (\ref{p1mimo1optT}) and (\ref{p1mimo1optt_app}), it holds that
{\setlength\arraycolsep{1pt}\bea \label{p1appe1} \vec{\Omega}\big(\tilde{\vec{M}}^{\ast}\!\otimes\!\big(\vec{F}(\vec{I}\!+\!\vec{F}\rmh\vec{F})^{-1}\!\vec{F}\rmh\big)\big)\vec{\Omega}\rmt\vec{\Omega}\mathrm{vec}(\vec{T}_{\mathrm{opt}})&=&\vec{\Omega}\mathrm{vec}\big(\vec{F}(\vec{I}\!+\!\vec{F}\rmh\vec{F})^{-1}\!\vec{F}\rmh\vec{T}_{\mathrm{opt}}\tilde{\vec{M}}\big)\notag\\
&=&-\vec{\Omega}\mathrm{vec}(\vec{FMP}),\eea}
\hspace*{-1.4mm}which shows that, the elements of the matrix $\Delta\!=\!\vec{F}(\vec{I}\!+\!\vec{F}\rmh\vec{F})^{-1}\!\vec{F}\rmh\vec{T}_{\mathrm{opt}}\tilde{\vec{M}}\!+\!\vec{FMP}$ are zero wherever $\vec{T}$ can be nonzero. Hence $\Delta$ is banded within diagonals $[0,\nu]$. On the other hand, with the optimal $\vec{W}$ given in (\ref{p1mimo1optW}) and $\vec{M}$, $\tilde{\vec{M}}$ defined in (\ref{p1M}) and (\ref{p1Mt}), we have
{\setlength\arraycolsep{2pt} \bea  \label{p1appe3} \vec{F}\rmh(\vec{W}_{\mathrm{opt}}\vec{H}\!-\!\vec{T}_{\mathrm{opt}})\!-\!(\vec{I}\!+\!\vec{F}\rmh\vec{F})=(\vec{I}\!+\!\vec{F}\rmh\vec{F})\vec{F}^{-\!1}\Delta\vec{P}^{-\!1}. \eea}
\hspace*{-1.4mm}Note that, $\vec{F}^{-1}$ is lower triangular, $\vec{I}\!+\!\vec{F}\rmh\vec{F}$ is banded within diagonals $[-\nu,\nu]$, and $\vec{P}$ is diagonal. Utilizing Lemma \ref{lem2}, the r.h.s in (\ref{p1appe3}) is banded within diagonals $[-\nu,K\!-\!1]$. Therefore $\vec{F}\rmh(\vec{W}_{\mathrm{opt}}\vec{H}-\vec{T}_{\mathrm{opt}})$ is also banded within diagonals $[-\nu,K\!-\!1]$.
\end{proof}

Proposition 2 reveals an interesting and somewhat surprising fact that, although the BCJR only has a memory size $\nu$, the interference outside the memory size $\nu$ shall not be perfectly canceled with the optimal CS demodulator in Method I. As will be shown later, such a property also holds for the other two designs of CS demodulator, i.e., Method II and III.

\subsection{Method II}\label{MIMO_M2}
Method II origins from Ungerboeck's 1974 paper \cite{U74}. Different from Method I, an Ungerboeck detection model (\ref{p1md2}) instead of the Forney model (\ref{p1md1}) is applied. The Ungerboeck model has been extensively discussed in \cite{CB, RC15, FMP07}. The system model (\ref{p1md2}) has the following constraints:
\begin{itemize}
\item $\vec{V}$ is a $K\!\times\!N$ matrix with no constraints.
\item $\vec{G}$ is a $K\!\times\!K$ Hermitian matrix satisfying $\vec{G}\!=\![\vec{G}]_{\nu}$ and $\vec{I}\!+\!\vec{G}\!\succ\!0$, where $\nu$ is the memory size of $\vec{G}$.
\item $\vec{R}$ is a $K\!\times\!K$ matrix where the shape can be specified.
\end{itemize}
Instead of optimizing $(\vec{W},\vec{T},\vec{F})$, in Method II we optimize $(\vec{V}, \vec{R}, \vec{G})$ for (\ref{p1metricMIMO2}). The same definition of the indication matrix $\vec{\Omega}$ is used as in Method I, but now $\vec{\Omega}$ corresponds to $\vec{R}$ instead of $\vec{T}$. We continue to let $S$ denote the number of elements that are allowed to be nonzero in $\vec{R}$. That is, the $S\!\times\!1$ vector $\vec{\Omega}\mathrm{vec}(\vec{R})$ stacks the columns of $\vec{R}$ on top of each other but with all elements that are constrained to zero removed. In Method II, we have Proposition \ref{p1prop3} that shows the GMI calculation with optimal $\vec{V}$ and $\vec{R}$.

\begin{proposition} \label{p1prop3}
Define an $S\!\times\!1$ vector $\vec{d}=\vec{\Omega}\mathrm{vec}(\vec{MP})$, the optimal $\vec{V}$ for the GMI in (\ref{p1metricMIMO2}) is
 {\setlength\arraycolsep{2pt} \bea \label{p1optgm2} \vec{V}_{\mathrm{opt}} = (\vec{I}+\vec{G}+\vec{R_{\mathrm{opt}}P})\vec{H}\rmh(\vec{HH}\rmh+N_0\vec{I})^{-1},\eea}
\hspace*{-1.4mm}and when $\vec{P}\!\neq\!\vec{0}$, the optimal $\vec{R}$ maximizing the GMI is given by,
 {\setlength\arraycolsep{2pt} \bea \label{p1mimo2optR} \mathrm{vec}(\vec{R}_{\mathrm{opt}})=
 -\vec{\Omega}\rmt\big(\vec{\Omega}\big(\tilde{\vec{M}}^{\ast}\!\otimes\!(\vec{I}+\vec{G})^{-1}\big)\vec{\Omega}\rmt\big)^{-1}\vec{d}.\eea}
\hspace*{-1.4mm}With the optimal $\vec{V}$ and $\vec{R}$, the GMI in (\ref{p1metricMIMO2}) equals
{\setlength\arraycolsep{2pt} \bea \label{p1Ivr} I_{\mathrm{GMI}}(\vec{V}_{\mathrm{opt}}, \vec{R}_{\mathrm{opt}},\vec{G} )=\left\{\begin{array}{ll} I_2(\vec{G}),& \vec{P}\!=\!\vec{0}\\
I_2(\vec{G})+\delta_2(\vec{G}),& \vec{P}\!\neq\!\vec{0} \end{array}\right.\eea}
\hspace*{-1.4mm}where the functions $ I_2(\vec{G})$ and $\delta_2(\vec{G})$ are defined as,
{\setlength\arraycolsep{2pt} \bea \label{p1Ivr1} I_2(\vec{G})&=& K\!+\!\log\!\big(\!\det(\vec{I}\!+\!\vec{G})\big)\!+\!\mathrm{Tr}\big(\vec{M}(\vec{I}\!+\!\vec{G})\big), \\
\label{p1deta2}\delta_2(\vec{G}) &=& -\vec{d}\rmh\big(\vec{\Omega}\big(\tilde{\vec{M}}^{\ast}\!\otimes\! (\vec{I}\!+\!\vec{G})^{-\!1}\big)\vec{\Omega}\rmt\big)^{\!-\!1}\!\vec{d}.\eea}
\end{proposition}
The proof is given in Appendix E. Similar to $\delta_1(\vec{F})$ in Method I, $\delta_2(\vec{G})\!\geq\!0$ represents the GMI increment from the soft information.

When $\vec{P}\!\neq\!\vec{0}$, the optimization over $\vec{G}$ in (\ref{p1Ivr}) also uses a gradient based numerical optimization, and the gradient of $I_{\mathrm{GMI}}(\vec{V}_{\mathrm{opt}}, \vec{R}_{\mathrm{opt}},\vec{G})$ with respect to (the nonzero part of) $\vec{G}$ is provided in Appendix F. The closed-from $\vec{G}$ from Theorem \ref{p1thm2} with $\vec{P}\!=\!\vec{0}$ is still used as the starting point for $\vec{P}\!\neq\!\vec{0}$. However, different from Method I, the optimization procedure is concave and the proof is given in Appendix G.

Although the optimal $\vec{R}$ is solved for in closed-form as in (\ref{p1mimo2optR}), we shall specify the constraint (reflected by $\vec{\Omega}$) on it. We consider two types of $\vec{R}$ in Method II. Firstly, as we are interested in the comparison between Method I and Method II, we also consider the shape (a) in Fig.\,\ref{p1fig3}, which has the same shape as for $\vec{R}\!=\!\vec{F}\rmh\vec{T}$ in Method I. Secondly, we consider a band-shaped $\vec{R}$ with memory size $\nR$, where shape (b) and (c) in Fig.\,\ref{p1fig3} are typical cases with $\nR\!=\!0$ and $\nR\!=\!\nu$, respectively. With shape (b), we only limit the diagonal elements of $\vec{R}$ to be zero and intend to eliminate the interference as much as possible. With shape (c), we limit $\vec{R}$ to have the opposite form of $\vec{G}$, that is, the elements of $\vec{R}$ are constrained to be zero wherever $\vec{G}$ is nonzero. The intention is to only cancel the interference that the BCJR represented by $\vec{G}$ cannot handle. Shape (c) is based on the same idea as Method I, but now operates on the Ungerboeck model.

The connection between the optimal front-end filter $\vec{V}$ and IC matrix $\vec{R}$ in Method II is now established in Proposition \ref{p1prop4}.

\begin{proposition} \label{p1prop4}
For $\vec{P}\!\neq\!\vec{0}$ and the optimal $\vec{V}$ and $\vec{R}$,
\be [\vec{V}_{\mathrm{opt}}\vec{H}]_{\backslash(\nu\!+\!\nu_{\mathrm{R}})}= [\vec{R}_{\mathrm{opt}}]_{\backslash(\nu\!+\!\nu_{\mathrm{R}})}. \ee
That is, the elements of $\vec{V}_{\mathrm{opt}}\vec{H}$ and $\vec{R}_{\mathrm{opt}}$ are equal outside the center $2(\nu\!+\!\nu_{\mathrm{R}})\!+\!1$ diagonals for any $\vec{G}$ that is banded within diagonals $[-\nu,\nu]$, where $\nu_{\mathrm{R}}\!=\!0$ for $\vec{R}$ with both shape (a) and (b), while $\nu_{\mathrm{R}}\!=\!\nu$ for $\vec{R}$ with shape (c).
\end{proposition}

\begin{proof} 
Following similar steps as in the proof of Proposition \ref{p1prop2}, (\ref{p1mimo2optR}) can be rewritten as,
 {\setlength\arraycolsep{2pt} \bea \vec{\Omega}\mathrm{vec}\big((\vec{I}+\vec{G})^{-1}\vec{R}_{\mathrm{opt}}\tilde{\vec{M}}\big)=-\vec{\Omega}\mathrm{vec}(\vec{MP}).\eea}
\hspace*{-1.4mm}It shows that, the elements of the matrix $\Delta\!=\!(\vec{I}+\vec{G})^{-1}\vec{R}_{\mathrm{opt}}\tilde{\vec{M}}\vec{P}^{-1}\!+\!\vec{M}$ are zero wherever $\vec{R}$ can be nonzero. On the other hand, with the optimal $\vec{V}$ in (\ref{p1optgm2}) we have
\bea \label{p1appi3}\vec{V}_{\mathrm{opt}}\vec{H}\!-\!\vec{R}_{\mathrm{opt}}\!-\!(\vec{I}\!+\!\vec{G})\!=\!(\vec{I}\!+\!\vec{G})\Delta. \eea
As $\vec{I}\!+\!\vec{G}$ is banded within diagonals $[-\nu,\nu]$, utilizing Lemma \ref{lem2} ($\vec{R}$ with shape (a) is slightly different, but it can be verified straightforwardly), and with the three shapes of $\vec{R}$ in Fig.\,\ref{p1fig3}, it can be shown that the r.h.s in (\ref{p1appi3}) is banded within diagonals $[-(\nu+\nu_{\mathrm{R}}),\nu+\nu_{\mathrm{R}}]$, where $\nu_{\mathrm{R}}\!=\!0$ for the shape (a) and (b), and $\nu_{\mathrm{R}}\!=\!\nu$ fot the shape (c). Therefore, $\vec{V}_{\mathrm{opt}}\vec{H}\!-\!\vec{R}_{\mathrm{opt}}$ on the l.h.s in (\ref{p1appi3}) is banded within diagonals $[-(\nu+\nu_{\mathrm{R}}),\nu+\nu_{\mathrm{R}}]$.
\end{proof}

The same as Proposition 2 for Method I, Proposition 4 shows that the signal part that is not considered in $\vec{G}$ (the BCJR) shall not be perfectly canceled inside the center $2(\nu\!+\!\nu_{\mathrm{R}})\!+\!1$ diagonals for Method II, instead of the center $2\nu\!+\!1$ diagonals where $\vec{G}$ is constrained to be nonzero. With LMMSE-PIC, we have $\nu\!=\!\nu_{\mathrm{R}}\!=\!0$ and Proposition \ref{p1prop4} is natural and frequently used. However, when $\nu_{\mathrm{R}}\!>\!0$, a more general property is revealed that, $\vec{V}_{\mathrm{opt}}\vec{H}$ and $\vec{R}$ are only equal outside the center $2(\nu\!+\!\nu_{\mathrm{R}})\!+\!1$ diagonals.

\subsection{Method III}
So far we have discussed two types of CS demodulators based on Forney and Ungerboeck detection models, respectively. One disadvantage of them is that, in general both methods need an numerical optimization to obtain the optimal target response. Next, we construct a third method that has closed-form solutions for all CS parameters, although its GMI is suboptimal in general.

Method III relies on the same operations as Method II for $\vec{P}\!=\!\vec{0}$. By inserting $\vec{V}_{\mathrm{opt}}$ in (\ref{p1optgm2}) into (\ref{p1md2}) and setting $\vec{R}\!=\!\vec{0}$, the demodulator actually operates on the mismatched function
{\setlength\arraycolsep{2pt}\bea \label{p1m3p}  \tilde{p}(\vec{y}|\vec{x})&=&\exp\!\big(2\Re\{\vec{x}\rmh\vec{V}_{\mathrm{opt}}\vec{y}\}\!-\!\vec{x}\rmh\vec{G}\vec{x}\big) \notag \\
&=&\exp\!\big(2\Re\big\{\vec{x}\rmh(\vec{I}\!+\!\vec{G})\check{\vec{x}}\big\}\!-\!\vec{x}\rmh\vec{G}\vec{x}\big)\eea}
\hspace*{-1.4mm}where $\check{\vec{x}}\!=\!\vec{H}\rmh(\vec{H\!H}\rmh\!+\!N_0\vec{I})^{-\!1}\vec{y}$ is the LMMSE estimate. As can be seen from (\ref{p1m3p}), the BCJR is based on $\check{\vec{x}}$. With soft feedback, we can therefore replace $\check{\vec{x}}$ by LMMSE-PIC estimates $\tilde{\vec{x}}$. That is, instead of (\ref{p1m3p}) we operate on
\bea  \label{p1md30} \tilde{p}(\vec{y}|\vec{x},\vec{\tilde{x}}) =\exp\!\big(2\Re\!\left\{\vec{x}\rmh(\vec{I}\!+\!\vec{G})\tilde{\vec{x}}\right\}\!-\!\vec{x}\rmh\vec{G}\vec{x}\big) \eea
where $\vec{G}$ has the same banded-shape as the first two methods, but optimized according to $\vec{\tilde{x}}$. The estimate $\tilde{\vec{x}}$ is constructed as follows. As we prefer to handle the interference through the trellis-search process, the IC should not be present within the memory size $\nu$. In other words, the signal vector after the IC that is used to form the $k$th symbol of $\tilde{\vec{x}}$ is denoted as $\tilde{\vec{y}}_k$ and
\bea  \tilde{\vec{y}}_k\!=\!\vec{y}\!-\! \sum_{n\in\mathcal{A}_k}\vec{\vec{h}}_n\hat{x}_n\eea
where $\mathcal{A}_k\!=\!\big\{0\!\leq\! n\!\leq\! K\!-\!1\!:\!n\!\notin\![\max(0,k\!-\!\nu), \min( k\!+\!\nu,K\!-\!1)]\big\}$. Denote $p_n$ as the $n$th diagonal element of $\vec{P}$, the Wiener filtering coefficients \cite{OB98} for the $k$th symbol are calculated through
\bea \label{p1wiec} \hat{\vec{w}}_k=\vec{h}_k^{\rmh}(\vec{H}\rmh\vec{C}_k\vec{H}+\!N_0\vec{I}\big)^{-\!1}\eea
where $\vec{C}_k$ is a diagonal matrix with the $n$th diagonal element defined as
 \bea \label{p1matc} C_k(n)=\left\{\begin{array}{ll} 1-p_n,& k\in\mathcal{A}_k\\
1,& \mathrm{otherwise}. \end{array}\right.\eea
The estimate $\tilde{\vec{x}}$ is then obtained through
\bea \label{p1picxest} \tilde{\vec{x}}\!=\!\big[\!\!\begin{array}{cccc}\hat{\vec{w}}_1\tilde{\vec{y}}_1& \hat{\vec{w}}_2\tilde{\vec{y}}_2& \cdots& \hat{\vec{w}}_K\tilde{\vec{y}}_K
\end{array}\!\!\big]\rmt\!=\!\hat{\vec{W}}\vec{y}\!-\!\hat{\vec{C}}\hat{\vec{x}}\eea
where the coefficient matrix $\hat{\vec{W}}$ and IC matrix $\hat{\vec{C}}$ defined as
{\setlength\arraycolsep{2pt}\bea  \hat{\vec{W}}\!&=&\![\begin{array}{cccc}\hat{\vec{w}}_1\rmt& \hat{\vec{w}}_2\rmt& \cdots& \hat{\vec{w}}_K\rmt\end{array}]\rmt, \\
 \label{p1covmatc} \hat{\vec{C}}\!&=&\![\hat{\vec{W}}\vec{H}]_{\backslash\nu}.\eea}
\hspace*{-1.4mm}Inserting $\tilde{\vec{x}}$ in (\ref{p1picxest}) back into (\ref{p1md30}), the detection model we operate on reads
\bea \label{p1md3} \tilde{p}(\vec{y}|\vec{x},\vec{\hat{x}})=\exp\!\big(2\Re\big\{\vec{x}\rmh\big((\vec{I}\!+\!\vec{G})\hat{\vec{W}}\vec{y}\!-\!(\vec{I}\!+\!\vec{G})\hat{\vec{C}}\hat{\vec{x}}\big)\big\}\!-\!\vec{x}\rmh\vec{G}\vec{x}\big).
\eea
Note that, (\ref{p1md3}) is a also special case of (\ref{p1md2}) by identifying 
{\setlength\arraycolsep{2pt}\bea
\vec{V}&=&(\vec{I}\!+\!\vec{G})\tilde{\vec{W}}, \notag \\
\vec{R}&=&(\vec{I}\!+\!\vec{G})\tilde{\vec{C}}. \notag \eea}
\hspace*{-1.4mm}The GMI in (\ref{p1metricMIMO2}) in this case reads, after some manipulations,
{\setlength\arraycolsep{2pt}\bea \label{p1metricMIMO3} I_{\mathrm{GMI}}(\vec{G})\!&=&\!K\!+\!\log\!\big(\!\det(\vec{I}\!+\!\vec{G})\big)\!+\!\mathrm{Tr}\big(\hat{\vec{M}}(\vec{I} \!+\!\vec{G})\big) \eea}
\hspace{-1.4mm}with $\hat{\vec{M}}$ (the updated $\vec{M}$ in Method II) defined as
{\setlength\arraycolsep{2pt}\bea \label{p1mpic} \hat{\vec{M}}&=&\hat{\vec{W}}\vec{H}\vec{P}\hat{\vec{C}}\rmh\!+\!\hat{\vec{W}}\vec{H}\!-\!\vec{P}\hat{\vec{C}}\rmh\!+\!\big(\hat{\vec{W}}\vec{H}\vec{P}\hat{\vec{C}}\rmh\!+\!\hat{\vec{W}}\vec{H}\!-\!\vec{P}\hat{\vec{C}}\rmh\big)\rmh \notag \\
&& -\hat{\vec{W}}(\vec{H}\vec{H}\rmh\!+\!N_0\vec{I})\hat{\vec{W}}\rmh\!-\!\hat{\vec{C}}\vec{P}\hat{\vec{C}}\rmh
\!-\!\vec{I}, \quad\eea}
\hspace{-1.4mm}which can be shown to be the negative of the MSE matrix since $$\hat{\vec{M}}\!=\!-\mathbb{E}\big[(\vec {x}\!-\!\tilde{\vec{x}})(\vec{x}\!-\!\tilde{\vec{x}})\rmh\big]\!=\!-\mathbb{E}\big[(\vec {x}\!-\!\hat{\vec{W}}\vec{y}\!+\!\hat{\vec{C}}\hat{\vec{x}})(\vec{x}\!-\!\hat{\vec{W}}\vec{y}\!+\!\hat{\vec{C}}\hat{\vec{x}})\rmh\big].$$
The optimal $\vec{G}$ for (\ref{p1metricMIMO3}) is then obtained from Theorem \ref{p1thm2}, and the optimal GMI reads
\bea  I_{\mathrm{GMI}}(\vec{G}_{\mathrm{opt}})=\log\!\big(\!\det(\vec{I}\!+\!\vec{G}_{\mathrm{opt}})\big). \notag \eea

\begin{figure}[t]
\begin{center}
\hspace*{-2mm}
\scalebox{.7}{\includegraphics{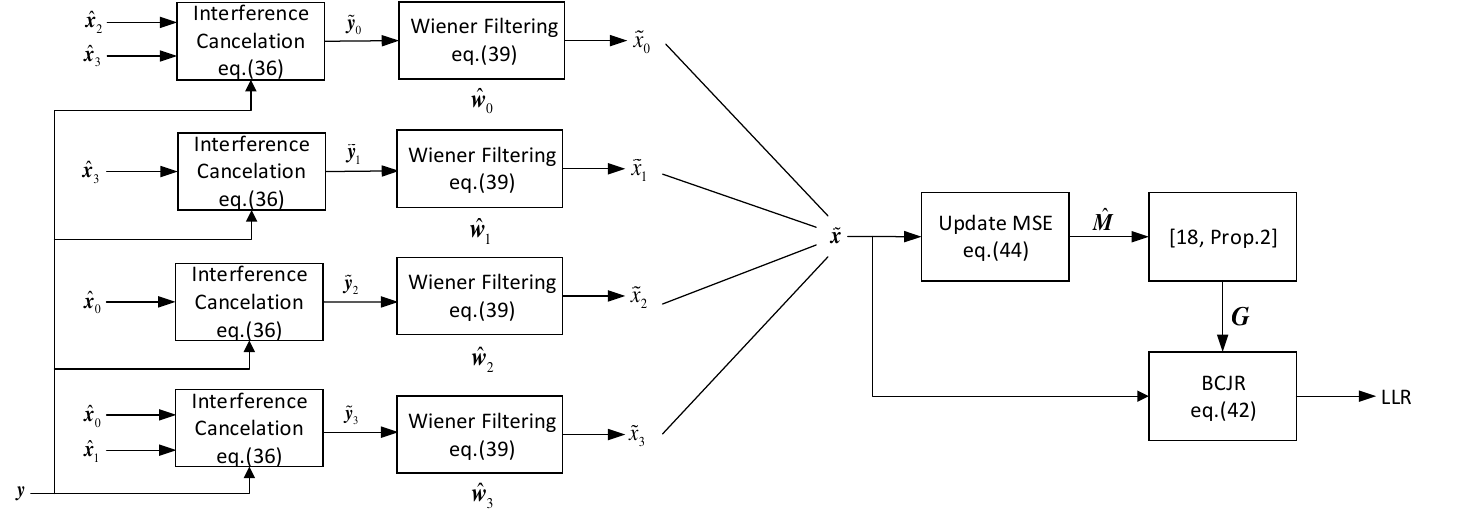}}
\vspace*{-12mm}
\caption{\label{p1fig4} An graphical overview of Method III with $K\!=\!4$ and $\nu\!=\!1$.}
\vspace*{-10mm}
\end{center}
\end{figure}

An graphical overview of Method III for $K\!=\!4$ and $\nu\!=\!1$ is illustrated in Fig.\,\ref{p1fig4}. For any $\vec{G}$ with memory size $\nu$, the IC matrix $(\vec{I}\!+\!\vec{G})\tilde{\vec{C}}$ is zero along the main diagonal, which guarantees that the extrinsic information will not be used for current symbols in the IC process. In GMI sense, Method III will not outperform Method II with a shape (b) $\vec{R}$, but it may outperform the GMI of Method II with a shape (c) $\vec{R}$, as it can be verified that a shape (c) $\vec{R}$ has zeros at the positions where $(\vec{I}\!+\!\vec{G})\hat{\vec{C}}$ are also zeros.

 \begin{remark}
As $\hat{\vec{W}}\vec{H}\!-\!\hat{\vec{C}}\!=\![\hat{\vec{W}}\vec{H}]_{\nu}$, by Lemma \ref{lem2} $(\vec{I}\!+\!\vec{G})(\hat{\vec{W}}\vec{H}-\hat{\vec{C}})$ is banded within diagonals $[-2\nu, 2\nu]$, which shows that, $[(\vec{I}\!+\!\vec{G})\hat{\vec{W}}\vec{H}]_{\backslash2\nu}\!=\![(\vec{I}\!+\!\vec{G})\hat{\vec{C}}]_{\backslash2\nu}$. Therefore, Proposition \ref{p1prop4} also holds for Method III with $\nu_{\mathrm{R}}\!=\!\nu$.
\end{remark}

\section{Parameter Optimization for ISI Channel}
\

In this section, we extend the CS demodulators to ISI channels where the matrix $\vec{P}\!=\!\alpha\vec{I}$ and the block length $K$ is infinitely large. The formulas for the achievable rates in (\ref{p1metricMIMO2}), (\ref{p1metricMIMO1}) and (\ref{p1metricMIMO3}) can be directly applied to (\ref{p1sm1}), but as the achievable rate $I_{\mathrm{GMI}}$ (as a function of the specified CS parameters) is then dependent on the block length $K$, we are interested in asymptotic rate
$$\bar{I} =\lim_{K\to \infty} \frac{1} {K}I_{\mathrm{GMI}}. $$

Ideally, in the ISI case the front-end matrix $\vec{V}$ and IC matrix $\vec{R}$ correspond to linear filtering operations and the filters are infinitely long, but in practice filters with finite tap lengths are used. Therefore, we analyze the properties of $\vec{V}$, $\vec{R}$ (and $\vec{W}$, $\vec{T}$) with a finite number of taps and approximate them by band-shaped Toeplitz matrices. Furthermore, the trellis representation matrix $\vec{G}$ (and $\vec{F}$), and channel matrix $\vec{H}$ are also band-shaped Toeplitz matrices. Therefore, in the ISI case all matrices we consider are assumed to be band-shaped Toeplitz matrices, and the band size can be arbitrary and sufficiently large so that we can analyze the asymptotic properties. In \cite{HM91} a complete theoretic machinery for ISI channels is derived and a result is that, as $K\!\to\! \infty$ the linear convolution in (\ref{p1sm1}) can be replaced with a circular convolution.

In the following, we denote the Fourier series associated to a band-shaped Toeplitz matrix $\vec{E}$ with infinitely large dimensions by $E(\omega)$, where $\vec{E}$ is constrained to be zero outside the middle $2N_{\mathrm{E}}\!+\!1$ diagonals, and $N_{\mathrm{E}}$ is referred to as the tap length of $E(\omega)$. The series $E(\omega)$ defined as $$E(\omega)\!=\!\sum\limits_{k=-N_{\mathrm{E}}}^{N_{\mathrm{E}}}\!e_k\exp\!\left(j k\omega \right)$$ is specified by a vector $\vec{e}\!=\![\!\begin{array}{ccccccc}e_{-N_{\mathrm{E}}}\;\ldots \; e_{-1}\; e_0\;  e_1\; \ldots \;e_{N_{\mathrm{E}}}\end{array}\!]$, where $e_{0}$ is the element on the main diagonal and $e_{k}$ is the element on $k$th lower ($k\!>\!0$) or upper ($k\!<\!0$) diagonal. As all quantities are evaluated as the block length $K$ grows large, $E(\omega)$ approaches the eigenvalue distribution of $\vec{E}$ (see \cite{Sze,MK} for a precise statement of this result). We first state Theorem \ref{p1thm3}, which is an asymptotic version of Theorem \ref{p1thm2} for ISI channels.

\begin{theorem}\label{p1thm3}
Assume that two band-shaped Toeplitz matrices $\vec{G}$ and $\vec{M}$ with infinitely large dimensions satisfying $[\vec{G}]_{\backslash\nu}\!=\!\vec{0}$, $\vec{I}\!+\!\vec{G}\!\succ\!0$ and $\vec{M}\!\prec\!0$. Define a scalar function
\bea \label{p1isibar} \bar{I}= 1\!+\!\frac{1}{2\pi}\! \int_{-\pi}^{\pi}\!\!\big( \log(1\!+\! G(\omega))\!+\!M(\omega)(1\!+\!G(\omega))\big)\mathrm{d}\omega. \eea
Then, the optimal $G(\omega)$ that maximizes $\bar{I}$ is
$$ G_{\mathrm{opt}}(\omega)\!=\!|u_0\!+\!\vec{\hat{u}}\varphi(\omega)|^2\!-\!1,$$
where the $1\!\times\!\nu$ vector $\varphi(\omega)\!=\![\!\begin{array}{cccc}\exp\!\left(j\omega \right)\; \exp\!\left(j2\omega \right)\;\!\ldots\! \; \exp\!\left(j\nu\omega \right)\end{array}\!]\rmt$, and
 {\setlength\arraycolsep{2pt}\bea u_0 &=&\frac{1}{\sqrt{\vec{\tau}_1\rmh\vec{\tau}_2^{-\!1}\vec{\tau}_1\!-\!\tau_0}}, \notag \\
  \label{p1optisim3}\vec{\hat{u}}&=&-u_0\vec{\tau}_1\rmh\vec{\tau}_2^{-\!1}.\eea}
\hspace{-1.4mm}The real scalar $\tau_0$, $\nu\!\times\!1$ vector $\vec{\tau}_1$, and $\nu\!\times\!\nu$ matrix $\vec{\tau}_2$ are defined as
{\setlength\arraycolsep{2pt} \bea
  \tau_0&=&\frac{1}{2\pi}\!\int_{-\pi}^{\pi}\!\!M(\omega)\mathrm{d}\omega, \notag \\
   \vec{\tau}_1&=&\frac{1}{2\pi}\!\int_{-\pi}^{\pi}\!\!M(\omega)\varphi(\omega)\mathrm{d}\omega, \notag \\
\vec{\tau}_2&=&\frac{1}{2\pi}\! \int_{-\pi}^{\pi}\!\!M(\omega)\varphi(\omega)\varphi(\omega)\rmh\mathrm{d}\omega. \; \eea}
\hspace{-1.4mm}Furthermore, with $G_{\mathrm{opt}}(\omega)$ the optimal $\bar{I}$ reads
\bea \label{p1optbarm3} \bar{I}=2\log(u_0).  \eea
\end{theorem}

\begin{proof}
As $\vec{I}\!+\!\vec{G}\!\succ\!0$, we assume that $1\!+\!G(\omega)\!=\!|U(\omega)|^2$, with $U(\omega)\!=\!u_0\!+\!\vec{\hat{u}}\varphi(\omega)$ and
$  \hat{\vec{u}}\!=\![\!\begin{array}{cccc}u_1\;u_2\;\!\ldots\! \;u_{\nu}\end{array}\!]$. Then $\bar{I}$ in (\ref{p1isibar}) can be rewritten as
\bea \label{p1isim3ibarnew}  \bar{I}\!=\! 1\!+\!2\log(u_0)\!+\!\frac{1}{2\pi}\!\int_{-\pi}^{\pi}\!M(\omega)\big(u_0^{2}\!+\!2\Re\{u_0\vec{\hat{u}}\varphi(\omega)\}\!+\!\vec{\hat{u}}\varphi(\omega)\varphi\rmh(\omega)\vec{\hat{u}}\rmh\big)\mathrm{d}\omega.  \eea
Taking the first order differentials with respect to $u_0$ and $\vec{\hat{u}}$ and optimizing them directly results in the optimal solution (\ref{p1optisim3}). Inserting (\ref{p1optisim3}) back into (\ref{p1isim3ibarnew}) and after some manipulations, the optimal asymptotic rate is then in (\ref{p1optbarm3}).
\end{proof}

\subsection{Method I}
The structures of $(\vec{W},\vec{T},\vec{F})$ are the same as in Section \ref{p1MIMO_M1}, except that now the matrices have infinite dimensions. Applying  Szeg\H{o}'s eigenvalue distribution theorem \cite{Sze} to (\ref{p1metricMIMO1}), the asymptotic rate reads
{\setlength\arraycolsep{2pt}\bea \label{p1ibarm1} \bar{I}\big(W(\omega),T(\omega),F(\omega)\big)&=&\lim_{K\to \infty} \frac{1} {K}I_{\mathrm{GMI}}(\vec{W},\vec{T},\vec{F}) \notag \\
&=&\frac{1}{2\pi}\! \int_{-\pi}^{\pi}\!\bigg( \log \!\big(1\!+\! |F(\omega)|^2\big)\!-\! |F(\omega)|^2 \!-\!\frac{L_1(\omega)}{1\!+\!|F(\omega)|^2} \bigg)\mathrm{d}\omega \nonumber \\
&&+\frac{1}{\pi}\!\int_{-\pi}^{\pi}\!\Re\!\big\{F^{\ast}(\omega)\big(W(\omega)H(\omega)\!-\!\alpha T(\omega)\big)\big\}\mathrm{d}\omega \quad
\eea}
\hspace{-1.4mm}where 
{\setlength\arraycolsep{2pt}\bea L_1(\omega)\!&=&\!|F(\omega)W(\omega)|^{2}\big(N_0\!+\!|H(\omega)|^2\big)\!+\!\alpha|F(\omega)T(\omega)|^2\notag \\ && \!-\!  2\alpha|F(\omega)|^{2}\!\Re\!\left\{H(\omega)W(\omega)T^{\ast}(\omega) \right\}.\notag \eea}
\hspace{-1.4mm}Note that, the Fourier series associated to $\vec{M}$ and $\tilde{\vec{M}}$ in (\ref{p1M}) and (\ref{p1Mt}) are
{\setlength\arraycolsep{2pt}\bea \label{p1mw} M(\omega)&=&\frac{|H(\omega)|^2}{N_0\!+\!|H(\omega)|^2}\!-\!1,  \\
 \label{p1mtw} \tilde{M}(\omega)&=&\alpha^{2}(M(\omega)+1)\!-\!\alpha.\eea}
\hspace{-1.4mm}Further, define a $(2N_{\mathrm{T}}\!-\!\nu)\!\times\!1$ vector
\be \label{p1phi} \phi(\omega)\!=\!\big[\!\!\begin{array}{cccccc}\exp\!\big(\!\!-\!\!jN_{\mathrm{T}}\omega \big)\; \ldots \; \exp\!\big(\!\!-\!\!j(\nu\!+\!1)\omega \big)\; \exp\!\big(j(\nu\!+\!1)\omega \big)\; \ldots \; \exp\!\big(jN_\mathrm{T}\omega \big)\end{array}\!\!\big]\rmt,\!\;\; \ee
a $(2N_{\mathrm{T}}\!-\!\nu)\!\times\!1$ vector $\vec{\varepsilon}_1$, and a $(2N_{\mathrm{T}}\!-\!\nu)\!\times\!(2N_{\mathrm{T}}\!-\!\nu)$ Hermitian matrix $\vec{\varepsilon}_2$ as
{\setlength\arraycolsep{2pt}\bea \label{p1vareps12} \vec{\varepsilon}_1&=&\frac{\alpha}{2\pi}\int_{-\pi}^{\pi}M(\omega)F^{\ast}(\omega)\phi(\omega)\mathrm{d}\omega, \notag \\
\vec{\varepsilon}_2&=&\frac{1}{2\pi} \!\int_{-\pi}^{\pi}\!\frac{\tilde{M}(\omega)|F(\omega)|^2\phi(\omega)\phi(\omega)\rmh}{1\!+\!|F(\omega)|^2}\mathrm{d}\omega, \eea}
\hspace{-1.4mm}where $N_{\mathrm{T}}$ is the tap length of $T(\omega)$, and $\nu\!+\!1$ is the band size where matrix $\vec{T}$ is constrained to zero. Then, we have Proposition \ref{p1prop5} with the proof\footnote{Proposition 5 is the same as \cite[Theorem 1]{HKHR17} which has been derived for hard feedback symbols. For completeness, we restate the proof in Appendix H.} given in Appendix H.

\begin{proposition} \label{p1prop5}
The optimal $W(\omega)$ for the asymptotic rate in (\ref{p1ibarm1}) is
\bea \label{p1optisiw} W_{\mathrm{opt}}(\omega)=\frac{H^{\ast}(\omega)}{F^{\ast}(\omega)(N_0\!+\!|H(\omega)|^2)}\big(1\!+\! |F(\omega)|^2\!+\!\alpha F^{\ast}(\omega)T_{\mathrm{opt}}(\omega)\big),\eea
and when $0\!<\!\alpha\!\leq\!1$, the optimal $T(\omega)$ reads
\bea \label{p1optisit} T_{\mathrm{opt}}(\omega)=-\vec{\varepsilon}_1\rmh\vec{\varepsilon}_2^{-\!1}\phi(\omega).\eea
With the optimal $W(\omega)$ and $T(\omega)$, the asymptotic rate equals
{\setlength\arraycolsep{2pt} \bea \label{p1ibarm1optwt}\bar{I}\big(W_{\mathrm{opt}}(\omega),T_{\mathrm{opt}}(\omega),F(\omega)\big)=\left\{\begin{array}{ll} \bar{I}_1(F(\omega)),& \alpha=0\\
\bar{I}_1(F(\omega))+\bar{\delta}_1(F(\omega)),&0<\alpha\leq1. \end{array}\right.\eea}
\hspace*{-1.4mm}The functions $\bar{I}_1(F(\omega))$ and $\bar{\delta}_1(F(\omega))$ are defined as\footnote{Similar to finite length linear vector channels, $\bar{\delta}_1(F(\omega))$ in (\ref{bardeta1}) is only defined for $\alpha\!\neq\!0$ which represents the rate increment with soft information. The same holds for $\bar{\delta}_2(G(\omega))$ in (\ref{p1bardeta2}) for Method II.},
{\setlength\arraycolsep{2pt} \bea \label{p1ibarm1optwt1} \bar{I}_1(F(\omega))&=& 1+\!\frac{1}{2\pi}\! \int_{-\pi}^{\pi}\!\!\Big( \!\log\!\big(1\!+\! |F(\omega)|^2\big)\!+\!M(\omega)\big(1\!+\!|F(\omega)|^2\big)\! \Big)\mathrm{d}\omega, \\
\label{bardeta1}\bar{\delta}_1(F(\omega)) &=&-\vec{\varepsilon}_1\rmh\vec{\varepsilon}_2^{\!-\!1}\vec{\varepsilon}_1.\eea}
\end{proposition}

In the ISI case, Method I is still not concave an example is also provided in Appendix C, and a gradient based optimization is used to optimize $F(\omega)$ with the optimal solution of $G_{\mathrm{opt}}(\omega)$ from Theorem \ref{p1thm3} is used to initialize the starting point. 

The connection between the optimal front-end filter $W(\omega)$ and the IC filter $T(\omega)$ in Proposition \ref{p1prop2} also holds for ISI channels. An asymptotic version of Proposition \ref{p1prop2} is stated in Proposition \ref{p1prop6}.
\begin{proposition} \label{p1prop6}
When $0\!<\!\alpha\!\leq\!1$, $a_k\!=\!b_k$ holds for $k\!<\!-(\nu\!+\!1)$, where
{\setlength\arraycolsep{2pt}\bea a_k&=&\frac{1}{2\pi}\! \int_{-\pi}^{\pi}F^{\ast}(\omega)W_{\mathrm{opt}}(\omega)H(\omega)\!\exp\big(\!\!-\!\!j k \omega\big)\mathrm{d}\omega \notag \\
 b_k&=&\frac{1}{2\pi}\! \int_{-\pi}^{\pi}F^{\ast}(\omega)T_{\mathrm{opt}}(\omega)\!\exp\big(\!\!-\!\!j k \omega\big)\mathrm{d}\omega.  \notag\eea}
\end{proposition}
\begin{proof}
In Appendix H, the optimal $\tilde{\vec{t}}$ in (\ref{twopt}) satisfies $\tilde{\vec{t}}_{\mathrm{opt}}\vec{\varepsilon}_2=-\vec{\varepsilon}_1\rmh$. With the definitions of $\vec{\varepsilon}_1$, $\vec{\varepsilon}_2$ in (\ref{p1vareps12}), this is equivalent to
\bea \label{p1ref} \frac{1}{2\pi} \!\int_{-\pi}^{\pi}\!\frac{\tilde{M}(\omega)|F(\omega)|^{2}T_{\mathrm{opt}}(\omega)\phi(\omega)\rmh}{1+|F(\omega)|^{2}}\mathrm{d}\omega =-\frac{\alpha}{2\pi}\! \int_{-\pi}^{\pi}\!F(\omega)M(\omega)\phi(\omega)\rmh\mathrm{d}\omega.\eea
On the other hand, with $W_{\mathrm{opt}}$ in (\ref{p1optisiw}) and $M(\omega)$, $\tilde{M}(\omega)$ defined in (\ref{p1mw}) and (\ref{p1mtw}), we have
\bea \label{p1w-t} &&\frac{1}{2\pi}\!\int_{-\pi}^{\pi}\!(F^{\ast}(\omega)W_{\mathrm{opt}}(\omega)H(\omega)\!-\!F^{\ast}(\omega)T_{\mathrm{opt}}\!-\!\big(1\!+\!|F(\omega)|^{2}\big)\big)\!\exp\!\big(\!\!-\!\!jk\omega\big)\mathrm{d}\omega \notag \\
 &&\quad =\frac{1}{2\pi}\!\int_{-\pi}^{\pi}\!\Big(\frac{\tilde{M}(\omega)F^{\ast}(\omega)T_{\mathrm{opt}}(\omega)}{\alpha}\!+\!\big(1\!+\!|F(\omega)|^{2}\big)M(\omega)\Big)\!\exp\!\big(\!\!-\!\!jk\omega\big)\mathrm{d}\omega. \quad\quad \eea
Transforming (\ref{p1ref}) and (\ref{p1w-t}) back into matrix forms, we have that (\ref{p1appe1}) and (\ref{p1appe3}) hold. Following the same arguments as in the proof of Proposition \ref{p1prop2}, $\vec{F}\rmh(\vec{W}_{\mathrm{opt}}\vec{H}\!-\!\vec{R}_{\mathrm{opt}})$ is banded within diagonals $[-\nu,K\!-\!1]$. Therefore we have
\bea  &&\frac{1}{2\pi}\!\int_{-\pi}^{\pi}\!\big(F^{\ast}(\omega)W_{\mathrm{opt}}(\omega)H(\omega)\!-\!\!F^{\ast}(\omega)T_{\mathrm{opt}}(\omega)\big)\!\exp\!\big(\!\!-\!\!jk\omega\big)\mathrm{d}\omega\!=\!0 \notag \eea
whenever $k\!<\!-(\nu\!+\!1)$, which proves Proposition \ref{p1prop6}.
\end{proof}

\subsection{Method II}
The matrices $(\vec{V},\vec{R},\vec{G})$ have the same constraints as in Section \ref{MIMO_M2} while the dimensions of these matrices are infinitely large. However, as the shape (a) of $\vec{R}$ in Fig.\,\ref{p1fig3} is not meaningful as $N, K\!\to\! \infty$, it is not considered for ISI case. Applying Szeg\H{o}'s eigenvalue distribution theorem to (\ref{p1metricMIMO2}), the asymptotic rate of Method II reads
{\setlength\arraycolsep{2pt}\bea \label{p1ibarm2} \bar{I}(V(\omega),R(\omega),G(\omega))&=&\lim_{K\to \infty} \frac{1} {K}I_{\mathrm{GMI}}(\vec{V},\vec{R},\vec{G}) \notag \\
&=&\frac{1}{2\pi}\! \int_{-\pi}^{\pi}\!\!\bigg(\! \log\! \big(1\!+ \!G(\omega)\big)\!-\! G(\omega) \!-\!\frac{L_2(\omega)}{1\!+\!G(\omega)}\bigg)\mathrm{d}\omega \nonumber \\
&&+\frac{1}{\pi}\!\int_{-\pi}^{\pi}\!\Re\big\{\big(V(\omega)H(\omega)\!-\!\alpha R(\omega)\big)\big\}\mathrm{d}\omega \quad\quad \quad\quad \quad \quad  \eea}
\hspace*{-1.4mm}where $$L_2(\omega)=|V(\omega)|^2\big(N_0\!+\!|H(\omega)|^2\big) \!+\!\alpha|R(\omega)|^2 \!-\!  2\alpha\Re\big\{H(\omega)V(\omega)R^{\ast}(\omega) \big\}.$$

Define a $2(N_{\mathrm{R}}\!-\!\nu_{\mathrm{R}})\!\times\!1$ vector
\be \label{p1lpsi} \psi(\omega)\!=\!\big[\!\!\begin{array}{cccccc}\exp\!\big(\!\!-\!\!jN_{\mathrm{R}}\omega \big)\; \!\ldots\! \; \exp\!\big(\!\!-\!\!j(\nu_{\mathrm{R}}\!+\!1)\omega \big)\; \exp\!\big(j (\nu_{\mathrm{R}}\!+\!1)\omega \big)\; \!\ldots\! \; \exp\!\big(jN_{\mathrm{R}}\omega \big)\end{array}\!\!\big]\rmt\!,\, \,\ee
a $2(N_{\mathrm{R}}\!-\!\nu_{\mathrm{R}})\!\times\!1$ vector $\vec{\zeta}_1$, and a $2(N_{\mathrm{R}}\!-\!\nu_{\mathrm{R}})\!\times\!2(N_{\mathrm{R}}\!-\!\nu_{\mathrm{R}})$ Hermitian matrix $\vec{\zeta}_2$ as
{\setlength\arraycolsep{2pt}\bea \label{p1zeta12} \vec{\zeta}_1&=&\frac{\alpha}{2\pi}\!\int_{-\pi}^{\pi}\!M(\omega)\psi(\omega)\mathrm{d}\omega, \notag \\
\vec{\zeta}_2&=&\frac{1}{2\pi}\!\int_{-\pi}^{\pi}\!\frac{\tilde{M}(\omega)\psi(\omega)\psi(\omega)\rmh}{1\!+\!G(\omega)}\mathrm{d}\omega
,\eea}
\hspace*{-1.4mm}where $N_{\mathrm{R}}$ denotes the tap length of $R_{\mathrm{opt}}(\omega)$, $2\nu_{\mathrm{R}}\!+\!1$ is the band size where $\vec{R}$ is constrained to zero, and $M(\omega)$ and $\tilde{M}(\omega)$ are in (\ref{p1mw}) and (\ref{p1mtw}). Then, we have Proposition \ref{p1prop6} with the proof in Appendix I, where we also show that $R(\omega)$ is real and $\vec{R}$ has Hermitian symmetry.
\begin{proposition} \label{p1prop7}
The optimal $V(\omega)$ for (\ref{p1ibarm2}) is,
\bea  \label{p1optisiv} V_{\mathrm{opt}}(\omega)=\frac{H^{\ast}(\omega)}{N_0\!+\!|H(\omega)|^2}\big(1\!+\!G(\omega)\!+\!\alpha R_{\mathrm{opt}}(\omega)\big),\eea
and when $0\!<\!\alpha\!\leq\!1$, the optimal $R(\omega)$ reads
\bea \label{p1optisir} R_{\mathrm{opt}}(\omega)=-\vec{\zeta}_1\rmh\vec{\zeta}_2^{-\!1}\psi(\omega).\eea
With the optimal $V(\omega)$ and $R(\omega)$, the asymptotic rate equals
{\setlength\arraycolsep{2pt} \bea \label{p1ibarm2optvr}\bar{I}\big(V_{\mathrm{opt}}(\omega),R_{\mathrm{opt}}(\omega),G(\omega)\big)=\left\{\begin{array}{ll} \bar{I}_2(G(\omega)),& \alpha=0\\
\bar{I}_2(G(\omega))+\bar{\delta}_2(G(\omega)),&0<\alpha\leq1. \end{array}\right.\eea}
\hspace*{-1.4mm}The functions $\bar{I}_1(G(\omega))$ and $\bar{\delta}_2(G(\omega))$ are defined as,
{\setlength\arraycolsep{2pt} \bea \label{p1ibarm2optvr1} \bar{I}_2(G(\omega))&=& 1\!+\!\frac{1}{2\pi}\! \int_{-\pi}^{\pi}\!\Big(\!\log \!\big(\!1\!+\! G(\omega)\big)\!+\!M(\omega)\big(1\!+\!G(\omega)\big)\!\Big)\mathrm{d}\omega, \\
\label{p1bardeta2}\bar{\delta}_2(G(\omega)) &=&-\vec{\zeta}_1\rmh\vec{\zeta}_2^{\!-\!1}\vec{\zeta}_1.\eea}
\end{proposition}
For $0\!<\!\alpha\!\leq\!1$, it still needs a gradient based optimization to find the optimal $G(\omega)$ for (\ref{p1ibarm2optvr1}), and the closed-form solution in Theorem \ref{p1thm3} is utilized as the starting point. The asymptotic rate $\bar{I}(V_{\mathrm{opt}}(\omega),R_{\mathrm{opt}}(\omega),G(\omega))$ is also concave with respect to $G(\omega)$, which is shown in Appendix J.

\begin{proposition} \label{p1prop8}
When $0\!<\!\alpha\!\leq\!1$, $a_k\!=\!b_k$ holds for $|k|\!>\!\nu\!+\!\nu_{\mathrm{R}}$, where
{\setlength\arraycolsep{2pt}\bea \label{p1ak} a_k&=&\frac{1}{2\pi}\! \int_{-\pi}^{\pi}V_{\mathrm{opt}}(\omega)H(\omega)\!\exp\big(\!\!-\!\!j k \omega\big)\mathrm{d}\omega,  \\
\label{p1bk} b_k&=&\frac{1}{2\pi}\! \int_{-\pi}^{\pi}R_{\mathrm{opt}}(\omega)\!\exp\big(\!\!-\!\!j k \omega\big)\mathrm{d}\omega.  \eea}
\end{proposition}
The connection of the optimal $V(\omega)$ and $R(\omega)$ stated in Proposition \ref{p1prop8} is an asymptotic version of Proposition \ref{p1prop4}, and the proof follows the similar approach as Proposition 6. We show an example in Fig.\,\ref{p1fig5} to illustrate Proposition \ref{p1prop8} with Method II and $\nu\!=\!\nu_{\mathrm{R}}\!=\! 1$. The Proakis-C \cite{Prok} channel is tested at an SNR of 10 dB and $\alpha$ equals $0.1$, $0.4$ and $0.8$, respectively. As $\nu_{\mathrm{R}}\!=\! 1$, $b_k$ as defined in (\ref{p1bk}) is constrained to zero for $0\!\leq\!k\!\leq\!1$. As can be seen, $a_k$ as defined in (\ref{p1ak}) equals $b_k$ only for $|k|\!>\!2$, and when $|k|\!=\!2$, $a_k$ and $b_k$ are not identical. This shows that with the optimal $V(\omega)$ and $R(\omega)$, the signal part along the second upper and lower diagonals that is not considered in $G(\omega)$ shall not be perfectly canceled out. This behavior cannot be seen in \cite{TSK00} which treats LMMSE-PIC for ISI channels, due to $\nu\!=\!\nu_{\mathrm{R}}\!=\!0$.

\begin{figure}[t]
\begin{center}
\vspace*{-6mm}
\hspace*{-2mm}
\scalebox{.45}{\includegraphics{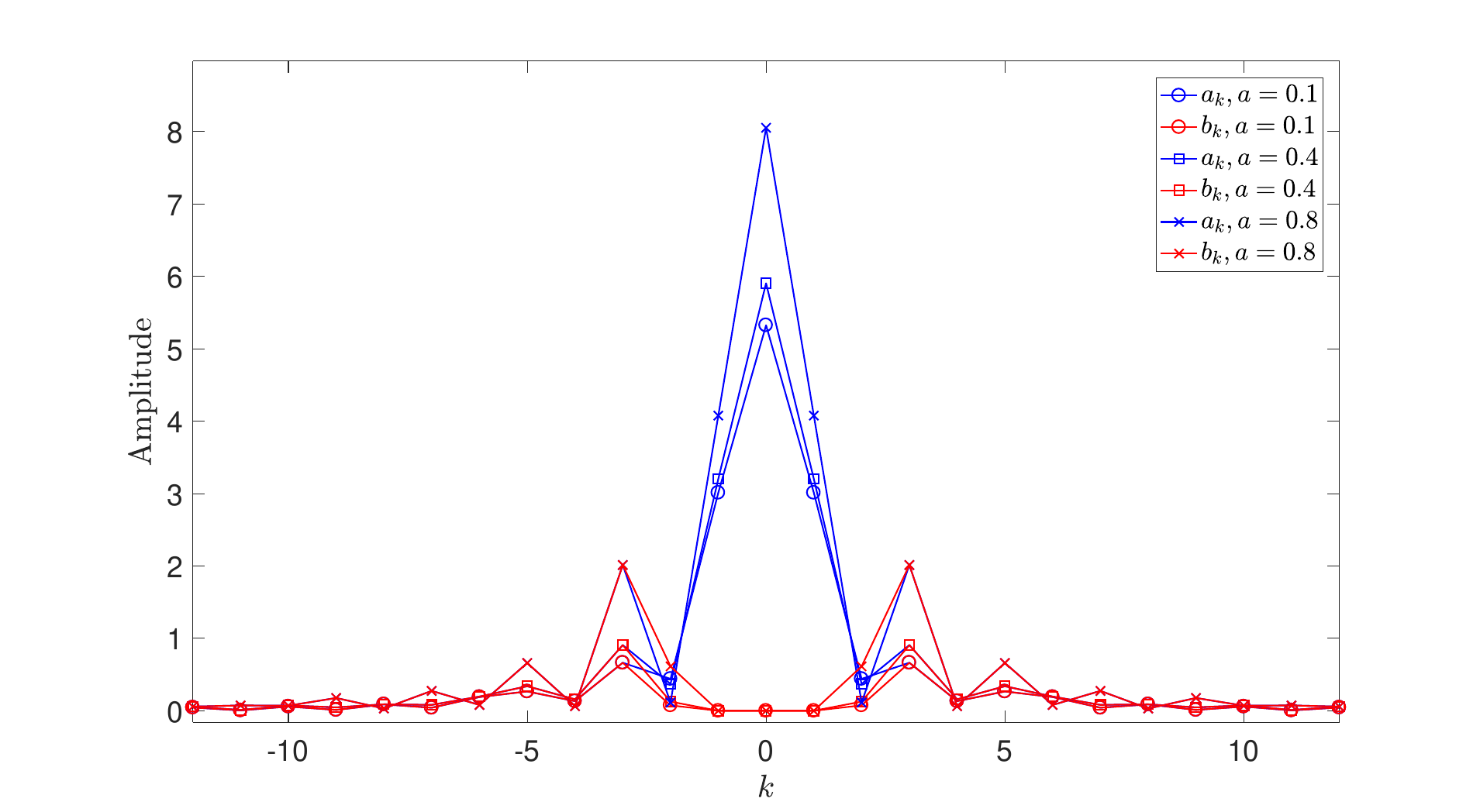}}
\vspace*{-7mm}
\caption{\label{p1fig5} Comparison between $a_k$ and $b_k$ for Method II under Proakis-C channel $\vec{h}\!=\![\,0.227\; 0.46\;0.688\; 0.46\;0.227\,]$.}
\vspace*{-10mm}
\end{center}
\end{figure}

\subsection{Method III}
In Method III, from (\ref{p1metricMIMO3}) the asymptotic rate reads
\be \bar{I}(G(\omega))=\lim_{K\to \infty} \frac{1} {K}I_{\mathrm{GMI}}(\vec{G})
= 1\!+\!\frac{1}{2\pi}\! \int_{-\pi}^{\pi}\!\!\big( \log\! \big(1\!+\! G(\omega)\big)\!+\!\hat{M}(\omega)(1\!+\!G(\omega) )\big)\mathrm{d}\omega\;\;  \ee
where according to (\ref{p1mpic}), 
{\setlength\arraycolsep{2pt}\bea  \hat{M}(\omega)&=&2\Re\!\big\{\alpha\hat{W}(\omega)H(\omega)\hat{C}^{\ast}(\omega)\!+\!\hat{W}(\omega)H(\omega)\notag \\ &&-\alpha\hat{C}^{\ast}(\omega)\big\}\!-\!\frac{|\hat{W}(\omega)|^2}{N_0\!+\!|H(\omega)|^2}\!-\!\alpha|\hat{C}(\omega)|^2\!-\!1. \notag \eea}
\hspace{-1.4mm}Replacing $M(\omega)$ by $\hat{M}(\omega)$, the optimal $G(\omega)$ and asymptotic rate $\bar{I}$ follow from Theorem \ref{p1thm3}.

\begin{remark}
Proposition \ref{p1prop8} also holds for Method III with $\nu_{\mathrm{R}}\!=\!\nu$, due to the fact that $[(\vec{I}\!+\!\vec{G})(\hat{\vec{W}}\vec{H}-\hat{\vec{C}})]_{\backslash2\nu}\!=\!\vec{0}$.
\end{remark}

\section{SNR Asymptotics}

In this section, we analyze asymptotic properties of the CS demodulators, and show that, as $N_0$ goes to $0$ and $\infty$, Method III and Method II are asymptotically equivalent. As Method I is inferior to Method II in GMI sense, we limit our investigations to Method II and Method III, and start the analysis for finite length linear vector channels first. The following limits can be verified straightforwardly:
{\setlength\arraycolsep{2pt} \bea \label{p1limitM} &&\lim_{N_0\to 0}\vec{M}/N_0=-(\vec{H}\rmh\vec{H})^{-1}, \notag\\
 &&\lim_{N_0\to \infty}N_0(\vec{I}\!+\!\vec{M})=\vec{H}\rmh\vec{H}.  \eea}
\hspace{-1.4mm}Moreover, it also holds that
{\setlength\arraycolsep{2pt} \bea && \lim_{N_0\to 0}\tilde{\vec{M}}=\vec{P}^2-\vec{P}, \notag \\
&& \label{p1limitMt} \lim_{N_0\to \infty}\tilde{\vec{M}}=-\vec{P}.  \eea}
\hspace{-1.4mm}As $\tilde{\vec{M}}$ should be invertible from the definition of $\delta_2(\vec{G})$ in (\ref{p1deta2}), we restrict that $\vec{P}\!\prec\!\vec{I}$.

\begin{lemma} \label{p1lem3}
When $N_0\!\to\!0$ and $\infty$, the optimal $\vec{G}$ for (\ref{p1Ivr}) in Method II  satisfies (\ref{p1optcond}), and the following limits hold,
 {\setlength\arraycolsep{2pt}  \bea \label{p1limit1} \lim_{N_0\to 0}\left[(N_0(\vec{I}\!+\!\vec{G}_{\mathrm{opt}}))^{-1}\right]_{\nu}&=&[(\vec{H}\rmh\vec{H})^{-1}]_{\nu},  \\
 \label{p1limit2} \lim_{N_0\to \infty}\left[N_0\vec{G}_{\mathrm{opt}}\right]_{\nu}&=&[\vec{H}\rmh\vec{H}]_{\nu}.  \eea}
\end{lemma}

\begin{proof}
When $\vec{P}\!=\!\vec{0}$, from Theorem \ref{p1thm2} the optimal $\vec{G}$ for (\ref{p1Ivr}) satisfies (\ref{p1optcond}). From (\ref{p1limitM}), when $N_0\!\to\!0$, $\vec{M}\!\to\!\vec{0}$ and $N_0\!\to\!\infty$, $\vec{M}\!\to\!-\vec{I}$. Therefore, by the definition of $\vec{\Omega}$,
 $$\lim_{N_0\to 0,\infty}\!\!\vec{d}=\vec{\Omega}\mathrm{vec}(\vec{MP})=\vec{0}.$$ This implies that the gradient $ d_{\vec{G}}(\delta_2)$ in (\ref{p1part2}) (Appendix F) converges to zero. Hence the differentials of $I_{\mathrm{GMI}}(\vec{V}_{\mathrm{opt}},\vec{R}_{\mathrm{opt}},\vec{G})$ in (\ref{p1Ivr}) with $\vec{P}\!\neq\!\vec{0}$ converges to the differentials with $\vec{P}\!=\!\vec{0}$. From (\ref{p1optcond}) and (\ref{p1limitM}), the limit (\ref{p1limit1}) follows.
 
Next, since
{\setlength\arraycolsep{2pt} \bea \label{p1limit4} \lim_{N_0\to \infty}\left[N_0\left(\vec{I}\!-\!(\vec{I}\!+\!\vec{G}_{\mathrm{opt}})^{-1}\right)\right]_{\nu}=\lim_{N_0\to \infty}[N_0(\vec{I}\!+\!\vec{M})]_{\nu}=[\vec{H}\rmh\vec{H}]_{\nu}, \eea}
\hspace{-1.4mm}it shows that $\vec{I}\!-\!(\vec{I}+\vec{G}_{\mathrm{opt}})^{-1}\!\to\!\vec{0}$\footnote{A matrix $\vec{A}\!\to\!\vec{B}$ or a vector $\vec{a}\!\to\!\vec{b}$ means the nonzero elements of $\vec{A}\!-\!\vec{B}$ or $\vec{a}\!-\!\vec{b}$ converges to zero.} as $N_0\to\infty$. By the matrix inversion lemma, $\vec{I}-(\vec{I}+\vec{G}_{\mathrm{opt}})^{-1}\!\to\!\vec{G}_{\mathrm{opt}}$ as $N_0\to\infty$, and combining this with (\ref{p1limit4}) proves the limit (\ref{p1limit2}).
\end{proof}

\begin{lemma} \label{p1lem4}
In Method II, with the optimal $\vec{G}$, when $N_0\!\to\!0$ the GMI increment $\delta_2(\vec{G})$ in (\ref{p1deta2}) converges to zero with speed $\mathcal{O}(1/N_0)$\footnote{Two scalars $A$ and $B$ as functions of a variable $n$ converging to each other with speed $\mathcal{O}(n)$ means that, there exists a constant $C$ such that $\lim\limits_{n\to\infty}n|A-B|\!<\!C$.} and when $N_0\!\to\!\infty$ the GMI increment $\delta_2(\vec{G})$ converges to zero with speed $\mathcal{O}(N_0^2)$.
\end{lemma}

\begin{proof}
As $N_0\!\to\!0$, from (\ref{p1limitM}) we have $$\lim_{N_0\to 0}\vec{d}/N_0=\lim_{N_0\to 0}\vec{\Omega}\mathrm{vec}(\vec{MP}/N_0)=-\vec{\Omega}\mathrm{vec}\big((\vec{H}\rmh\vec{H})^{-1}\vec{P}\big).$$
Based on (\ref{p1limitMt}) and Lemma \ref{p1lem3}, the below equalities hold,
$$\delta_2(\vec{G}_{\mathrm{opt}})=N_0\frac{\vec{d}\rmh}{N_0}\Big(\vec{\Omega}\Big(\tilde{\vec{M}}^{\ast}\!\otimes\!\frac{ (\vec{I}\!+\!\vec{G}_{\mathrm{opt}})^{-\!1}}{N_0}\Big)\vec{\Omega}\rmt\Big)^{-\!1}\!\frac{\vec{d}}{N_0}=\mathcal{O}(N_0).$$
On the other hand, as $N_0\!\to\!\infty$, and by the definition of $\vec{\Omega}$, from (\ref{p1limitM}) we also have
 $$\lim_{N_0\to \infty}\!\!N_0\vec{d}=\!\lim_{N_0\to \infty}\!\!\vec{\Omega}\mathrm{vec}(N_0\vec{MP})\!=\!\lim_{N_0\to \infty}\!\!\vec{\Omega}\mathrm{vec}\big(N_0(\vec{I}\!+\!\vec{M})\vec{P}\big)\!=\!\vec{\Omega}\mathrm{vec}(\vec{H}\rmh\vec{H}\vec{P}).$$
Again utilizing (\ref{p1limitMt}) and Lemma \ref{p1lem3}, the below equalities hold,
\bea \delta_2(\vec{G}_{\mathrm{opt}})=\frac{1}{N_0^2}(N_0\vec{d}\rmh)\big(\vec{\Omega}\big(\tilde{\vec{M}}^{\ast}\!\otimes\! (\vec{I}\!+\!\vec{G}_{\mathrm{opt}})^{-\!1}\big)\vec{\Omega}\rmt\big)^{-\!1}\!(N_0\vec{d})=\mathcal{O}( 1/N_0^2). \notag\eea
Therefore, Lemma \ref{p1lem4} holds.
 \end{proof}

\begin{lemma}\label{p1lem5}
When $N_0\!\to\!0$ and $\infty$, the optimal GMI in Method III is independent of $\vec{P}$ and converges to the optimal GMI with $\vec{P}\!=\!\vec{0}$. Moreover, (\ref{p1limit1}) and (\ref{p1limit2}) hold.
\end{lemma}
The proof is given in Appendix K. Combining Lemmas \ref{p1lem3}-\ref{p1lem5}, and using the fact that Method III and Method II are equivalent with $\vec{P}\!=\!\vec{0}$, we have the following Theorem \ref{thm4}.

\begin{theorem} \label{thm4}
Assume that $\vec{P}\!\prec\!\vec{I}$, when $N_0\!\to\!0$ and $\infty$, the optimal GMI in Method III converges to the optimal GMI in Method III with $\vec{P}\!=\!\vec{0}$. Moreover, the optimal GMI in Method II also converges to the optimal GMI in Method III with $\vec{P}\!=\!\vec{0}$, with speed $\mathcal{O}(1/N_0)$ when SNR increase and $\mathcal{O}(N_0^2)$ when SNR decreases. The optimal $\vec{G}$ for both methods has the limits (\ref{p1limit1}) and (\ref{p1limit2}).
\end{theorem}

From Theorem \ref{thm4} we know that, except for the case where one of the elements in the diagonal matrix $\vec{P}$ is 1, the soft feedback information becomes asymptotically insignificant for the design of the CS parameters. The reason is that, when $N_0\!\to\! 0$, $\hat{\vec{x}}$ is overwhelmed by the noise, while when $N_0\!\to\! \infty$, the optimal front-end filter will null out $\hat{\vec{x}}$ since the filter can perfectly reconstruct the transmitted symbols without using the side information.

\begin{remark}
When $N_0\!\to\! 0$ and $\infty$, the optimal CS demodulator is the EZF demodulator defined in Example \ref{p1exam1}, and the TMF defined in Example \ref{p1exam2}, respectively.
\end{remark}

With ISI channels, as the same constraint $\vec{P}=\alpha\vec{I}\!\prec\!\vec{I}$ shall hold, we make the restriction that $\alpha\!<\!1$. The asymptotic properties for ISI channels are presented in Corollary \ref{p1cor1}, which is an asymptotic version of Theorem \ref{thm4} when the channel matrix $\vec{H}$ and CS parameters are band-shaped Toeplitz matrices with infinite dimensions. The detailed proof is following the same analysis as for the finite linear vector channels and omitted.

\begin{corollary} \label{p1cor1}
Assume that $0\!\leq\!\alpha\!<\!1$, when $N_0\!\to\!0$ and $\infty$, the optimal GMI in Method III converges to the optimal GMI in Method III with $\alpha\!=\!0$. Moreover, the optimal GMI in Method II also converges to the optimal GMI in Method III with $\alpha\!=\!0$, with speed $\mathcal{O}(1/N_0)$ when SNR increase and $\mathcal{O}(N_0^2)$ when SNR decreases. The optimal $\vec{G}$ for both methods has the following asymptotic properties hold for $|k|\!\leq\!\nu$:
{\setlength\arraycolsep{2pt} \bea  \lim_{N_0\to 0}\int_{-\pi}^{\pi}\frac{1}{N_{0}(1+G_{\mathrm{opt}}(\omega))}\exp(-jk\omega)\mathrm{d}\omega&=&\int_{-\pi}^{\pi}\frac{1}{|H(\omega)|^2}\exp(-jk\omega)\mathrm{d}\omega, \notag \\
  \lim_{N_0\to \infty}\int_{-\pi}^{\pi}N_{0}G_{\mathrm{opt}}(\omega)\exp(-jk\omega)\mathrm{d}\omega&=&\int_{-\pi}^{\pi}|H(\omega)|^2\exp(-jk\omega)\mathrm{d}\omega. \notag\eea}
\end{corollary}

\section{Empirical Results}

In this section, we provide empirical results to show the behaviors of CS demodulators in an iterative detection and decoding receiver designs. With the considered MIMO channels, all channel elements are assumed to be independent identically distributed (IID) complex Gaussian with zero-means, and the received signal power at each receive antenna is normalized to unity. For ISI case, we test with the typical Proakis-C channel as in Fig.\,\ref{p1fig5}.

\subsection{GMI Evaluation }
We first evaluate the GMI under $5\!\times\!5$ MIMO channels with memory size $\nu\!=\!1$ for all CS demodulators. We simulate $10000$ channel realizations for each SNR point. The GMIs are compared with that of the static CS demodulator \cite{RP12}, which is equivalent to Method II with $\vec{P}\!=\!\vec{0}$. The channel capacity is also presented for comparison. The results of GMI are shown in Fig.\,\ref{p1fig6}. As the quality of soft information improves beyond $\vec{P}\!=\!\vec{0}$, Method II with $\nR\!=\!0$ performs the best among all CS demodulators, as it has the most degrees of freedom (DoF) in $\vec{R}$. Method II with $\nR\!=\!\nu$ is the worst among Method I and Method II, while Method I is slightly worse than Method II with $\vec{R}$ of shape (a) in Fig.\,\ref{p1fig3}, which is because although the IC matrix $\vec{R}$ is shape (a) in both cases, $\vec{R}$ in Method II is more general than in Method I which is constrained to $\vec{R}\!=\!\vec{F}\rmh\vec{T}$. The GMI of Method III is inferior to Method II as expected. 

The results show consistent GMI increments for all CS demodulators when the feedback quality improves. When $\vec{P}$ increases from $\vec{P}\!=\!\vec{0}$ to the ideal case $\vec{P}\!=\!\vec{I}$, the channel capacity becomes inferior to the GMI as the pair $(\vec{x}, \hat{\vec{y}})$ is superior to $(\vec{x},\vec{y})$ for information transfer.

\begin{figure}
\begin{center}
\vspace*{-6mm}
\hspace*{-18mm}
\scalebox{.55}{\includegraphics{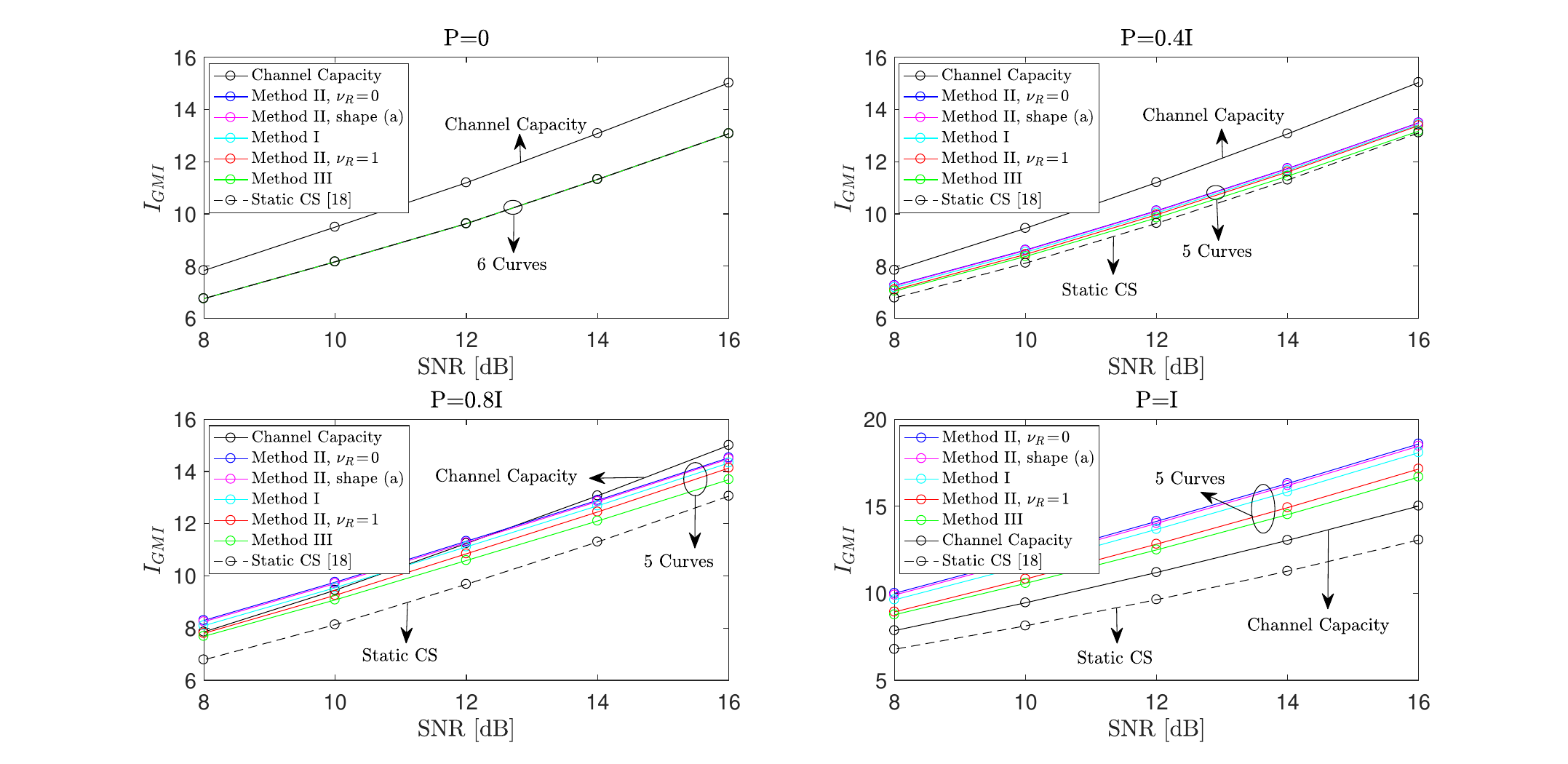}}  
\vspace*{-16mm}
\caption{\label{p1fig6} GMI of CS demodulators under $5\!\times\!5$ MIMO IID complex Gaussian channels with $\nu\!=\!1$.}
\vspace*{-10mm}
\end{center}
\end{figure}

\subsection{SNR Asymptotic of the GMI}
Next, we evaluate the asymptotic properties of the GMI described in Theorem \ref{thm4} under $5\!\times\!5$ MIMO channels. As shown in Fig.\,\ref{p1fig7}, the GMIs of Method II and Method III both converge to Method III with $\vec{P}\!=\!\vec{0}$. Moreover, the GMI of the CS demodulators converge to the EZF in Example 1 at high SNR, and the TMF in Example 2 at low SNR, respectively, which are well aligned with Theorem \ref{thm4}.

\subsection{EXIT Charts of CS Demodulators}

In order to predict the dynamics of iterative receivers, we use the tool of extrinsic information transfer (EXIT) charts invented by ten Brink \cite{S99, S01} for analysis of iterative receiver behavior. For EXIT analysis, the CS demodulator and the decoder measure the output extrinsic information $I_E$ based on a sequence of observations $\vec{y}$ and \textit{a priori} information $I_A$ into a new sequences.

In Fig.\,\ref{p1fig8}, we evaluate the EXIT charts for CS demodulators under $4\times6$ MIMO channels with $\nu\!=\!2$ for $\vec{F}$ and $\vec{G}$ at an SNR of 10dB. With Method II, we test different values of $\nR$. As can be seen, when $\nR\!>\!\nu$, the demodulation performance is inferior to $\nR\!\leq\!\nu$. This is because, the interference outside memory size $\nu$ and inside memory size $\nR$ is neither considered in the IC process nor in the BCJR module. Moreover, with $\nR\!\leq\!\nu$, the CS demodulators with Method II performs quite close to each other as well as Method I and III. For Method II with $\nR\!<\!\nu$, the interference inside memory size $\nu$ and outside $\nR$ are considered both in the IC and BCJR processes. However, an interesting observation is that, with large \textit{a priori} input $I_A$, Method II with $\nR\!=\!0$ is inferior to $\nR\!=\!1$ and $\nR\!=\!2$. Therefore, a conservative approach with Method II is to set $\nR\!=\!\nu$ such that the interference is either removed in IC process or dealt with in the BCJR module, to get rid of potential error propagation caused by redundant processings of the same part of interference.

In Fig.\,\ref{p1fig9}, we show the iterative detection and decoding trajectories for CS demodulators under Proakis-C channel with $\nu\!=\!2$ and at an SNR of 10dB. We use an [7, 5] convolutional code \cite{TS11} with a coded block-length $K\!=\!2004$, and a random permutation is applied to the coded bits. As can be seen, the CS demodulators with Method II and Method III are superior to the LMMSE-PIC demodulator, and the iterative detection and decoding trajectories are well aligned with the measured EXIT charts.

\begin{figure}
\begin{center}
\vspace*{-6mm}
\hspace*{-2mm}
\scalebox{.45}{\includegraphics{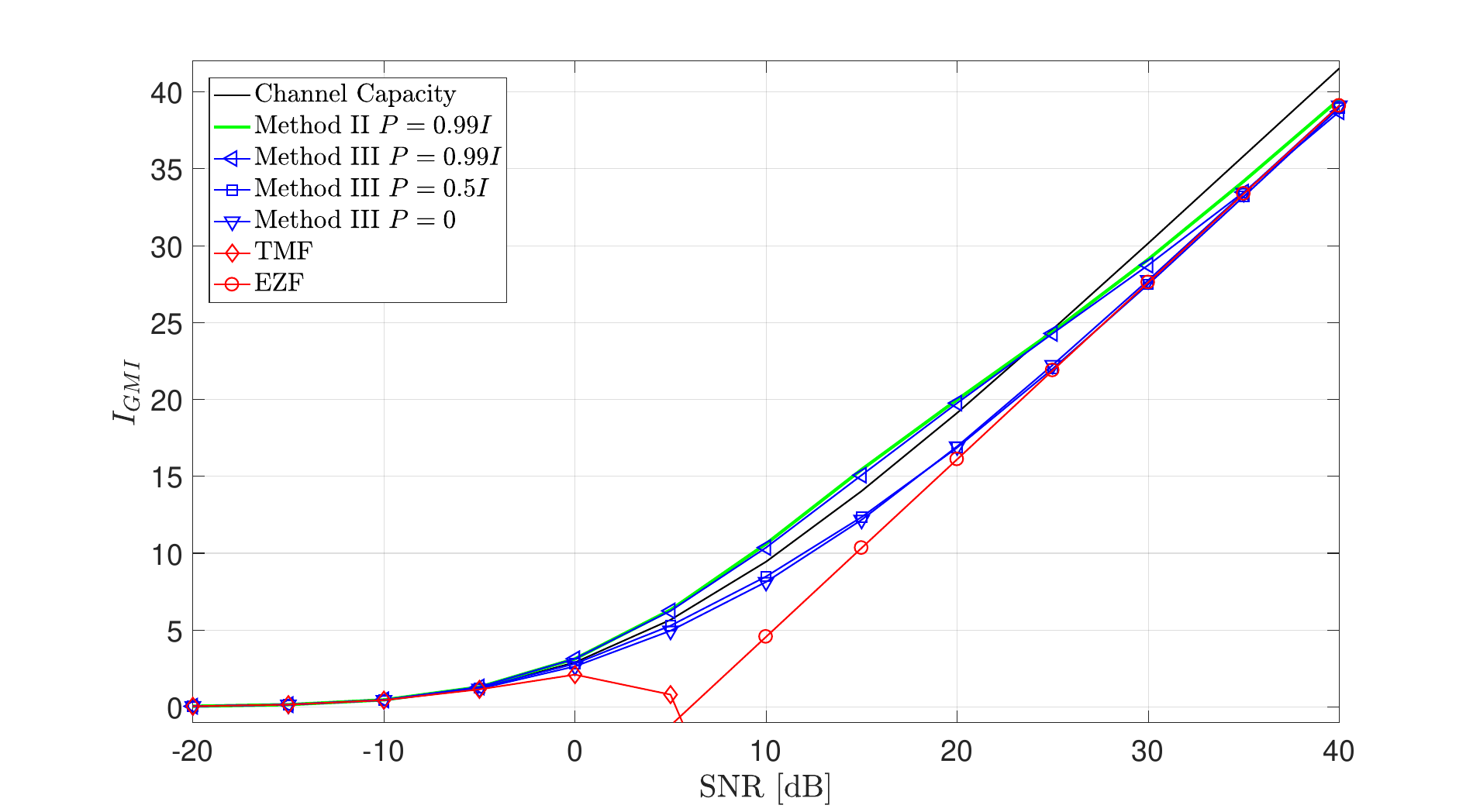}}  
\vspace*{-7mm}
\caption{\label{p1fig7} SNR asymptotic under $5\!\times\!5$ MIMO IID complex Gaussian channels with $\nu\!=\!1$.}
\vspace*{-10mm}
\end{center}
\end{figure}

\subsection{Link Performance}

We next turn to link-level simulations with turbo codes \cite{Tdec} where the outer decoder uses 8 internal iterations. A single code-block over all transmit symbols is used. At each SNR point $20000$ data blocks are simulated and the block-error-ratio (BLER) is measured. In all simulations, at most three global iterations are used between the demodulators, and the decoder the tap length of the front-end and IC filters are all set to $8L$, and $\nR\!=\!\nu$ for Method II.

\begin{figure} [t]
\begin{center}
\vspace*{-6mm}
\hspace*{-2mm}
\scalebox{.45}{\includegraphics{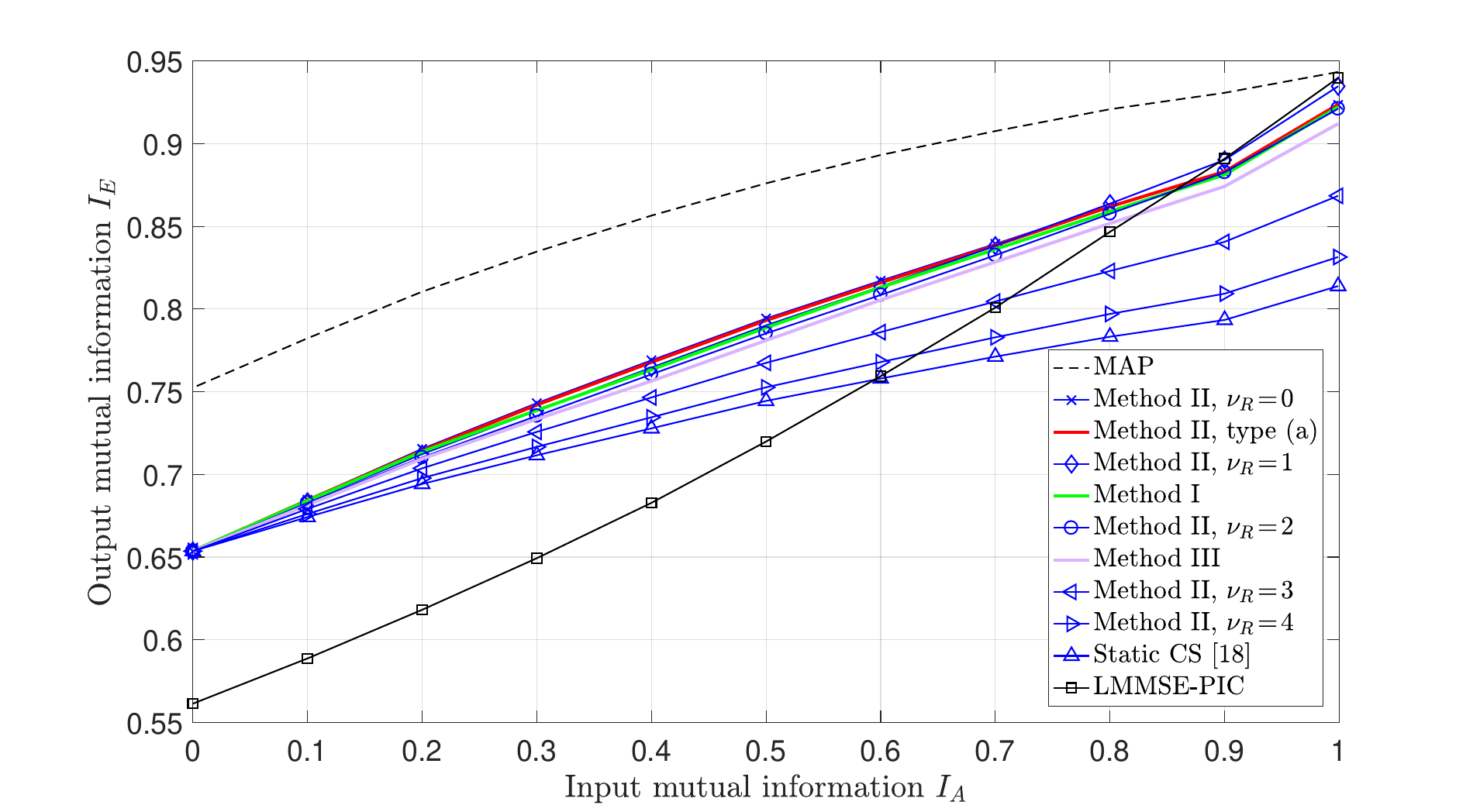}} 
\vspace*{-7mm}
\caption{\label{p1fig8} EXIT charts under $4\!\times\!6$ IID complex Gaussian MIMO channels with $\nu\!=\!2$ and different values of $\nR$.}
\vspace*{-10mm}
\end{center}
\end{figure}

\begin{figure}
\begin{center}
\vspace*{-0mm}
\hspace*{-2mm}
\scalebox{.45}{\includegraphics{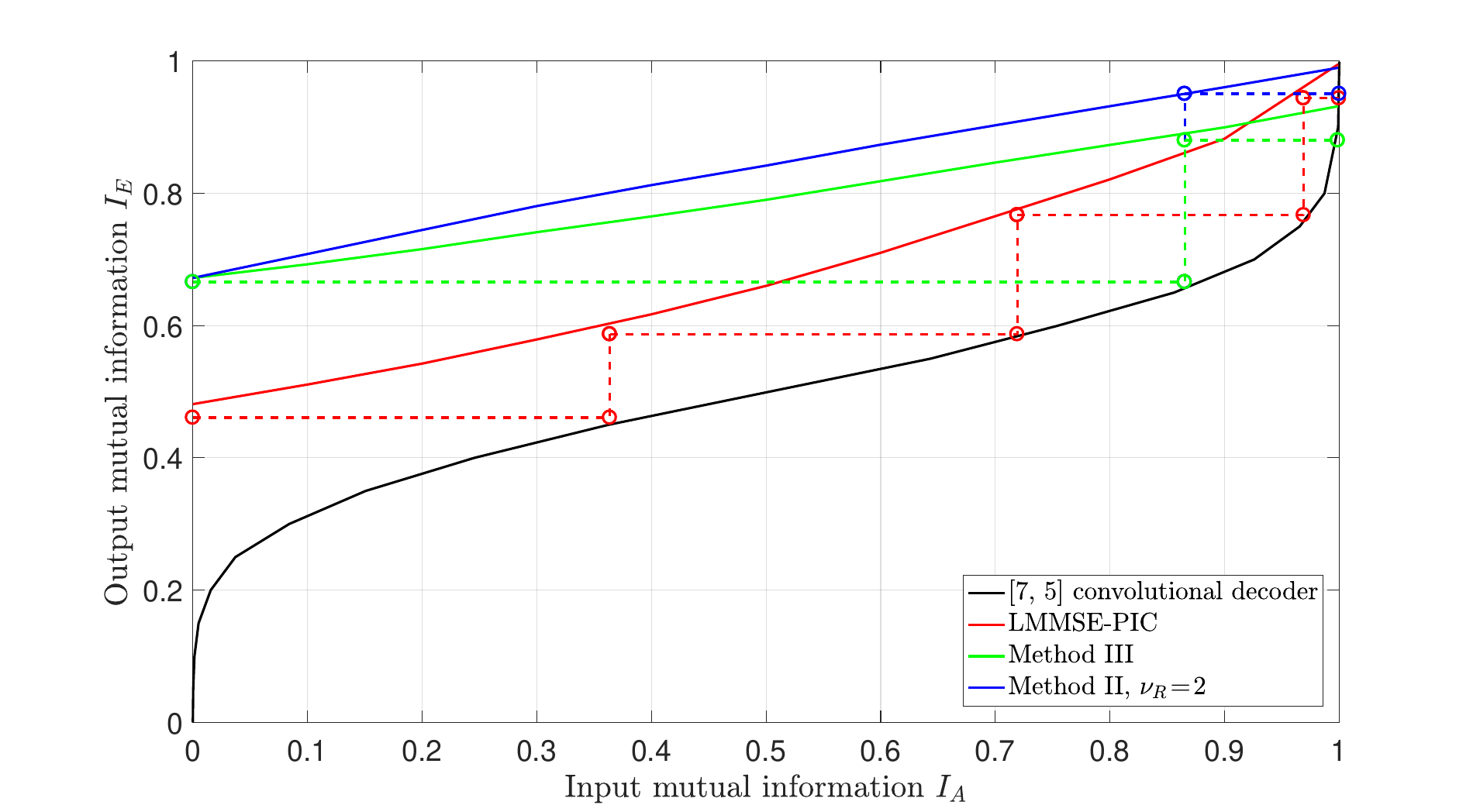}} 
\vspace*{-7mm}
\caption{\label{p1fig9} Iterative detection and decoding trajectories under Proakis-C channel at an SNR of 10 dB. The outer code is an [7, 5] convolutional code with generator polynomials $g_0(D)\!=\!1\!+\!D^2$ and $g_1(D)\!=\!1\!+\!D\!+\!D^2$. A random permutation of the code block is used and the black curve is the decoding EXIT chart. The dashed lines are the iterative detection and decoding trajectories for LMMSE-PIC, Method III and II, respectively.}
\vspace*{-10mm}
\end{center}
\end{figure}

In Fig.\,\ref{p1fig10}, we evaluate the BLER under Proakis-C channel with QPSK symbols and $\nu\!=\!2$ for all CS demodulators. A (1064, 1600) turbo code is used. Note that, at the first iteration when there is no soft information, Method II and III overlap with each other. With CS demodulators, the gap to the MAP demodulator is less than 1 dB, while the LMMSE-PIC has a gap to the MAP that is up to 10 dB. Moreover, Method II performs slightly better than Method I, and Method III is slightly inferior to both methods. However, Method III has the advantage of less computational complexity than the others since all parameters are in closed-forms.

In Fig.\,\ref{p1fig11}, we evaluate the BLER under $4\!\times\!6$ MIMO channels with QPSK symbols and $\nu\!=\!3$ for all CS demodulators. A (1064, 1800) turbo code is used. As $N\!<\!K$, the LMMSE-PIC fails \cite{FO} at the first iteration due to the lack of receive diversity. However, the CS demodulators with $\nu\!=\!3$ significantly improve the performance and with less than 1 dB gap at $10\%$ BLER to the MAP. CS demodulators with $\nu\!=\!1$ after three iterations is less than 2 dB away from the MAP. With less computational cost, Method III still performs close to Method II.

Finally we remark that, for the sake of complexity savings, both for finite linear vector channels and ISI channels, the parameters of CS demodulators do not need to be updated through all iterations. Once the feedback information quality is good enough and the parameter $\vec{P}$ or $\alpha$  are close to ideal, the CS parameters can be kept unchanged in successive iterations.

\begin{figure}[t]
\begin{center}
\vspace*{-6mm}
\hspace*{-2mm}
\scalebox{.45}{\includegraphics{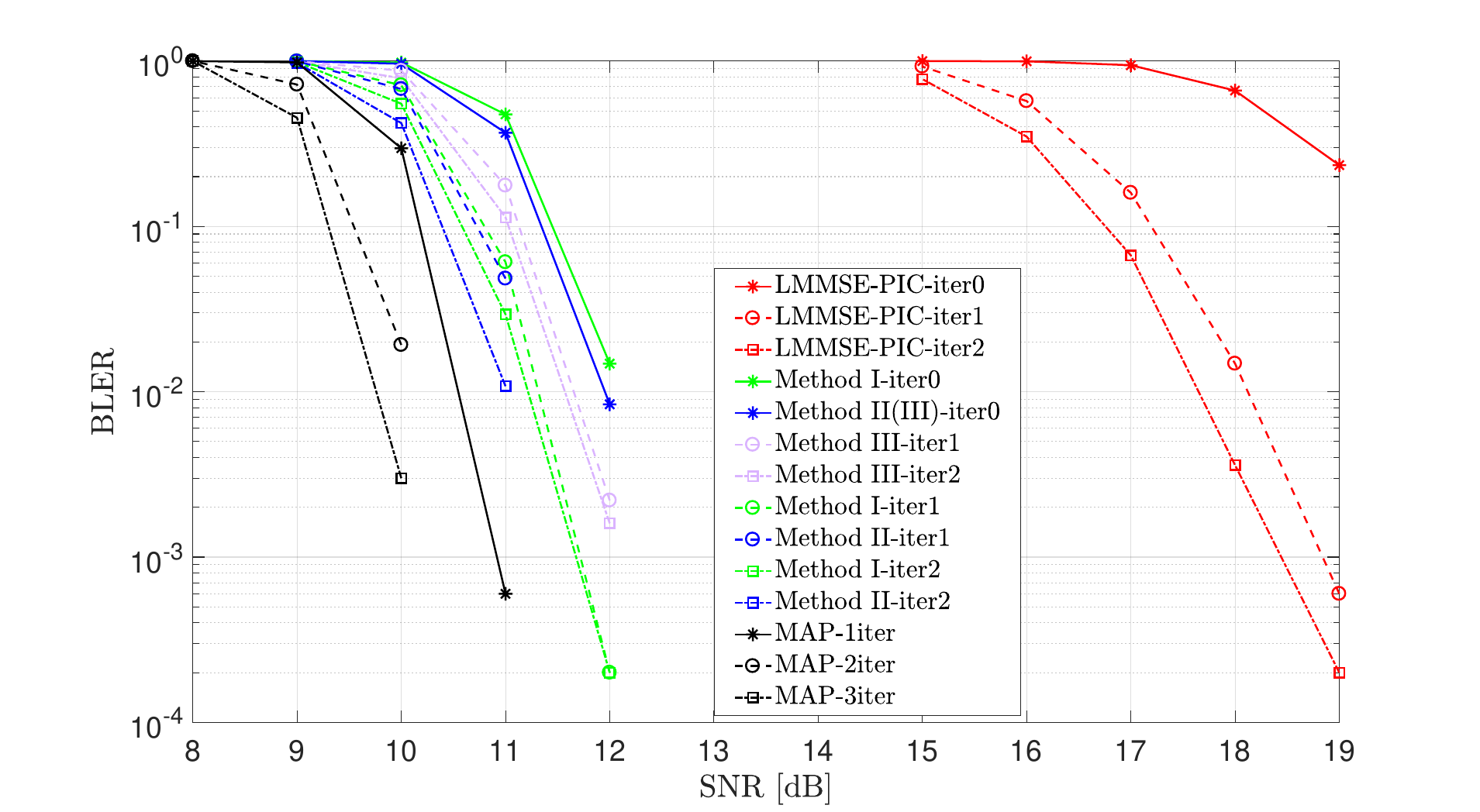}}   
\vspace*{-7mm}
\caption{\label{p1fig10} BLER performance of the LMMSE-PIC, Method I-III, and MAP under Proakis-C channel with QPSK modulation.}
\vspace*{-10mm}
\end{center}
\end{figure}

\begin{figure}
\begin{center}
\vspace*{-0mm}
\hspace*{-2mm}
\scalebox{.45}{\includegraphics{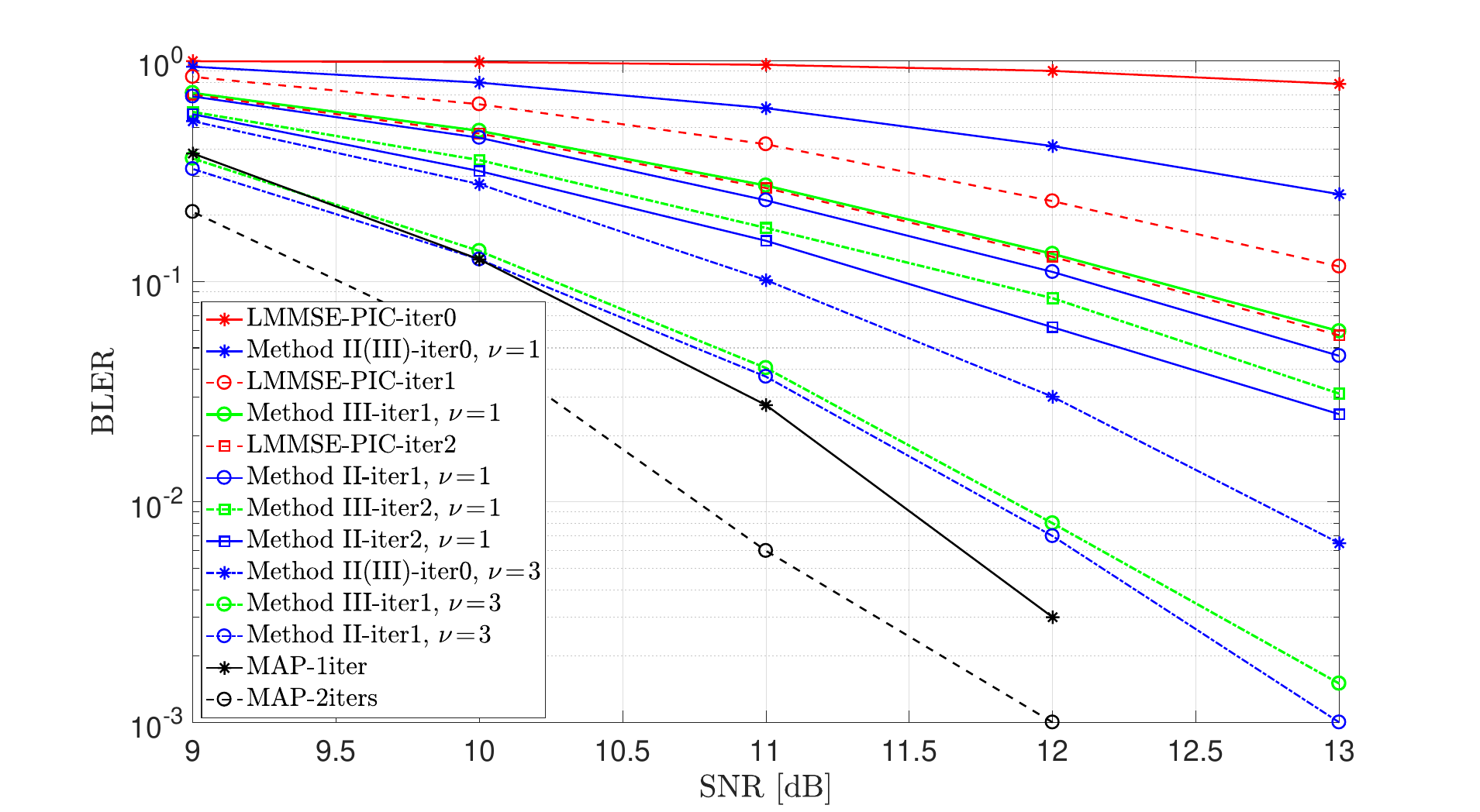}} 
\vspace*{-7mm}
\caption{\label{p1fig11} BLER performance of the LMMSE-PIC, Method II, Method III, and MAP under $4\!\times\!6$ MIMO channels with QPSK modulation.}
\vspace*{-10mm}
\end{center}
\end{figure}

\section{Summary}

In this paper we considered the design of channel shortening (CS) demodulators for linear vector channels that use a trellis representation of the received signal in combination with interference cancellation (IC) of the signal part that is not appropriately modeled by the trellis. In order to reach a trellis representation, a linear filter is applied as front-end. It is an extension of the well studied CS demodulators to iterative receivers and a generalization of the LMMSE-PIC demodulator to cooperate with trellis-search in turbo equalization.

We analyzed the properties of three different approaches for designing such optimal CS demodulators as all of them may come across as natural \lq\lq{}CS\rq\rq{} demodulators. In the used framework, there are three parameters that need to be optimized. Based on a generalized mutual information (GMI) cost function, two of these are solved for in closed-form, while the third needs to be numerically optimized except for Method III where we constructed it explicitly at the cost of a small performance loss. A simple gradient based optimization is used and turns out to perform well.

Numerical results are provided to illustrate the behavior of the proposed CS demodulators. In general, Method II based on the Ungerboeck model is superior to Method I that is based on the Forney model. Method II has the advantage over Method I that the optimization procedure is concave. The suboptimal Method III performs close to Method I and Method II, and it has all parameters in closed-forms. An interesting result is that the interference cancellation of the CS demodulators should not cancel the effective channel perfectly outside the memory size $\nu$, a property that cannot be seem in LMMSE-PIC demodulator as $\nu\!=\!0$. Moreover, we have also analyzed asymptotic properties of the CS demodulators and showed that, Method III converges to Method II asymptotically when the noise density goes to zero or infinity.

\section*{Appendix A: Derivation of the GMI }

By making the eigenvalue decomposition $\vec{Q}\vec{\Lambda}\vec{Q}\rmh \!=\! \vec{G}$ and letting $\vec{s}\!=\!\vec{Q}\rmh \vec{x}$. As $\vec{x}$ is assumed to be zero mean complex Gaussian random vector with covariance matrix $\vec{I}$, we can write $\tilde{p}(\vec{y}|\vec{x})$ in (\ref{p1md2}) as
\bea \tilde{p}(\vec{y}|\vec{x})=\exp\!\big(2\!\Re\!\big\{\vec{s}\rmh\vec{d}\big\}\!-\!\vec{s}\rmh\vec{\Lambda s}\big),\eea
where
$\vec{d}\!=\!\vec{Q}\rmh(\vec{Vy}\!-\!\vec{R}\hat{\vec{x}}).$
We can now evaluate
{\setlength\arraycolsep{2pt}  \bea \tilde{p}(\vec{y})&=&\int \!\tilde{p}(\vec{y}|\vec{x}) p(\vec{x})\mathrm{d}\vec{x} \notag \\
&=&\frac{1}{\pi^K}\int\!\exp\!\big(2\!\Re\!\big\{\vec{s}\rmh\vec{d}\big\}\!-\!\vec{s}\rmh\vec{\Lambda s}\big)\exp\!\big(\!\!-\!\!\vec{s}\rmh\vec{s})\mathrm{d}\vec{s} \notag \\
&=&\prod_{k=0}^{K-1} \frac{1}{1\!+\!\lambda_k}\exp\!\bigg(\frac{|d_k|^2}{1\!+\!\lambda_k}\bigg). \notag \eea}
\hspace{-1.4mm}where $\lambda_k$ is the $k$th diagonal element of $\vec{\Lambda}$ and $d_k$ is the $k$th entry of $\vec{d}$. Taking the expectation over $\vec{y}$ gives
\bea -\mathbb{E}_{p(\vec{y})}[\log(\tilde{p}(\vec{y}))]= \log\!\big(\!\det(\vec{I}\!+\!\vec{G})\big)\!-\!\mathrm{Tr}\big(\vec{L}(\vec{I}\!+\!\vec{G})^{-\!1}\big) \notag\eea
where the matrix $\vec{L}\!=\!\mathbb{E}_{p(\vec{y})}[\vec{Q}\vec{d}\vec{d}\rmh\vec{Q}\rmh]$ is given by
{\setlength\arraycolsep{2pt} \bea \vec{L}=\vec{V}(N_0\vec{I}\!+\!\vec{HH}\rmh)\vec{V}\rmh\!-\!\vec{V}\vec{H}\vec{P}\vec{R}\rmh\!-\!\vec{R}\vec{P}\vec{H}\rmh\vec{V}\rmh\!+\!\vec{R}\vec{P}\vec{R}\rmh.  \eea }
\hspace*{-1.4mm}On the other hand, we have
\bea -\mathbb{E}_{p(\vec{y},\vec{x})}\left[\log(\tilde{p}(\vec{y}|\vec{x}))\right]=\mathrm{Tr}(\vec{G})\!-\!2\!\Re\!\big\{\mathrm{Tr}(\vec{VH}\!-\!\vec{RP})\big\}.\notag\eea
Combining the two expectations, the GMI reads,
{\setlength\arraycolsep{0pt}\bea I_{\mathrm{GMI}}(\vec{V},\vec{R},\vec{G})&=&\log\!\big(\!\det(\vec{I}\!+\!\vec{G})\big)\!-\!\mathrm{Tr}\big(\vec{L}(\vec{I}\!+\!\vec{G})^{-\!1}\big)\!-\!\mathrm{Tr}(\vec{G})\!+\!2\!\Re\!\big\{\mathrm{Tr}(\vec{VH}\!-\! \vec{RP})\big\} \notag \\
&=&\log\big(\det(\vec{I}\!+\!\vec{G})\big)\!-\!\mathrm{Tr}(\vec{G})\!+\!2\Re\big\{\mathrm{Tr}(\vec{V}\vec{H}\!-\!\vec{R}\vec{P})\big\} \nonumber \\
&&-\mathrm{Tr}\big((\vec{I}\!+\!\vec{G})^{-\!1}\big(\vec{V}[\vec{H}\vec{H}\rmh\!+\!N_0\vec{I}]\vec{V}\rmh\!-\!2\Re\big\{\vec{V}\vec{H}\vec{P}\vec{R}\rmh\big\} \!+\!\vec{R}\vec{P}\vec{R}\rmh\big)\big)
.\quad  \notag\eea}

\section*{Appendix B: The Proof of Proposition \ref{p1prop1}}
As the formula of GMI in (\ref{p1metricMIMO1}) is quadratic in $\vec{W}$ and no constraints apply to $\vec{W}$, taking the gradient of $I_{\mathrm{GMI}}(\vec{W}, \vec{T},\vec{F} )$ with respect to $\vec{W}$ and setting it to zero, the optimal $\vec{W}$ is given in (\ref{p1mimo1optW}). Inserting $\vec{W}_{\mathrm{opt}}$ into (\ref{p1metricMIMO1}) gives, after some manipulations,
{\setlength\arraycolsep{2pt} \bea \label{p1gmiW} I_{\mathrm{GMI}} (\vec{W}_{\mathrm{opt}}, \vec{T}, \vec{F})&=&K\!+\!\log\!\big (\!\det(\vec{I}\!+\!\vec{F}\rmh\vec{F})\big)\!+\!\mathrm{Tr}\big(\vec{T}\rmh\vec{F}(\vec{I}\!+\!\vec{F}\rmh\vec{F})^{-\!1}\vec{F}\rmh\vec{T}\tilde{\vec{M}}\big)\nonumber \\
&&+\mathrm{Tr}\big(\vec{M}(\vec{I}\!+\!\vec{F}\rmh\vec{F})\big)\!+\!2\!\Re\!\big\{\mathrm{Tr}\big(\vec{PMF}\rmh\vec{T}\big)\!\big\}. \eea}
\hspace*{-1.4mm}where $\vec{M}$ and $\tilde{\vec{M}}$ are defined in (\ref{p1M}) and (\ref{p1Mt}). If $\vec{P}\!=\!\vec{0}$, (\ref{p1gmiW}) equals
 \bea I _1 (\vec{F})=K\!+\!\log\!\big (\!\det(\vec{I}\!+\!\vec{F}\rmh\vec{F})\big)\!+\!\mathrm{Tr}\big(\vec{M}(\vec{I}\!+\!\vec{F}\rmh\vec{F})\big). \notag \eea
In this case, there is no soft information available and the matrix $\vec{T}$ is not included in the formula. When $\vec{P}\!\neq\!\vec{0}$, the terms of $I_{\mathrm{GMI}}$ in (\ref{p1gmiW}) related to $\vec{T}$ are
\be f(\vec{T})\!=\!\mathrm{Tr}\big(\vec{T}\rmh\!\vec{F}(\vec{I}\!+\!\vec{F}\rmh\vec{F})^{-\!1}\vec{F}\rmh\vec{T}\tilde{\vec{M}}\big)\!+\!2\!\Re\!\big\{\mathrm{Tr}\big(\vec{PMF}\rmh\vec{T}\big)\!\big\}.\notag\ee
Let $\vec{t}_k$ denote the $k$th column of $\vec{T}$, but all elements in rows $[k,\min(k\!+\!\nu, K\!-\!1)]$ removed, and define the column vector $\vec{t}\!=\![\!\begin{array}{cccc}\vec{t}_0\rmt\; \vec{t}_1\rmt\; \ldots \; \vec{t}_{K\!-\!1}\rmt\end{array}\!\!]\rmt$, then by the definition of the indication matrix $\vec{\Omega}$, we have
 \bea  \vec{t}=\vec{\Omega}\mathrm{vec}(\vec{T}).\notag\eea

Similarly, let $\vec{z}_k$ denote the $k$th column of the matrix $\vec{FMP}$ but with all elements in rows $[k,\min(k\!+\!\nu, K\!-\!1)]$ removed, and define a row vector $\vec{z}\!=\![\!\begin{array}{cccc}\vec{z}_0\rmt \; \vec{z}_1\rmt \;\ldots \;\vec{z}_{K\!-\!1}\rmt\end{array}\!\!]\rmt$, then we have
\bea  \vec{z}\!=\!\Omega\mathrm{vec}(\vec{F\!M\!P})\!=\!\vec{\Omega}\big((\vec{P\!M}^{\ast})\!\otimes\!\vec{I}_K\big)\mathrm{vec}(\vec{F}).\notag\eea
Finally, defining a Hermitian matrix $\hat{\vec{B}}_1\!=\!\vec{\Omega}\big(\tilde{\vec{M}}^{\ast}\!\otimes\!\big(\vec{F}(\vec{I}\!+\!\vec{F}\rmh\vec{F})^{-\!1}\!\vec{F}\rmh\big)\big)\vec{\Omega}\rmt$, and with that we can rewrite $f(\vec{T})$ as $f(\vec{T}) \!=\!  \vec{t}\rmh\hat{\vec{B}}_1 \vec{t}\!+\!2\!\Re\{\vec{z}\rmh\vec{t}\}$. Taking the gradient with respect to $\vec{t}$ and setting it to zero yields,
\bea \label{p1mimo1optt_app} \vec{t}_{\mathrm{opt}} =-\hat{\vec{B}}_1^{-\!1}\vec{z}.\eea
Transferring $\vec{t}_{\mathrm{opt}}$ back into $\vec{T}_{\mathrm{opt}}$ given the optimal $\vec{T}$ in (\ref{p1mimo1optt_app}) and inserting this into $f(\vec{T})$ gives
\bea f(\vec{T}_{\mathrm{opt}})=-\vec{z}\rmh\hat{\vec{B}}_1^{-\!1}\vec{z}. \notag\eea
Thus, with the optimal $\vec{W}$ and $\vec{T}$, when $\vec{P}\!\neq\!\vec{0}$ the GMI equals
{\setlength\arraycolsep{2pt} \bea  &&I_{\mathrm{GMI}}(\vec{W}_{\mathrm{opt}},\vec{T}_{\mathrm{opt}},\vec{F})= K\!+\! \log\!\big(\!\det(\vec{I}\!+\!\vec{F}\rmh\vec{F})\big)\!+\!\mathrm{Tr}\big(\vec{M}(\vec{I}\!+\!\vec{F}\rmh\vec{F})\big) \nonumber \\
&&\qquad\qquad\qquad   -\mathrm{vec}(\vec{F})\rmh\vec{D}\rmh\Big(\vec{\Omega}\big(\tilde{\vec{M}}^{\ast}\!\otimes\!\big(\vec{F}(\vec{I}\!+\!\vec{F}\rmh\vec{F})^{-\!1}\!\vec{F}\rmh\big)\big)\vec{\Omega}\rmt\Big)^{\!-\!1}\!\vec{D}\mathrm{vec}(\vec{F}). \quad \nonumber\eea}
\hspace*{-1.4mm}where $\vec{D}\!=\!\vec{\Omega}\big((\vec{P\!M}^{\ast})\!\otimes\!\vec{I}_K\big)$.

\section*{Appendix C: Non-Concavity Examples of Method I}

\begin{figure}[t]
\begin{center}
\vspace*{-6mm}
\hspace*{-2mm}
\scalebox{.45}{\includegraphics{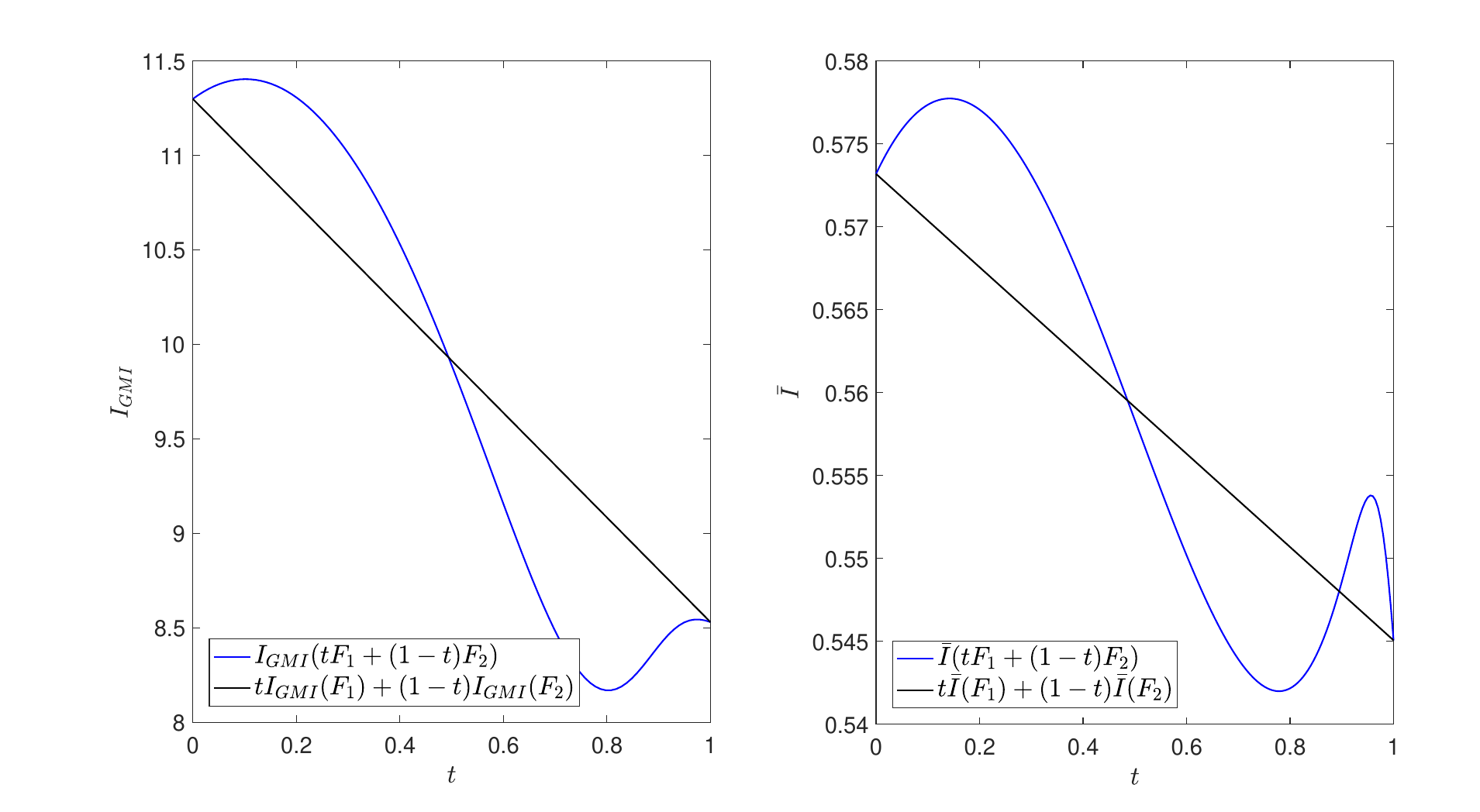}}  
\vspace*{-7mm}
\caption{\label{p1fig12} Non-concaveness of Method I under $5\!\times\!5$ MIMO channel (left figure) and Proakis-C ISI channel (right figure).}
\vspace*{-10mm}
\end{center}
\end{figure}

We give examples to demonstrate the non-concavity of Method I for MIMO and ISI channels with assuming that $\vec{P}\!=\!\vec{I}$ and $\alpha\!=\!1$, respectively. The memory size $\nu\!=\!1$ and the noise density $N_0$ equals 1 in both cases. A $5\!\times\!5$ MIMO channel and the Proakis-C channel are used.
\begin{example}
 MIMO case:
 {\setlength\arraycolsep{1.5pt}
 \bea  \vec{H}\!=\!\left[\!\begin{array}{ccccc}
2&   0&     -3&    5&     4\\
-5&     2&    -1&    0&     2\\
2&    -4&    3&    3&     3\\
-1&    -5&     -4&     1&    2\\
0&     -2&    0&    5&     5 \end{array}\! \right]\!\!,
\vec{F}_1\!=\!\left[\!\begin{array}{ccccc}
    4.94&   4.45&         0&         0&         0\\
         0&   0.21&    3.85&         0&         0\\
         0&         0&   5.56&    1.76&         0\\
         0&         0&         0&    0.61&    7.10\\
         0&         0&         0&         0&    2.79\end{array} \!\right]\!\!, 
          \vec{F}_2\!=\!\left[\!\begin{array}{ccccc}
    2.03&  6.17&         0&         0&         0\\
         0&    5.22&    3.56&         0&         0\\
         0&         0&    7.43&   0.73&         0\\
         0&         0&         0&    4.98&    4.32\\
         0&         0&         0&         0&    10.11\end{array} \!\right]\!\!.   \notag\eea}    
\end{example}

 \begin{example}
  ISI case:
{\setlength\arraycolsep{3pt}
 \bea  \vec{h}\!=\!\left[\begin{array}{ccccc}
0.227&   0.460&     0.688&    0.460&     0.227 \end{array} \right],
\vec{f}_1\!=\!\left[\begin{array}{ccccc}
   0.1606&     0.9009\end{array} \right], \vec{f}_2\!=\!\left[\begin{array}{ccccc}
    0.2230&    0.2035\end{array} \right]\!\! .\notag\eea}
 \end{example}
The $I_{\mathrm{GMI}}(\vec{W}_{\mathrm{opt}},\vec{T}_{\mathrm{opt}},\vec{F})$ given in (\ref{p1Iwt}) as a function $\vec{F}$ is plotted on the left in Fig.\,\ref{p1fig12}, while the $\bar{I}(W_{\mathrm{opt}}(\omega),T_{\mathrm{opt}}(\omega),F(\omega))$ given in (\ref{p1ibarm1optwt}) as a function of $F(\omega)$ is plotted on the right. If $I_{\mathrm{GMI}}(\vec{W}_{\mathrm{opt}},\vec{T}_{\mathrm{opt}},\vec{F})$ and $\bar{I}(W_{\mathrm{opt}}(\omega),T_{\mathrm{opt}}(\omega),F(\omega))$ are concave or convex, the blue curves lie above or below the black curves, which clearly does not hold in our examples.

\section*{Appendix D: The Gradient in Method I for Finite Linear Vector Channel}

In this section we derive the first order differential of the GMI given in (\ref{p1Iwt}) with respect to $\vec{F}$. In order to utilize the differential with respect to a matrix, we use the $\alpha$-differential as defined in \cite{JRM}. Assume a matrix $\vec{Y}_{N,K}$ with dimension $N\!\times\!K$ and a matrix $\vec{X}_{M, S}$ with dimension $M\!\times\!S$, define $d_{\vec{X}}\vec{Y}$ as the $\alpha$-differential of $\vec{Y}$ with respect to $\vec{X}$. Furthermore, defining $y_{\ell}$ and $x_{\ell}$ as $[\!\begin{array}{cccc}  y_1&y_2&\cdots&y_{N\!K}\end{array}\!]\!=\!\mathrm{vec}(\vec{Y})\rmt$ and  $[\!\begin{array}{cccc}  x_1&x_2&\cdots&x_{M\!S}\end{array}\!]\!=\!\mathrm{vec}(\vec{X})\rmt$, the $\alpha$-differential $d_{\vec{X}}\vec{Y}$ is
\be d_{\vec{X}}\vec{Y}=\frac{\partial\mathrm{vec}(\vec{Y})}{\partial\mathrm{vec}(\vec{X})\rmt}=\!\left[\begin{array}{cccc}
    \frac{\partial{y_1}}{\partial{x_1}}&  \frac{\partial{y_1}}{\partial{x_2}}&\cdots&\frac{\partial{y_1}}{\partial{x_{M\!S}}}\\
    \frac{\partial{y_2}}{\partial{x_1}}&  \frac{\partial{y_2}}{\partial{x_2}}&\cdots&\frac{\partial{y_2}}{\partial{x_{M\!S}}}\\
    \vdots&    \vdots&    \vdots&    \vdots\\
 \frac{\partial{y_{N\!K}}}{\partial{x_1}}&  \frac{\partial{y_{N\!K}}}{\partial{x_2}}&\cdots&\frac{\partial{y_{N\!K}}}{\partial{x_{M\!S}}} \end{array} \right]\!.  \quad\quad\notag\ee

The reason for adopting the $\alpha$-differential is because it keeps the chain rule and the product rule. We introduce an $NK\!\times\!NK$ permutation matrix $\vec{Z}_{N,K}$, which satisfies the condition $\mathrm{vec}(\vec{Y}\rmt)=\vec{Z}_{N,K}\mathrm{vec}(\vec{Y})$. It is easy to verify that $\vec{Z}_{N,K}^{-\!1}\!=\!\vec{Z}_{K,N}$, and when $N\!=\!1$ or $K\!=\!1$, $\vec{Y}$ is a vector and $\mathrm{vec}(\vec{Y}\rmt)\!=\!\mathrm{vec}(\vec{Y})$, hence $\vec{Z}_{N,1}\!=\!\vec{I}_{N}$ and $\vec{Z}_{1,K}\!=\!\vec{I}_{K}.$
Furthermore, by definition we have $d_{\vec{F}}(\vec{F})\!=\!d_{\vec{F}}(\mathrm{vec}(\vec{F}))\!=\!\vec{I}$, and $d_{\vec{F}}(\vec{F\rmh})\!=\! d_{\vec{F}}\left(\mathrm{vec}(\vec{F}\rmh)\right)\!=\!\vec{0}$. We start by reviewing a few properties \cite{JRM, PLF} of $\alpha$-differential below that will be used later, where both matrices $\vec{X}$ and $\vec{Y}$ are functions of $\vec{F}$ and the dimensions are specified by the subscripts associated to them:
 {\setlength\arraycolsep{1pt} \bea
 &&d_{\vec{F}}(\vec{X}_{K,K}^{-\!1})=-(\vec{X}_{K,K}\rmit\!\otimes\!\vec{X}_{K,K}^{-\!1})d_{\vec{F}}\vec{X}_{K,K}\notag \\
  &&d_{\vec{F}}(\vec{Y}_{N,K}\vec{X}_{K,S})=(\vec{X}_{K,S}\rmt\!\otimes\!\vec{I}_N)d_{\vec{F}}\vec{Y}_{N,K}+(\vec{I}_S\!\otimes\!\vec{Y}_{N,K})d_{\vec{F}}\vec{X}_{K,S}\notag \\
&&d_{\vec{F}}(\mathrm{log}(\mathrm{det}(\vec{X}_{K,K})))=\mathrm{vec}(\vec{X}_{K,K}\rmit)\rmt d_{\vec{F}}\vec{X}_{K,K}\notag \\
 &&d_{\vec{F}}(\vec{Y}_{N,K}\!\otimes\!\vec{X}_{M,S})=(\vec{I}_K\!\otimes\!\vec{Z}_{S,N}\!\otimes\!\vec{I}_M)(\vec{I}_{N\!K}\!\otimes\!\mathrm{vec}(\vec{X})d_{\vec{F}}\vec{Y}_{N,K} \notag \\
 &&\quad\quad\quad\quad\quad\quad\quad\quad\quad+(\vec{I}_K\!\otimes\!\vec{Z}_{S,N}\!\otimes\!\vec{I}_M)(\mathrm{vec}(\vec{Y})\!\otimes\!\vec{I}_{M\!S})d_{\vec{F}}\vec{X}_{M,S}. \notag\eea}
\hspace{-1mm}The $\alpha$-differential of $I_1(\vec{F})$ with respect to $\vec{F}$ is
 {\setlength\arraycolsep{2pt}\bea \label{p1di1} d_{\vec{F}}(I_1)&=&\mathrm{vec}((\vec{I}\!+\!\vec{F}\rmh\vec{F})\rmit)\rmt(\vec{I}_K\!\otimes\!\vec{F}\rmh)\!+\! \mathrm{vec}(\vec{F}^{\ast}\!\vec{M}\rmt)\rmt \notag \\
 &=&\mathrm{vec}(\vec{F\!M}\!+\!\vec{F}(\vec{I}\!+\!\vec{F}\rmh\vec{F})\rmih)\rmh. \eea}
\hspace*{-1.4mm}Defining a $K\!\times\!K$ matrix $\vec{B}\!=\!\vec{F}(\vec{I}\!+\!\vec{F}\rmh\vec{F})^{-\!1}\vec{F}\rmh$ and an $S\!\times\!S$ matrix $\vec{\Pi}\!=\!\big(\vec{\Omega}(\tilde{\vec{M}}^{\ast}\!\otimes\!\vec{B})\vec{\Omega}\rmt\big)^{\!-\!1}$, the $\alpha$-differential of $\delta_1(\vec{F})$ with respect to $\vec{F}$ is
 \bea d_{\vec{F}}(\delta_1)=
-\mathrm{vec}(\vec{F})\rmh\vec{D}\rmh\left((\mathrm{vec}(\vec{F})\rmt\vec{D}\rmt)\!\otimes\!\vec{I}_{S}\right)d_{\vec{F}}(\vec{\Pi})\!-\!\mathrm{vec }(\vec{F})\rmh\vec{D}\rmh\vec{\Pi}\vec{D},
 \eea
where
 {\setlength\arraycolsep{2pt}\bea
 d_{\vec{F}}(\vec{\Pi})&=&-(\vec{\Pi}\rmt\!\otimes\!\vec{\Pi})d_{\vec{F}}(\vec{\Omega}(\tilde{\vec{M}}^{\ast}\!\otimes\!\vec{B})\vec{\Omega}\rmt) \notag \\
%\!\!\!\!&=&\!\!\!\!-(\vec{\Pi}\rmt\!\otimes\!\vec{\Pi})(\vec{\Omega}\!\otimes\!\vec{\Omega})d_{\vec{F}}(\tilde{\vec{M}}^{\ast}\!\otimes\!\vec{B}) \notag \\
&=&-\big((\vec{\Pi}\rmt\vec{\Omega})\!\otimes\!(\vec{\Pi}\vec{\Omega})\big)(\vec{I}_K\!\otimes\!\vec{Z}_{K,K}\!\otimes\!\vec{I}_K)(\mathrm{vec}(\vec{\tilde{\vec{M}}^{\ast}})\!\otimes\!\vec{I}_{K^2})d_{\vec{F}}\vec{B}
 \eea}
\hspace*{-1.4mm}and
 {\setlength\arraycolsep{2pt} \bea
d_{\vec{F}}(\vec{B})&=&d_{\vec{F}}\big(\vec{I}\!-\!(\vec{I}\!+\!\vec{F}\vec{F}\rmh)^{-\!1}\big) \notag \\
%\!\!\!&=&\!\!\!\big((\vec{I}\!+\!\vec{F}\vec{F}\rmh)\rmit\big)\!\otimes\!\big((\vec{I}\!+\!\vec{F}\vec{F}\rmh)^{-\!1}\big)d_{\vec{F}}(\vec{I}\!+\!\vec{F}\vec{F}\rmh )\notag \\
&=&\big((\vec{I}\!+\!\vec{F}\vec{F}\rmh)\rmit\big)\!\otimes\!\big((\vec{I}\!+\!\vec{F}\vec{F}\rmh)^{-\!1}\big)(\vec{F}^{\ast}\!\otimes\!\vec{I}_K) \notag \\
&=&\big(\vec{F}^{\ast}(\vec{I}\!+\!\vec{F}\vec{F}\rmh)\rmit)\!\otimes\!(\vec{I}\!-\!\vec{B}).
 \eea}
\hspace*{-1.4mm}Then, defining a $K\!\times\!K$ matrix $\tilde{\vec{F}}\!=\!(\vec{I}\!+\!\vec{F}\rmh\vec{F})^{-\!1}\vec{F}\rmh$ and a $K^4\!\times\!K^2$ matrix
 \bea \label{p1psi}
 \vec{\Psi}=(\vec{I}_K\!\otimes\!\vec{Z}_{K,K}\!\otimes\!\vec{I}_K)\big(\mathrm{vec}(\vec{\tilde{\vec{M}}^{\ast}})\!\otimes\!\vec{I}_{K^2}\big),\eea
and by combing (\ref{p1di1})-(\ref{p1psi}), we finally have when $\vec{P}\!\neq\!\vec{0}$,
 \bea &&d_{\vec{F}}\big(I_{\mathrm{GMI}}(\vec{W}_{\mathrm{opt}},\vec{T}_{\mathrm{opt}},\vec{F})\big)\! \notag \\ &&= d_{\vec{F}}(I_1)\!+ \!d_{\vec{F}}(\delta_1) \notag\\
&&= \mathrm{vec}\big(\vec{F\!M}\!+\!\tilde{\vec{F}}\rmh\big)\rmh \!-\!\mathrm{vec }(\vec{F})\rmh\vec{D}\rmh\vec{\Pi}\vec{D} \notag \\
&&\quad+\mathrm{vec}(\vec{F})\rmh\vec{D}\rmh\big((\vec{\Omega}\rmt\vec{\Pi}\vec{D}\mathrm{vec}(\vec{F}))\rmt\!\otimes\!\big(\vec{\Pi}\vec{\Omega})\big)\!\vec{\Psi}\!\big(\tilde{\vec{F}}\rmt\!\otimes\!(\vec{I}\!-\!\vec{B})\big). \quad\quad\;\;  \notag\eea

\section*{Appendix E: The Proof of Proposition \ref{p1prop3} }

As the formula of GMI in (\ref{p1metricMIMO2}) is quadratic in $\vec{V}$ and no constraints apply to $\vec{V}$, taking the gradient of  $I_{\mathrm{GMI}}(\vec{V}, \vec{R},\vec{G} )$ with respect to $\vec{V}$ and setting it to zero, yields the optimal $\vec{V}$ given in (\ref{p1optgm2}). Inserting $\vec{V}_{\mathrm{opt}}$ into (\ref{p1metricMIMO2}) gives, after some manipulations
{\setlength\arraycolsep{2pt}\bea \label{p1gmiv} I_{\mathrm{GMI}}(\vec{V}_{\mathrm{opt}}, \vec{R},\vec{G} )&=&K\!+\!\log\!\big(\!\det(\vec{I}\!+\!\vec{G})\big)\!+\!2\Re\big\{\mathrm{Tr}(\vec{PMR})\big\}\nonumber \\
&&+\mathrm{Tr}\big(\vec{M}(\vec{I}\!+\!\vec{G})\big)
\!+\!\mathrm{Tr}\big((\vec{I}\!+\!\vec{G})^{-\!1}\!\vec{R}\tilde{\vec{M}}\vec{R}\rmh\big)
\eea}
\hspace*{-1.4mm}where $\vec{M}$ and $\tilde{\vec{M}}$ are defined in (\ref{p1M}) and (\ref{p1Mt}). If $\vec{P}\!=\!\vec{0}$, (\ref{p1gmiv}) equals
\bea  I_2(\vec{G} )=K\!+\!\log\!\big(\!\det(\vec{I}\!+\!\vec{G})\big)\!+\!\mathrm{Tr}\big(\vec{M}(\vec{I}\!+\!\vec{G})\big). \notag\eea
When $\vec{P}\!\neq\!\vec{0}$, the terms of $I_{\mathrm{GMI}}(\vec{V}_{\mathrm{opt}}, \vec{R},\vec{G} )$ in (\ref{p1gmiv}) related to $\vec{R}$ are
\be g(\vec{R})\!=2\Re\big\{\mathrm{Tr}(\vec{PMR})\big\}\!+\!\mathrm{Tr}\big((\vec{I}\!+\!\vec{G})^{-\!1}\vec{R}\tilde{\vec{M}}\vec{R}\rmh\big).\notag\ee
Let $\vec{r}_k$ denote the $k$th column of $\vec{R}$, but where all elements in rows $[\max(0,k\!-\!\nu_{\mathrm{R}}), \min(k\!+\!\nu_{\mathrm{R}},K\!-\!1)]$ are removed, and define the column vector $\vec{r}=[\!\begin{array}{cccc}\vec{r}_0\rmt\; \vec{r}_1\rmt\;\ldots \; \vec{r}_{K\!-\!1}\rmt \end{array}\!\!]\rmt$, then we have $\vec{r}\!=\!\vec{\Omega}\mathrm{vec}(\vec{R})$. Moreover, let $\vec{d}_k$ denote the $k$th column of the matrix $\vec{MP}$ but with all elements in rows $[\max(0,k\!-\!\nu_{\mathrm{R}}), \min(k\!+\!\nu_{\mathrm{R}},K\!-\!1)]$ are removed and define the vector $\vec{d}=[\!\!\begin{array}{cccc} \vec{d}_0\rmt \;\vec{d}_1\rmt \; \ldots \;\vec{d}_{K\!-\!1}\rmt \end{array}\!\!]\rmt$. From the definition of $\vec{d}$, we have $\vec{d}=\vec{\Omega}\mathrm{vec}(\vec{MP})$. Defining a Hermitian matrix $\hat{\vec{B}}_2$ as
\be \hat{\vec{B}}_2=\vec{\Omega}\big(\tilde{\vec{M}}^{\ast}\!\otimes\!(\vec{I}\!+\!\vec{G}
)^{-\!1}\big)\vec{\Omega}\rmt,\notag\ee
we can write $f(\vec{R})$ as $ g(\vec{R})\!=\!\vec{r}\rmh\hat{\vec{B}}_2\vec{r}\!+\!2\Re\{\vec{d}\rmh\vec{r}\}$. Therefore, the optimal $\vec{r}$ is
\be \vec{r}_{\mathrm{opt}}\!=\!-\hat{\vec{B}}_2^{-\!1}\vec{d}.\ee
Transferring $\vec{r}_{\mathrm{opt}}$ back into $\vec{R}_{\mathrm{opt}}$ gives the optimal $\vec{R}$ in  (\ref{p1mimo2optR}) and inserting this into $g(\vec{R})$ gives
\bea g(\vec{R}_{\mathrm{opt}})\!=\!-\vec{d}\rmh\hat{\vec{B}}_2^{-\!1}\vec{d}.\notag\eea
Thus, with the optimal $\vec{V}$ and $\vec{R}$, when $\vec{P}\!\neq\!\vec{0}$ the GMI equals
{\setlength\arraycolsep{2pt} \bea  I_{\mathrm{GMI}}(\vec{V}_{\mathrm{opt}}, \vec{R}_{\mathrm{opt}},\vec{G} )&=&K\!+\!\log\!\big(\!\det(\vec{I}\!+\!\vec{G})\big)\!+\!\mathrm{Tr}\big(\vec{M}(\vec{I}\!+\!\vec{G})\big) \notag \\ &&-\!\vec{d}\rmh\big(\vec{\Omega}\big(\tilde{\vec{M}}^{\ast}\!\otimes\! (\vec{I}\!+\!\vec{G})^{-\!1}\big)\vec{\Omega}\rmt\big)^{-\!1}\!\vec{d}. \nonumber\eea}

\section*{Appendix F: The Gradient in Method II for Finite Linear Vector Channel}

Now we calculate the $\alpha$-differential of $I_{\mathrm{GMI}}(\vec{V}_{\mathrm{opt}},\vec{R}_{\mathrm{opt}},\vec{G})$ given in (\ref{p1Ivr}) with respect to $\vec{G}$ when $\vec{P}\!\neq\!\vec{0}$. Taking the $\alpha$-differential of $I_2(\vec{G})$ with respect to $\vec{G}$ yields,
 \bea  \label{p1part1} d_{\vec{G}}(I_2)=\mathrm{vec}((\vec{I}\!+\!\vec{G})^{-\!1}\!+\!\vec{M})\rmh.\eea
Define an $S\!\times\!S$ Hermitian matrix $\vec{\Phi}\!=\!\big(\vec{\Omega}\big(\tilde{\vec{M}}^{\ast}\!\otimes\!(\vec{I}+\vec{G})^{\!-\!1}\big)\vec{\Omega}\rmt\big)^{-1}$ and taking the $\alpha$-differential of $\delta_2(\vec{G})$ with respect to $\vec{G}$ yields,
{\setlength\arraycolsep{2pt} \bea \label{p1part2} d_{\vec{G}}(\delta_2)&=&-(\vec{d}\rmt\!\otimes\!\vec{d}\rmh)d_{\vec{G}}(\vec{\Phi}) \notag \\
 &=&(\vec{d}\rmt\!\otimes\!\vec{d}\rmh)(\vec{\Phi}\rmt\!\otimes\!\vec{\Phi})d_{\vec{G}}\big(\vec{\Omega}\big(\tilde{\vec{M}}^{\ast}\!\otimes\!(\vec{I}+\vec{G})^{\!-\!1}\big)\vec{\Omega}\rmt\big) \notag \\
 &=&\big((\vec{d}\rmt\vec{\Phi}\rmt)\!\otimes\!(\vec{d}\rmh\vec{\Phi})\big)\!(\vec{\Omega}\!\otimes\!\vec{\Omega})d_{\vec{G}}\big(\tilde{\vec{M}}^{\ast}\!\otimes\!(\vec{I}\!+\!\vec{G})^{-\!1}\big) \notag \\
 &=&\big((\vec{d}\rmt\vec{\Phi}\rmt\vec{\Omega})\!\otimes\!(\vec{d}\rmh\vec{\Phi}\vec{\Omega})\big)\vec{\Psi}d_{\vec{G}}\big((\vec{I}\!+\!\vec{G})^{-\!1}\big) \notag \\
&=&-\big((\vec{d}\rmt\vec{\Phi}\rmt\vec{\Omega})\!\otimes\!(\vec{d}\rmh\vec{\Phi}\vec{\Omega})\big)\vec{\Psi}\big((\vec{I}\!+\!\vec{G})\rmit\!\otimes\!(\vec{I}\!+\!\vec{G})^{-\!1}\big) \eea}
\hspace{-1.4mm}where $\vec{\Psi}$ is defined in (\ref{p1psi}). Combining (\ref{p1part1}) and (\ref{p1part2}), we can obtain
 {\setlength\arraycolsep{2pt}\bea d_{\vec{G}}\big(I_{\mathrm{GMI}}(\vec{V}_{\mathrm{opt}},\vec{R}_{\mathrm{opt}},\vec{G})\big) &=&d_{\vec{G}}(I_2)\!+\! d_{\vec{G}}(\delta_2)  \notag \\
 &=&\mathrm{vec}\big((\vec{I}\!+\!\vec{G})^{-\!1}\!+\!\vec{M}\big)\rmh
 \notag \\
 &&\!\!\!-\big((\vec{d}\rmt\vec{\Phi}\rmt\vec{\Omega})\!\otimes\!(\vec{d}\rmh\vec{\Phi}\vec{\Omega})\big)\vec{\Psi}\!\big((\vec{I}\!+\!\vec{G})\rmit\!\otimes\!(\vec{I}\!+\!\vec{G})^{-\!1}\big).\quad\quad \notag\eea}

\section*{Appendix G: The Concavity Proof of Method II with Finite Linear Vector Channels}
When $\vec{P}\!=\!\vec{0}$, as $\log\!\big(\!\det(\vec{I}\!+\!\vec{G})\big)\!$ is concave \cite{SLConvex} and $\!\mathrm{Tr}\big(\vec{M}(\vec{I}\!+\!\vec{G})\big)\!$ is linear in $\vec{G}$, the function $I_2(\vec{G})$ in (\ref{p1Ivr1}) is concave with respect to $\vec{G}$ whenever $\vec{I}\!+\!\vec{G}$ is positive definite.

The concavity when $\vec{P}\!\neq\!\vec{0}$ can be deduced from the composition theorem in \cite[Chpater 3.6]{SLConvex}. For a positive definite matrix $\vec{X}$, $\vec{d}\rmh\vec{X}^{\!-\!1}\vec{d}$ is convex and non-increasing (with respect to the generalized inequality for positive definite Hermitian matrices, see \cite{SLConvex, D63})  for any column vector $\vec{d}$. Furthermore, since $\vec{I}\!+\!\vec{G}$ is positive definite, $(\vec{I}\!+\!\vec{G})^{-\!1}$ is convex. As $\tilde{\vec{M}}\!\prec\!\vec{0}$ $\vec{X}=\vec{\Omega}\big(\!\tilde{\vec{M}}^{\ast}\!\otimes\! (\vec{I}\!+\!\vec{G})^{-\!1}\big)\vec{\Omega}\rmt$ is concave in $\vec{G}$. By the composition theorem, $\vec{d}\rmh\big(\vec{\Omega}\big(\tilde{\vec{M}}^{\ast}\!\otimes\! [\vec{I}\!+\!\vec{G}]^{-\!1}\big)\vec{\Omega}\rmt\big)^{-\!1}\!\vec{d}$ is convex, and $\delta_2(\vec{G})$ is then concave. Therefore the function $I_{\mathrm{GMI}}(\vec{V}_{\mathrm{opt}}, \vec{R}_{\mathrm{opt}},\vec{G})$ in (\ref{p1Ivr}) is concave with respect to $\vec{G}$ whenever $\vec{I}\!+\!\vec{G}$ is positive definite.

\section*{Appendix H: The Proof of Proposition \ref{p1prop5}}

The Fourier series associated to the Toeplitz matrix $\vec{W}$ is
\bea \label{Ww} W(\omega)=\sum_{k=-\infty}^{\infty}w_k\exp\!\left(j k\omega \right), \notag\eea
and the differential of $\bar{I}(W(\omega),T(\omega),F(\omega))$ in (\ref{p1ibarm1}) with respect to $w_k$ (where $\omega$ is fixed) is
{\setlength\arraycolsep{2pt}  \bea \label{p1wwopt} \frac{\partial\bar{I}}{\partial w_k}&=&-\frac{1}{2\pi}\! \int_{-\pi}^{\pi}\frac{|F(\omega)|^{2}\big(N_0+|H(\omega)|^2\big)W^{\ast}(\omega)}{1+ |F(\omega)|^2}\exp\!\left(j k\omega \right)\mathrm{d}\omega\! \notag \\
&&+ \frac{1}{\pi}\!\int_{-\pi}^{\pi}\Big(F^{\ast}(\omega)H(\omega)+ \frac{\alpha |F(\omega)|^{2}H(\omega)T^{\ast}(\omega)}{1+ |F(\omega)|^2}\Big)\!\exp\!\big(j k\omega \big)\mathrm{d}\omega. \;\;\eea}
\hspace{-1.4mm}Since (\ref{p1wwopt}) should equal zero for all $k$, the optimal $W(\omega)$ is given in (\ref{p1optisiw}). Inserting $W_{\mathrm{opt}}(\omega)$ back into (\ref{p1ibarm1}) yields,
{\setlength\arraycolsep{2pt} \bea \label{p1ibarm1optw} &&\bar{I}(W_{\mathrm{opt}}(\omega),T(\omega),F(\omega))\!=\! 1\!+\!\frac{\alpha}{\pi}\!\int_{-\pi}^{\pi}\!\Re\!\big\{F^{\ast}(\omega)T(\omega)M(\omega) \big\}\mathrm{d}\omega \!\nonumber \\
&&\quad+\!\frac{1}{2\pi}\! \int_{-\pi}^{\pi}\!\!\Big(\!\log\!\big(1\!+\! |F(\omega)|^2\big) \!+\frac{\tilde{M}(\omega)|T(\omega)F(\omega)|^2}{1\!+\!|F(\omega)|^2}\!+\!M(\omega)\big(1\!+\!|F(\omega)|^2\big) \!\Big)\mathrm{d}\omega .\qquad \eea}
\hspace{-1.4mm}where $M(\omega)$ and $\tilde{M}(\omega)$ are defined in (\ref{p1mw}) and (\ref{p1mtw}). When $\alpha\!=\!0$, the GMI in (\ref{p1ibarm1optw}) equals (\ref{p1ibarm1optwt1}), and when $0\!<\!\alpha\!\leq\!1$, the terms related to $T(\omega)$ in (\ref{p1ibarm1optw}) are
\bea \label{p1isif2} f(T(\omega))=\frac{\alpha}{\pi}\!\int_{-\pi}^{\pi}\!\Re\!\big\{F^{\ast}(\omega)T(\omega)M(\omega) \big\}\mathrm{d}\omega \!+\!\frac{1}{2\pi}\! \int_{-\pi}^{\pi}\!\frac{\tilde{M}(\omega)|T(\omega)F(\omega)|^2}{1\!+\!|F(\omega)|^2}\mathrm{d}\omega. \quad\eea
As the elements of the main diagonal and the first $\nu$ lower diagonals of matrix $\vec{T}$ are constrained to zero, we define the vector $\tilde{\vec{t}}$ that specifies the Toeplitz matrix $\vec{T}$ as
\bea  \tilde{\vec{t}}\!=\![\!\begin{array}{cccccc}t_{-N_{\mathrm{T}}}\;\ldots \; t_{-1}\; t_{\nu+1}\;\ldots \; t_{N_{\mathrm{T}}}\end{array}\!], \notag\eea
and with $\phi(\omega)$ defined in (\ref{p1phi}), the Fourier series $T(\omega)$ with a finite tap length $N_{\mathrm{T}}$ is
\bea  \label{p1Tw} T(\omega)=\!\!\sum_{-N_{\mathrm{T}}\leq k \leq N_{\mathrm{T}}, k\notin[0,\nu]}\!\!t_k\exp\!\big(j k\omega \big)=\tilde{\vec{t}}\phi(\omega). \eea
Furthermore, with $\vec{\varepsilon}_1$ and $\vec{\varepsilon}_2$ defined in (\ref{p1vareps12}), (\ref{p1isif2}) can be rewritten as 
$$f(T(\omega))\!=\!\tilde{\vec{t}}\vec{\varepsilon}_2\tilde{\vec{t}}\rmh\!+\!2\!\Re\!\big\{\tilde{\vec{t}}\vec{\varepsilon}_1 \big\}.$$
Therefore, the optimal $\tilde{\vec{t}}$ is
\bea \label{twopt} \tilde{\vec{t}}_{\mathrm{opt}}=-\vec{\varepsilon}_1\rmh\vec{\varepsilon}_2^{-\!1}. \eea
Putting $\tilde{\vec{t}}_{\mathrm{opt}}$ back into (\ref{p1ibarm1optw})-(\ref{p1Tw}), the optimal $T(\omega)$ is given in (\ref{p1optisit}) and $\bar{I}(W(\omega),T(\omega),F(\omega))$ for the optimal $W(\omega)$ and $T(\omega)$ is given in (\ref{p1ibarm1optwt}).

\section*{Appendix I: The Proof of Proposition \ref{p1prop7}}

The Fourier series associated to the Toeplitz matrix $\vec{V}$ is $V(\omega)\!=\!\sum\limits_{k=-\infty}^{\infty}v_k\exp\!\left(j k\omega \right)$ and the differential of $\bar{I}(V(\omega),R(\omega),G(\omega))$ in (\ref{p1ibarm2}) with respect to $v_k$ (where $\omega$ is fixed) is
{\setlength\arraycolsep{2pt}\bea \label{p1Vwopt} \frac{\partial\bar{I}}{\partial v_k}&=&-\frac{1}{2\pi}\! \int_{-\pi}^{\pi}\!\frac{\big(N_0\!+\!|H(\omega)|^2\big)V^{\ast}(\omega)}{1\!+\! G(\omega)}\exp\!\left(j k\omega \right)\mathrm{d}\omega \notag \\
&&+ \frac{1}{\pi}\!\int_{-\pi}^{\pi}\!\!\Big(H(\omega)\!+\! \frac{\alpha H(\omega) R^{\ast}(\omega)}{1\!+\! G(\omega)}\Big)\!\exp\!\big(j k\omega \big)\mathrm{d}\omega. \eea}
\hspace{-1.4mm}Since (\ref{p1Vwopt}) shall equal zero for all $k$, the optimal $V(\omega)$ is given in (\ref{p1optisiv}). Putting $V_{\mathrm{opt}}(\omega)$ in (\ref{p1optisiv}) back into (\ref{p1ibarm2}) yields,
{\setlength\arraycolsep{2pt}\bea \label{p1ibarm2optv} \bar{I}(V_{\mathrm{opt}}(\omega),R(\omega),G(\omega))&=&1\!+\!\frac{\alpha}{\pi}\!\int_{-\pi}^{\pi}\!\Re\!\big\{M(\omega)R(\omega)\big\}\mathrm{d}\omega \!+\!\frac{1}{2\pi}\! \int_{-\pi}^{\pi}\!\Big(\! \log\! \big(1\!+\! G(\omega)\big) \nonumber \\
&&+\frac{\tilde{M}(\omega)|R(\omega)|^2}{1\!+\!G(\omega)}\!+\!M(\omega)\big(1\!+\!G(\omega)\big) \!\Big)\mathrm{d}\omega, \quad \eea}
\hspace{-1.4mm}where $M(\omega)$ and $\tilde{M}(\omega)$ are defined in (\ref{p1mw}) and (\ref{p1mtw}). When $\alpha\!=\!0$, the GMI in (\ref{p1ibarm2optv}) equals (\ref{p1ibarm2optvr1}), and when $0\!<\!\alpha\!\leq\!1$, the terms of $\bar{I}(V_{\mathrm{opt}}(\omega),R(\omega),G(\omega))$ related to $R(\omega)$ in (\ref{p1ibarm2optv}) are
\bea \label{p1isir2} g(R(\omega))\!=\!\frac{\alpha}{\pi}\!\int_{-\pi}^{\pi}\!\Re\!\big\{M(\omega)R(\omega) \big\}\mathrm{d}\omega\!+\!\frac{1}{2\pi} \int_{-\pi}^{\pi}\!\frac{\tilde{M}(\omega)|R(\omega)|^2}{1\!+\!G(\omega)}\mathrm{d}\omega. \eea
Define the vector $\tilde{\vec{r}}$ that specifies the Toeplitz matrix $\vec{R}$ as $$\tilde{\vec{r}}\!=\![\!\begin{array}{cccccc} r_{-N_{\mathrm{R}}}\;\!\ldots\! \; r_{-\nu_{\mathrm{R}}-1}\; r_{\nu_{\mathrm{R}}+1}\;\!\ldots\! \; r_{N_{\mathrm{R}}} \end{array}\!\!],$$and with $\psi(\omega)$ defined in (\ref{p1lpsi}), the Fourier series $R(\omega)$ with a finite tap length $N_{\mathrm{R}}$ is
\bea \label{Rw}  R(\omega)=\!\!\sum_{-N_{\mathrm{R}}\leq k \leq N_{\mathrm{R}}, k\notin[-\nu_{\mathrm{R}},\nu_{\mathrm{R}}]}\!\!r_k\exp\!\big(j k\omega \big)=\tilde{\vec{r}}\psi(\omega)\eea
where $2\nu_{\mathrm{R}}\!+\!1$ is the band size that $\vec{R}$ is constrained to zero. With $\vec{\zeta}_1$ and $\vec{\zeta}_2$ defined in (\ref{p1zeta12}), (\ref{p1isir2}) can be written as $g(R(\omega))\!=\!\tilde{\vec{r}}\vec{\zeta}_2\tilde{\vec{r}}\rmh\!+\!2\!\Re\left\{\tilde{\vec{r}}\vec{\zeta}_1 \right\}$. Therefore, the optimal $\tilde{\vec{r}}$ is
\bea  \tilde{\vec{r}}_{\mathrm{opt}}\!=\!-\vec{\zeta}_1\rmh\vec{\zeta}_2^{-\!1}.\eea
This shows that $\tilde{\vec{r}}_{\mathrm{opt}}$ has Hermitian symmetry as $G(\omega)$, $M(\omega)$ and $\tilde{M}(\omega)$ are all real valued, thus $R_{\mathrm{opt}}(\omega)$ is real. Putting $\tilde{\vec{r}}_{\mathrm{opt}}$ back into (\ref{p1ibarm2optv})-(\ref{Rw}), the optimal $R(\omega)$ is given in (\ref{p1optisir}) and $\bar{I}(V(\omega),R(\omega),G(\omega))$ for the optimal $V(\omega)$ and $R(\omega)$ is given in (\ref{p1ibarm2optvr}).

\section*{Appendix J: The Concavity Proof of Method II with ISI Channels}

To prove $\bar{I}(V_{\mathrm{opt}}(\omega),R_{\mathrm{opt}}(\omega),G(\omega))$ in (\ref{p1ibarm2optvr}) is concave with respect to $G(\omega)$, it is sufficient to prove that $\vec{\zeta}_1\rmh\vec{\zeta}_2^{-\!1}\vec{\zeta}_1$ is convex with respect to $G(\omega)$. For a positive definite matrix $\vec{\zeta}_2$ , $\vec{\zeta}_1\rmh\vec{\zeta}_2^{-\!1}\vec{\zeta}_1$ is convex and non-increasing (with respect to a generalized inequality for positive definite Hermitian matrices) in $G(\omega)$ for any vector $\vec{\zeta}_1$ and with arbitrary finite tap length $N_{\mathrm{R}}$. As matrix $\tilde{\vec{M}}$ is negative definite, $\vec{\zeta}_2$ in (\ref{p1zeta12}) is concave with respect to $G(\omega)$ under the constraint that $\vec{I}\!+\!\vec{G}$ is positive definite. Hence $\vec{\zeta}_1\rmh\vec{\zeta}_2^{-\!1}\vec{\zeta}_1$ is convex in $G(\omega)$ by the composition theorem \cite{SLConvex}.

\section*{Appendix K: The Proof of Lemma \ref{p1lem5}}

From Theorem \ref{p1thm2}, the optimal $\vec{G}$ in Method III satisfies $[(\vec{I}\!+\!\vec{G}_{\mathrm{opt}})^{-1}]_{\nu}\!=-\![\hat{\vec{M}}]_{\nu}$. Note that when $\vec{P}\!=\!\vec{0}$, Method III and Method II are equivalent as $\hat{\vec{M}}\!=\!\vec{M}$. Hence, in order to prove Lemma \ref{p1lem4}, it is sufficient to show that $[\hat{\vec{M}}]_{\nu}$ converges to $[\vec{M}]_{\nu}$ as $N_0\!\to\!0$ and $\infty$. When $\vec{P}\!\prec\!\vec{I}$, $\vec{C}_k$ in (\ref{p1matc}) is positive definite, and as $N_0\to 0$,
 {\setlength\arraycolsep{2pt} \bea \vec{H}\rmh(\vec{H}\vec{C}_k\vec{H}\rmh\!+\!N_0\vec{I})^{-1}\vec{H}&=&\vec{C}_k^{-1}(\vec{H}\rmh\vec{H}\!+\!N_0\vec{C}_k^{-1})^{-1}\vec{H}\rmh\vec{H} \notag \\
&=&\vec{C}_k^{-1}\big(\vec{I}-N_0\vec{C}_k^{-1}(\vec{H}\rmh\vec{H})^{-1}\big)+\mathcal{O}(N_0^2). \notag\eea}
\hspace{-1.4mm}Therefore with $\hat{\vec{W}}$ and $\hat{\vec{C}}$ defined in (\ref{p1wiec})-(\ref{p1covmatc}),
 {\setlength\arraycolsep{2pt}  \bea \label{p1limit5} \hat{\vec{W}}\vec{H}&=&\vec{I}-N_0(\vec{H}\rmh\vec{H})^{-1}+\mathcal{O}(N_0^2), \notag\\
\hat{\vec{C}}&=&[\hat{\vec{W}}\vec{H}]_{\backslash\nu}=-N_0[(\vec{H}\rmh\vec{H})^{-1}]_{\backslash\nu}+\mathcal{O}(N_0^2). \eea}
\hspace{-1mm}With (\ref{p1limit5}) and $\hat{\vec{M}}$ in (\ref{p1mpic}), it can be verified that $\lim\limits_{N_0\to 0}[\hat{\vec{M}}/N_0]_{\nu}\!=\!-[(\vec{H}\rmh\vec{H})^{-1}]_{\nu}$. On the other hand, when $N_0\!\to\! \infty$, from (\ref{p1wiec})-(\ref{p1mpic}) we have
 {\setlength\arraycolsep{2pt} \bea \label{p1limit6}  N_0\hat{\vec{W}}&=&\vec{H}\rmh(\vec{H}\vec{C}_k\vec{H}\rmh/N_0\!+\!\vec{I})^{-1}=\vec{H}\rmh+\mathcal{O}(1/N_0),\notag\\
 N_0\hat{\vec{C}}&=&[\hat{\vec{W}}\vec{H}]_{\backslash\nu}=[\vec{H}\rmh\vec{H}]_{\backslash\nu}+\mathcal{O}(1/N_0), \eea}
\hspace{-1.4mm}With (\ref{p1limit6}) and $\hat{\vec{M}}$ defined in (\ref{p1mpic}), it can be verified that $\lim\limits_{N_0\to\infty}[N_0(\vec{I}\!+\!\hat{\vec{M}})]_{\nu}=[\vec{H}\rmh\vec{H}]_{\nu}$. Hence, from (\ref{p1limitM}) $[\hat{\vec{M}}]_{\nu}$ converges to $[\vec{M}]_{\nu}$ as $N_0\!\to\!0$ and $\infty$, which completes the proof.

\end{document}